%% file: answer-predicates.tex
\documentclass{article}
\usepackage{kr}
\usepackage{misc}

\usepackage{times}
\usepackage{soul}
\usepackage{url}
\usepackage[hidelinks]{hyperref}
\usepackage[utf8]{inputenc}
\usepackage[small]{caption}
\usepackage{graphicx}
\usepackage{amsmath}
\usepackage{amsthm}
\usepackage{amssymb}
\usepackage{booktabs}
\usepackage{algorithm}
\usepackage{algorithmic}
\usepackage{enumitem}
\usepackage{bbm}
\usepackage{microtype}
\usepackage{thm-restate}
\usepackage{tikz}
\usepackage{tikz-3dplot}
\usepackage{tikzsymbols}
\usepackage{tikz-feynman}
\usetikzlibrary[decorations.pathmorphing]
\definecolor{grey}{RGB}{65, 65, 65}
\definecolor{niebieski}{RGB}{153, 255, 255}
\definecolor{fioletowy}{RGB}{255, 0, 127}
\definecolor{zgnily}{RGB}{140, 140, 0}

\newtheorem{theorem}{Theorem}
\newtheorem{lemma}{Lemma}

\newtheorem{example}{Example}
\newtheorem{definition}{Definition}

\newtheorem{observation}{Observation}

\newcommand{\bodyCQ}{\ensuremath{\mn{bodyCQ}}\xspace}
\newcommand{\ata}{\text{2ATA}\xspace}

\title{
  Conservative Extensions
  for Existential Rules}

\author{%
Jean Christoph Jung$^1$\and
Carsten Lutz$^2$\and
Jerzy Marcinkowski$^{3}$ \\
\affiliations
$^1$Department of Computer Science, University of Hildesheim, Germany\\
$^2$Department of Computer Science, University of Bremen, Germany\\
$^3$Institute of Computer Science, University of Wrocław, Poland\\
\emails
jungj@uni-hildesheim.de,
clu@uni-bremen.de,
jma@cs.uni.wroc.pl
}

\begin{document}

\maketitle

\begin{abstract}
  We study the problem to decide, given sets $T_1,T_2$ of
  tuple-generating dependencies (TGDs), also
  called existential rules, whether $T_2$ is a conservative extension
  of $T_1$. We consider two natural notions of conservative extension,
  one pertaining to answers to conjunctive queries over databases and
  one to homomorphisms between chased databases. Our main results are
  that these problems are undecidable for linear TGDs,
  undecidable for guarded TGDs even when $T_1$ is empty,
  and decidable for frontier-one TGDs. 
\end{abstract}

\section{Introduction}

Tuple-generating dependencies (TGDs) are an expressive constraint
language that emerged in database theory, where it has various important
applications~\cite{AbHV95}. In knowledge representation,
TGDs are used as an ontology language under the names of
existential rules and Datalog$^\pm$
\cite{BMRT11,DBLP:conf/lics/CaliGLMP10}.  For the purposes of this
paper, however, we stick with the name of `TGDs'. A major application of
TGDs in KR is
ontology-mediated querying where a database query is enriched with an
ontology, aiming to deliver more complete answers and to extend  the
vocabulary available for query formulation~\cite{DBLP:journals/tods/BienvenuCLW14,BiOr15,DBLP:conf/rweb/CalvaneseGLLPRR09}. The semantics
of ontology-mediated querying can be given in terms of homomorphisms
and the widely known chase procedure that makes explicit the logical
consequences of a set of TGDs.

As the use of unrestricted TGDs makes the evaluation of
ontology-mediated queries undecidable, various computationally more
well-behaved fragments have been identified. In this paper, we
consider linear TGDs, guarded TGDs, and frontier-one TGDs
\cite{DBLP:journals/ws/CaliGL12,BMRT11,CaGK13}. For all of these,
ontology-mediated query evaluation is decidable.  Deferring a formal
definition to Section~\ref{sect:prelim} of this paper, we remark that
guarded generalizes linear, and that frontier-one is orthogonal to both
linear and guarded. Moreover, linear TGDs generalize description
logics (DLs) of the DL-Lite family \cite{DBLP:journals/jair/ArtaleCKZ09}
while both guarded and frontier-one generalize DLs of
the \ELI family~\cite{DBLP:books/daglib/0041477}.

On top of bare-bones query evaluation, there are other natural
problems that are suggested by the framework of ontology-mediated
querying. Consider the following: given sets of TGDs $T_1$ and~$T_2$, a
database schema $\Sigma_D$, and a query schema $\Sigma_Q$, decide
whether $T_2$ is a \emph{$\Sigma_D,\Sigma_Q$-CQ-conservative
  extension} of $T_1$, that is, whether there is a
$\Sigma_D$-database~$D$, a conjunctive query (CQ) $q(\bar x)$ in
schema $\Sigma_Q$, and a tuple $\bar c$ that is an answer to $q$ on
$D$ given $T_1$, but not an answer to $q$ on $D$
given~$T_2$~\cite{DBLP:conf/rweb/BotoevaKLRWZ16}.  Note that this is a
very relevant problem. If, for instance, $T_2$ is a
$\Sigma_D,\Sigma_Q$-CQ-conservative extension of $T_1$ and vice versa,
then we can safely replace $T_1$ with $T_2$ in any application where
databases are formulated in schema $\Sigma_D$ and queries in schema
$\Sigma_Q$.  CQ-conservative extensions have been studied for various
DLs and are decidable for many members of the DL-Lite
and \ELI families
\cite{DBLP:conf/aaai/KonevKLSWZ11,DBLP:journals/jair/JungLMS20}. In
this paper, we address the naturally emerging question whether
decidability extends to the more general settings of linear, guarded,
and frontier-one TGDs.

Another natural problem related to CQ-conservative extensions is
\emph{$\Sigma_D,\Sigma_Q$-hom-conservative extension} which asks
whether for every $\Sigma_D$-database, there is a
$\Sigma_Q$-homomorphism\footnote{A homomorphism that disregards
  symbols outside of $\Sigma_Q$.} from the chase $\mn{chase}_{T_2}(D)$
of $D$ with $T_2$ to $\mn{chase}_{T_1}(D)$ that is the identity on all
constants in $D$. In fact, this problem corresponds to CQ-conservative
extensions with CQs that may be infinitary, and it is known that these
two problems do not coincide even in the case of DLs~\cite{DBLP:conf/rweb/BotoevaKLRWZ16}. We study hom-conservative
extensions along with CQ-conservative extensions. In addition, we
consider the variant of CQ/hom-conservative extensions where the set
of TGDs $T_1$ is required to be empty. We refer to this as
\emph{$\Sigma_D,\Sigma_Q$-CQ/hom-triviality}. Note that these are also
very natural problems as they ask whether the given set of TGDs $T_2$
says \emph{anything at all} about $\Sigma_D$-databases as far as
conjunctive queries and homomorphisms over schema $\Sigma_Q$ are
concerned. It is not difficult to show that
$\Sigma_D,\Sigma_Q$-CQ-triviality and $\Sigma_D,\Sigma_Q$-hom-triviality coincide even for
unrestricted TGDs, and thus we only speak of
\emph{$\Sigma_D,\Sigma_Q$-triviality}.

Our main results are as follows.
\begin{enumerate}

\item For linear TGDs, CQ- and hom-conservative extensions are
  undecidable, but triviality is decidable.

\item For guarded TGDs, triviality is undecidable.

\item For frontier-one TGDs, CQ- and hom-conservative extensions are decidable.
  
\end{enumerate}
We consider it remarkable that undecidability already appears for a
class as restricted as linear TGDs.  Regarding Point~1, we also
determine the exact complexity of triviality for linear TGDs as being
\PSpace-complete, and \coNP-complete when the arity of relation
symbols is bounded by a constant. Regarding Point~3, our algorithms
yield 3\ExpTime upper bounds while 2\ExpTime lower bounds can be
imported from the DL \ELI, 
a restricted fragment of frontier-one TGDs
\cite{DBLP:conf/ijcai/Gutierrez-Basulto18,DBLP:journals/jair/JungLMS20}.
The exact complexity remains open.

Our undecidability results are proved by reduction from a convergence
problem that concerns Conway functions \cite{conway}. In a database
theory context, such a technique has been used in
\cite{DBLP:conf/icalp/GogaczM14}. As the reader shall see, the
reductions take place in the setting of Pyramus and Thisbe (see
\cite{Ovid}), a mythological couple that could only communicate
through a crack in the wall and whose fate it was to never meet again
in person. Bring some popcorn. The decidability result for
hom-conservative extensions rests on the observation that whenever
there is a database that witnesses non-conservativity, then there is
such a database of bounded treewidth. This enables a decision
procedure based on alternating tree automata. The case of
CQ-conservative extensions is more intricate as it requires the use of
\emph{homomorphism limits}, that is, families of homomorphisms that
can only look $n$ steps `into the model', for any $n$. It is not clear
how the existence of homomorphism limits can directly be verified by
tree automata. Our solution generalizes the approach to
CQ-conservative extensions in \ELI pursued in
\cite{DBLP:journals/jair/JungLMS20}. In short, the idea is to push the
use of homomorphism limits to parts of the chase that are
$\Sigma_Q$-disconnected from the database and regular in shape, and to
then characterize homomorphism limits from/into such regular
(infinite) databases in terms of unbounded homomorphisms.

\smallskip 

Most proof details are deferred to the appendix.

\smallskip

{\bf Related Work.} We already mentioned the work on DLs from the
DL-Lite and \ELI families
\cite{DBLP:conf/aaai/KonevKLSWZ11,DBLP:journals/jair/JungLMS20}. For
description logics such as \ALC that support negation and disjunction,
CQ- and hom-conservative extensions are undecidable
\cite{DBLP:journals/ai/BotoevaLRWZ19}. A different kind of
conservative extension is obtained by replacing databases and query
answers with logical consequences formulated in the ontology language
\cite{DBLP:conf/kr/GhilardiLW06}. While such conservative extensions
are decidable in 
\ALC
\cite{DBLP:conf/kr/GhilardiLW06,DBLP:conf/ijcai/LutzWW07}, they are
undecidable in the guarded fragment and in the two-variable fragment
of first-order logic \cite{DBLP:conf/icalp/JungLM0W17}. For
existential rule languages, the difference between this version of
conservative extensions and CQ-conservative extensions tends to
be small (depending on the class of rules considered).

\section{Preliminaries}
\label{sect:prelim}

\newcommand{\dbsig}{\ensuremath{ {\Sigma_D}}\xspace}
\newcommand{\querysig}{\ensuremath{{\Sigma_Q}}\xspace}
\newcommand{\fullsig}{\ensuremath{\Sigma}\xspace}
\newcommand{\chase}{{\sf chase}}
\newcommand{\tp}{{\sf tp}}
\newcommand{\adom}{{\sf adom}}
\newcommand{\con}{{\sf con}}
\renewcommand{\lim}{{\sf lim}}
\newcommand{\src}{{\sf src}}
\newcommand{\TP}{{\sf TP}}
\newcommand{\var}{{\sf var}}
\newcommand{\nul}{{\sf null}}

\paragraph{Relational Databases.}

Fix countably infinite and disjoint sets of \emph{constants} \Cbf and
\Nbf. We refer to the constants in \Nbf as \emph{nulls}.  A {\em
schema} $\Sigma$ is a set of relation symbols~$R$ with associated
arity $\mn{ar}(R) \geq 1$. 
A {\em $\Sigma$-fact} is an expression of the form $R(\bar c)$ with $R
\in \Sigma$ and $\bar c$ is an $\mn{ar}(R)$-tuple of constants from
$\Cbf \cup \Nbf$. 
A {\em $\Sigma$-instance} is a possibly infinite set of
$\Sigma$-facts and a {\em $\Sigma$-database} is a finite
$\Sigma$-instance that uses only constants from \Cbf. We write
$\mn{adom}(I)$ for the set of constants used in instance $I$.  For an
instance $I$ and a schema $\Sigma$, $I|_\Sigma $ denotes the
restriction of $I$ to $\Sigma$, that is, the set of all facts in
$I$ that use a relation symbol from $\Sigma$.  We say that $I$ is
\emph{connected} (resp., \emph{$\Sigma$-connected}) if the
Gaifman graph of $I$ (resp., $I|_\Sigma$) is
connected and that $I$ is of \emph{finite degree} if the Gaifman graph
of $I$ has finite degree. 

For a schema $\Sigma$, a {\em $\Sigma$-homomorphism} from instance $I$
to instance $J$ is a function
$h : \mn{adom}(I) \rightarrow \mn{adom}(J)$ such that
$R(h(\bar c)) \in J$ for every $R(\bar c) \in I$ with $R \in
\Sigma$. We say that $h$ is \emph{database-preserving} if it is the
identity on all constants from \Cbf and write $I \rightarrow_\Sigma J$
if there exists a database-preserving $\Sigma$-homomorphism from
$I$ to $J$. 

\smallskip\noindent\textbf{Conjunctive Queries.}
A {\em conjunctive query} (CQ) $q(\bar x)$ over a schema $\Sigma$
takes
the form $q(\bar x) \leftarrow \phi(\bar x, \bar y)$
where $\bar x$ and $\bar y$ are tuples of variables, $\phi$ is a
set of \emph{relational atoms} $R_i(\bar x_i)$ with
$R_i \in \Sigma$ and $\bar x_i$ a tuple of variables of length
$\mn{ar}(R_i)$.  We refer to the variables in $\bar x$ as the
\emph{answer variables} of $q$.
The {\em arity} of $q$
is the number of its answer variables and $q$ is \emph{Boolean} if
it is of arity~0. 

 Every CQ $q(\bar x)$ gives rise to a database $D_q$, known as the
 {\em canonical database} of~$q$, by viewing variables as constants.
  A \emph{$\Sigma$-homomorphism}
$h$  from $q$ to an instance $I$ is a $\Sigma$-homomorphism from $D_q$ to
 $I$.
 A tuple $\bar c \in \mn{adom}(I)^{|\bar x|}$ is an {\em answer} to
 $q$ on $I$ if there is a homomorphism $h$ from $q$ to $I$ with
 $h(\bar x) = \bar c$.
The {\em evaluation of $q(\bar x)$ on $I$}, denoted $q(I)$, is the set of all answers to $q$ on~$I$.
%

For a CQ $q$, but also for any other syntactic object $q$, we use
$||q||$ to denote the number of symbols needed to write $q$ as
a word over a suitable alphabet.

\smallskip\noindent\textbf{TGDs.}  
A {\em tuple-generating dependency} (TGD) $\vartheta$ 
is a
first-order sentence $\forall \bar x \forall \bar y \,
\big(\phi(\bar x,\bar y) \rightarrow \exists \bar z \, \psi(\bar
x,\bar z)\big) $ such that $q_\phi(\bar x) \leftarrow 
\phi(\bar x,\bar y)$ and $q_\psi(\bar x) \leftarrow 
\psi(\bar x,\bar z)$ are CQs.  We call
$\phi$ and $\psi$ the {\em body} and {\em head} of $\vartheta$.  The body may
be the empty conjunction, that is, logical truth. 
The variables in $\bar x$ are the 
\emph{frontier variables}. 
For readability, we write $\vartheta$ as $\phi(\bar x,\bar y) \rightarrow
\exists \bar z \, \psi(\bar x,\bar z)$.  
An instance $I$ 
\emph{satisfies}~$\vartheta$, denoted $I \models \vartheta$, if $q_\phi(I) \subseteq
q_\psi(I)$. It \emph{satisfies} $T$
if $I \models \vartheta$ for each $\vartheta \in T$. We then also say
that $I$ is a \emph{model} of $T$.

A TGD $\vartheta$ is {\em frontier-one} if it has exactly one frontier
variable \cite{BMRT11}. 
It is {\em linear} if its body contains at most one atom. Clearly,
every linear TGD is guarded. The \emph{body width} of a set $T$ of
TGDs is the maximum
number of variables in a rule body of a TGD in $T$, and the \emph{head
  width} is defined accordingly.  

Throughout this paper, we are going to make use of the well-known
chase procedure for making explicit the consequences of a set of TGDs
\cite{JoKl84,FKMP05,CaGK13}.  The result of chasing a database $D$
with a set of TGDs $T$ is denoted with $\chase_T(D)$, details are
given in the appendix.  We only mention here that our chase is not
oblivious, that is, it does not apply a TGD if the consequences of the
application are already there.  As a result, $\chase_T(D)$ is of
finite degree, which shall be important for our proofs.

We also mention that if $T$ is a set of frontier-one TGDs, then for
any database~$D$ the instance $\chase_T(D)$ generated by the chase can
be obtained from $D$ by `glueing' a (potentially infinite) instance
onto each constant $c \in \mn{adom}(D)$. We denote this instance with
$\mn{chase}_T(D)|^\downarrow_c$. A precise definition is given in the
appendix.

Let $T$ be a set of TGDs, $q(\bar x)$ a CQ and $D$ a database.  A
tuple $\bar c \in \mn{adom}(D)^{|\bar x|}$ is an {\em answer} to $q$
on $D$ w.r.t.\ $T$, written $D,T \models q(\bar c)$, if $q(\bar c)$ is
a logical consequence of $D \cup T$ or, equivalently, if there is a
homomorphism $h$ from $q$ to $\chase_T(D)$ with $h(\bar x)=\bar
c$. The \emph{evaluation of $q$ on $D$ w.r.t.\ $T$}, denoted $q_T(D)$,
is the set of all answers to $q$ on $D$ w.r.t.~$T$.

\section{Conservative Extensions}
\label{sect:consext}

We introduce the notions of conservative extension that are studied
in this paper and the associated decision problems.
\begin{definition}
  \label{def:conservative}
  Let $T_1, T_2$ be sets of TGDs and let $\dbsig,\querysig$ be schemas
  called the \emph{data schema} and \emph{query schema}. Then
\begin{itemize}

\item \emph{$T_2$ is $\dbsig,\querysig$-hom-conservative over $T_1$},
  written $T_1 \models^{\text{hom}}_{\dbsig,\querysig} T_2$, if
  there is a database-preserving $\querysig$-homomorphism from
  $\chase_{T_2}(D)$ to $\chase_{T_1}(D)$ for all $\dbsig$-databases $D$;

\item 
\emph{$T_2$ is $\dbsig,\querysig$-CQ-conservative
  over $T_1$}, written
$T_1 \models^{\text{CQ}}_{\dbsig,\querysig} T_2$, if
$q_{T_2}(D) \subseteq q_{T_1}(D)$ for all $\dbsig$-databases $D$ and
all CQs $q$ over schema $\querysig$.

\item $T_1$ is \emph{$\dbsig,\querysig$-hom-trivial} if $T_1$ is
$\dbsig,\querysig$-hom-conservative over the empty set of
TGDs, and likewise for \emph{$\dbsig,\querysig$-CQ-triviality}.

\end{itemize}
\end{definition}
%
It is easy to see that logical entailment $T_1 \models T_2$
implies $T_1 \models^{\text{hom}}_{\dbsig,\querysig} T_2$ for all
schemas $\dbsig$ and $\querysig$, and that
$\dbsig,\querysig$-hom-conservativity implies
$\dbsig,\querysig$-CQ-conservativity. The following example from
\cite{DBLP:conf/rweb/BotoevaKLRWZ16} shows that the converse fails.
\begin{example}
  \label{ex:basic}
  Consider the following sets of TGDs that are both linear and
  frontier-one:
  $$
  \begin{array}{r@{\,}c@{\;}c@{\;}l}
    T_1 \ = \ \{ & A(x) &\rightarrow& \exists y \, S(x,y), B(y), \\[1mm]
     & B(x) &\rightarrow& \exists y \, R(x,y), B(y) \ \} \\[3mm]
    T_2 \ = \ \{ & A(x) &\rightarrow& \exists y \, S(x,y), B(y),\\[1mm]
     & B(x) &\rightarrow& \exists y \, R(y,x), B(y) \ \}.
  \end{array}
  $$
  Let $\Sigma_D=\{A\}$ and $\Sigma_Q=\{R\}$.  We recommend to the
  reader to verify that $T_2$ is not
  $\Sigma_D,\Sigma_Q$-hom-conservative over $T_1$, but that it is
  $\Sigma_D,\Sigma_Q$-CQ-conservative.
\end{example}
However, $\dbsig,\querysig$-hom-conservativity is equivalent to
$\dbsig,\querysig$-CQ-conservativity with infinitary CQs.
We refrain from making this precise and instead consider the converse,
that is, $\dbsig,\querysig$-CQ-conservativity is equivalent to
$\dbsig,\querysig$-hom-conservativity when the latter is defined
in terms of a finitary version of homomorphisms that we introduce next.

\smallskip

Let $I_1,I_2$ be instances and $n \geq 0$, and let $\Sigma$ be a
schema. We write $I_1 \rightarrow^n_\Sigma I_2$ if for every induced
subinstance $I$ of $I_1$ with $|\adom(I)| \leq n$, there is a
database-preserving $\Sigma$-homomorphism from $I$ to $I_2$. We
further write $I_1\rightarrow^{\lim}_\Sigma I_2$ if $I_1
\rightarrow^n_\Sigma I_2$ for all $n \geq 1$.
\begin{theorem} \label{thm:firstchar} 
  Let $T_1$ and $T_2$ be sets of
  TGDs and $\dbsig,\querysig$ schemas.  Then $T_1
  \models^{\text{CQ}}_{\dbsig,\querysig} T_2$ iff $\chase_{T_2}(D)
  \rightarrow^{\lim}_{\querysig} \chase_{T_1}(D)$.  
\end{theorem}
For triviality, the above subtleties do not arise.
\begin{lemma} \label{lem:trivhomcqcoincide} Let $T_1,T_2$ be sets of
  TGDs and $\Sigma_D, \Sigma_Q$
  schemas. Then $T_1$ and $T_2$ are $\Sigma_D,\Sigma_Q$-hom-trivial if
  and only if they are
  $\Sigma_D,\Sigma_Q$-CQ-trivial.  \end{lemma}
Because of Lemma~\ref{lem:trivhomcqcoincide}, we from now on disregard
$\Sigma_D,\Sigma_Q$-CQ-triviality and refer to
$\Sigma_D,\Sigma_Q$-hom-triviality simply as
\emph{$\Sigma_D,\Sigma_Q$-triviality}.
Lemma~\ref{lem:trivhomcqcoincide} is an immediate consequence of
Theorem~\ref{thm:firstchar} and the following observation.
\begin{restatable}{lemma}{lemskippinghomsgeneral}
  \label{lem:skippinghomsgeneral} 
  Let $I_1,I_2$ be instances such that $I_1$ is countable and
  $I_2$ is finite, and let $\Sigma$ be a schema. If \mbox{$I_1
    \rightarrow^{\lim}_\Sigma I_2$}, then $I_1 \rightarrow_\Sigma
    I_2$.  
\end{restatable}
We sketch the proof of Lemma~\ref{lem:skippinghomsgeneral}, details
are in the appendix. If $I_1 \rightarrow^{\lim}_\Sigma I_2$, then we
find database-preserving $\Sigma$-homomorphisms $h_1,h_2,\dots$ from
finite subinstances \mbox{$J_1 \subseteq J_2 \subseteq \ldots$} of
$I_1$ to $I_2$ such that $I_1 = \bigcup_{i \geq 1} J_i$. If
$h_1,h_2,\dots$ are compatible in the sense that $h_i(c)=h_j(c)$
whenever $h_i(c),h_j(c)$ are both defined, then
$\bigcup_{i \geq 1} h_i$ is a $\Sigma$-homomorphism that witnesses
$I_1 \rightarrow_\Sigma I_2$.  If this is not the case, however, we
can still manipulate $h_1,h_2,\dots$ into a compatible sequence
$g_1,g_2,\dots$ by a technique that we call `skipping homomorphisms'
and which is used in several proofs in this paper.  We start with
$h_1$ and observe that since $J_1$ and $I_2$ are finite, there are
only finitely many homomorphisms $h$ from $J_1$ to $I_2$. Some such
homomorphism must occur infinitely often in the restrictions of
$h_1,h_2,\dots$ to $\mn{adom}(J_1)$ and thus we find a subsequence
$h'_1,h'_2,\dots$ of $h_1,h_2,\dots$ in which $h'_1$ is compatible
with all of $h'_2,h'_3,\dots$. We proceed in the same way for $h'_2$,
then for $h'_3$, ad infinitum, finding the desired sequence
$g_1,g_2,\dots$.

\smallskip

We consider the three decision problems \emph{hom-conservativity},
\emph{CQ-conservativity}, and \emph{triviality}, defined in the
obvious way. For instance, hom-conservativity means to decide, given
finite sets of TGDs $T_1$, $T_2$ and finite schemas
$\Sigma_D,\Sigma_Q$, whether $T_2$ is
$\Sigma_D,\Sigma_Q$-hom-conservative over $T_1$.



\section{Undecidability}
\label{sect:undec}

The aim of this section is to prove the following results.
\begin{theorem}\label{thm:undecidable} The following problems
  are undecidable:
  \begin{enumerate}

    \item hom-conservativity for linear TGDs; 

    \item CQ-conservativity for linear TGDs; 

    \item triviality for guarded TGDs.

  \end{enumerate}
\end{theorem}
We give a single proof that establishes Points~1 and~2. To attain
Point~3, a non-trivial modification of that proof is necessary. We start with the
proof of Points~1 and~2, first highlighting the main mechanism
that we use in our reduction, and then spelling out more details
of the reduction itself.


\newcommand{\eop}{\hspace*{\fill}$\square$}
\newcommand{\pp}{+_\gamma}
\newcommand{\mm}{-_\gamma}
\newcommand{\WH}{\mn{WH}}
\newcommand{\Tsibe}{\mn{Thisbe}}
\newcommand{\rulesig}{\Sigma_T}
\newcommand{\ccc}{{\ensuremath{\mathcal C}}\xspace}








\subsection{The Main Mechanism}
\label{sec:engine}


Consider the set of rules $T_\mn{myth}$. 
It comprises three TGDs:\smallskip
$$
\begin{array}{rcl}
  \mn{Encounter}(p,t) &\rightarrow& \exists p',c,t' \; M(p,p',c,t',t)
  \\[1mm]
  M(p,p',c,t',t) & \rightarrow & \exists p'',c',t'' \;
  M(p',p'',c',t'',t') \\[1mm]
M(p,p',c,t',t) &\rightarrow& \mn{Pyramus}(p,p'),
                             \mn{Thisbe}(t,t'),
  \\[1mm]
  && \mn{Channel}(c,p'), \mn{Channel}(c,t').
\end{array}
$$ 

\begin{figure}
\include{fig1}
\vspace{-.6cm}
\caption{Chase generated by $T_\mn{myth}$.}
\label{figA}
\end{figure}

Now consider the database $D=\{\mn{Encounter}(c_0,c'_0)\}$. The
instance $\chase_{T_\mn{myth}}(D)$, shown in Figure~\ref{figA}, will
play an important role. Its intuitive meaning is that `{\em after an
  initial brief encounter, Pyramus and Thisbe are not known to have
  ever met again, but remained forever able to connect via an
(indirect) channel}.' Notice that we do not explicitly 
show relation $M$ in
Figure \ref{figA} as $M$ is only a construction aid, needed to ensure
that the TGDs in $T_\mn{myth}$ are linear. 
As $\querysig$, we will use the set of relations symbols in
$T_\mn{myth}$ except $M$, plus a unary relation symbol
$\mn{Mouth}$. We advise the reader to not worry about the schema
$\dbsig$ at this point (it will actually be empty).

Let $\kappa=\langle [p_1,\ldots p_n], [t_1,\ldots,t_n]\rangle $ be a
pair of sequences of positive integers of the same length $n$. By
$\mn{River}_\kappa$, we mean the database that contains the following
facts, an example 
being displayed in
Figure~\ref{fig:river}:
\begin{itemize}

  \item There are 3 kinds of constants. The {\em eternities} are
    $e_1$ and~$e_2$. The \emph{channel} is $c$, not shown in the
    picture. All remaining constants are {\em worldly}.

  \item  For $1\leq i \leq n$ there is a $\mn{Pyramus}$ path of
    length $p_i$
    from $b_i$ to $b_{i-1}$ as well as a \mn{Thisbe} path of length $t_i$
    from $b_i$ to $b_{i-1}$.  Constants $b_i$ are called {\em
    bridges}. 

\item There is $\Tsibe(a,e_1)$ for each non-bridge constant $a$ on
  each of the $\mn{Thisbe}$-paths and there is $\mn{Pyramus}(a,e_2)$
  for each non-bridge constant $a$ on each of the
  $\mn{Pyramus}$-paths. There are also $\Tsibe(a,e_2)$ and
  $\mn{Pyramus}(a,e_1)$ for each bridge constant~$a$. In addition (and
  not in Figure~\ref{fig:river}), there are
  $\mn{Pyramus}(e_i,e_i)$ and $\mn{Thisbe}(e_i,e_i)$ for $i \in \{1,2\}$.

  \item  For each {\em worldly} constant $a$, there is
    $\mn{Channel}(c,a)$. Moreover, there are facts
    $\mn{Channel}(e_i,e_i)$, for $i\in \{1,2\}$. These
    facts are not shown in Figure~\ref{fig:river}. 

  \item  There are $\mn{Encounter}(b_n,b_{n-1})$ and
    $\mn{Mouth}(b_0)$. 

\end{itemize}
\begin{figure}
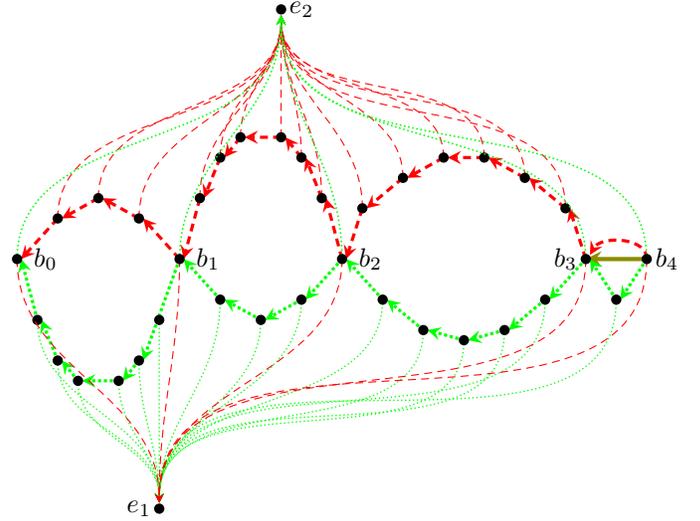

  \include{fig2}
  \vspace{-.6cm}
  \caption{The database $\mn{River}_{\kappa}$ for $\kappa=\langle 
  [4,7,7,1], [7,4,6,2]\rangle$.
}
\label{fig:river}
\end{figure}
It is easy to see that $\chase_{T_\mn{myth}}(\mn{River}_{\kappa})$ is
obtained from $\mn{River}_{\kappa}$ by adding a copy of the instance
shown in Figure~\ref{figA}, glueing the \mn{Encounter} fact to
the \mn{Encounter} fact in $\mn{River}_{\kappa}$ (and adding some
$M$-facts that are not important here).
%
%
Now, let us leave to our readers the pleasure to notice that:
\begin{observation}\label{see}
  There is a database-preserving $\querysig$-homomorphism from
  $\chase_{T_\mn{myth}}(\mn{River}_{\kappa})$ to $\mn{River}_{\kappa}$ if and only if 
  there exists $1\leq m <n$ such that $t_m\neq p_{m+1}$.
\end{observation}

\noindent
{\em Hint:} As long as Pyramus and Thisbe walk down their respective
river banks they are connected via the constant $c$. But for their union to last forever they need, at some point, to enter one of the eternities. Since 
eternity has no channel with the worldly constants (and the two eternities are not connected by a channel either),
Pyramus and Thisbe both need to enter the same eternity, and they need to do it 
simultaneously. But this can only happen when one of them is in a
bridge constant and the other one in a non-bridge one. 
\eop

That's nice, isn't it? But where could any undecidability be lurking here?

\subsection{Conway Functions}\label{sec:source}

Let
$\gamma,\alpha_0, \beta_0, \ldots, \alpha_{\gamma-1},
\beta_{\gamma-1}$ be positive integers such that $\beta_k|\gamma$ and
$\beta_k|k\alpha_k$ for $0\leq k< \gamma$.  For a positive integer
$n$, define $\digamma(n)$ by setting $\digamma(n) = n\alpha_k/\beta_k$
for $k= n\ \mn{mod} \ \gamma$. Thus, the remainder of $n$ when
dividing by $\gamma$ determines the pair $(\alpha_k,\beta_k)$ used to
compute the value $\digamma(n)$.  Note that due to the two
divisibility conditions, the range of $\digamma$ contains only
positive integers.

The function $\digamma$ is called the \emph{Conway function defined by}
 $\gamma,\alpha_0, \beta_0, \ldots, \alpha_{\gamma-1},
 \beta_{\gamma-1}$. We say that $\digamma$ \emph{stops} if there
 exists an $n\in \mathbb N$
 such that $\digamma^n(2) = 1$, where $\digamma^n$
 is $\digamma$ composed with itself, $n$ times.
 There is no special meaning to the numbers 1 and 2 used here,
we could also choose otherwise.
The following is well-known, see also \cite{DBLP:conf/icalp/GogaczM14}.
 
 \begin{theorem}\label{conway}
   It is undecidable whether the Conway function defined by a given sequence
   $\gamma,\alpha_0, \beta_0, \ldots,$
   $\alpha_{\gamma-1}, \beta_{\gamma-1}$ stops.
 \end{theorem}

 Take a Conway function $\digamma$ defined by a sequence
 $\gamma,\alpha_0, \beta_0, \ldots, \alpha_{\gamma-1},
 \beta_{\gamma-1}$.  We may assume without loss of generality that
 $\digamma(2)=3$ and $\digamma(1)=1$. Point~1 of
 Theorem~\ref{thm:undecidable} will be proven once we construct 
 sets $T_1$ and $T_2$ of linear TGDs along with schemas
 $\dbsig,\querysig$ such that 
\begin{description}
  \item[$(\heartsuit.l)$]
    $\digamma$ does not stop if and only if $T_2$ is  $\dbsig,\querysig$-hom-conservative over $T_1$.
\end{description}
It will then be easy to observe that $T_2$ is
$\dbsig,\querysig$-hom-conservative over $T_1$ iff $T_2$ is
$\dbsig,\querysig$-CQ-conservative over $T_1$, which also yields
Point~2 of Theorem~\ref{thm:undecidable}.  For Point~3, it will be
enough to construct a set $T$ of guarded TGDs, and schemas
$\dbsig,\querysig$,
such that
\begin{description}
  \item[$(\heartsuit.g)$]
    $\digamma$ does not stop if and only if $T$ is  $\dbsig,\querysig$-trivial.
\end{description}

\subsection{The Reduction}\label{sec:proof-long}

   We say that  $\kappa=\langle [p_1,\ldots, p_n], [t_1,\ldots,t_n]\rangle $
(or $\mn{River}_\kappa$)
 is
 \begin{itemize}
 \item 
 {\em locally correct} if the following conditions hold:
 \begin{enumerate}
 \item $p_1=2$ and $p_n=1$;

 \item $\digamma(p_i)=t_i$  for $1\leq i < n $;
 
 \end{enumerate}
 \item \emph{correct} if it is locally correct and $t_i = p_{i+1}$ for $1\leq i
<n$.

 \end{itemize}
 The database $\mn{River}_\kappa$ shown in Figure~\ref{fig:river}
 is not locally correct because $p_1 \neq 2$ and $t_n \neq 1$ (which
 must be the case as we assume $\digamma(1)=1$). It is also not
 correct.

 Clearly, $\digamma$ does not stop if and only if every locally
 correct $\mn{River}_\kappa$ is incorrect, and by Observation
 \ref{see} this is the case if and only if for each locally correct
 sequence $\kappa$ there exists a database-preserving
 $\querysig$-homomorphism from
 $\chase_{T_\mn{myth}}(\mn{River}_{\kappa})$ to $\mn{River}_{\kappa}$.
  
 Now the plan is as follows.  Take $\dbsig =\emptyset$. We define
 $T_1$ such that $\chase_{T_1}()$ is the `disjoint union' of all
 locally correct databases $\mn{River}_{\kappa}$.  Our $T_2$ will be
 the union of $T_1$ and $T_\mn{myth}$.
 A careful reader can notice that if this plan succeeds, then the proof
 of
 Point~1 of Theorem~\ref{thm:undecidable} will be completed.
And it will indeed succeed, but not without one little nuance. This
is the reason why we used quotations mark around 
the term  `disjoint union' above.
 
 \subsection{Constructing $T_1$: Recursive Rules.} \label{sec:recursive}
 
 The set of TGDs $T_1$ is the union of a set of linear TGDs
 $T_\mn{rec}$ constructed in this section and another set
 $T_\mn{proj}$ constructed in the subsequent section. As intended, it
 constructs the union of all locally correct databases
 $\mn{River}_{\kappa}$. The announced nuance is that the
 union is not disjoint, but massively overlapping. Fortunately, this
 does not compromise correctness of the reduction.

 The rules of $T_\mn{rec}$ will not mention symbols from
 $\querysig$. They instead use a schema $\Sigma_\digamma$ that
 consists of high arity relation symbols used as construction aids.
 These symbols are then later on related to those in $\querysig$ by
 $T_\mn{proj}$.  More precisely, $\Sigma_\digamma$ contains relation
 symbols $\mn{Start}$ of arity~8, $\mn{End}$ of arity~5, $\mn{Bridge}$
 of arity~4, $\WH^i_k$ (for WorkHorse) of arity $\alpha_k+\beta_k+5$
 for $0\leq k,i<\gamma$ and $\mn{BH}_k$ (for BridgeHead) of arity
 $\alpha_k+\beta_k+5$ for $0\leq k<\gamma$. In what follows, we use
 $\dagger$ as shorthand for `$c,e_1,e_2$'.  With $\pp$ and $\mm$, we
 denote addition and subtraction in the ring ${\mathbb Z}_\gamma$.

Since $\dbsig=\emptyset$, first of all we need a rule that will
create something out of nothing:
$$\rightarrow\;\; \exists \dagger, b_0,x_1,y_1,y_2,b_1\; \mn{Start}(\dagger, b_0,x_1,y_1,y_2,b_1).$$
In the next section, $T_{\mn{proj}}$ will generate a Pyramus-path from
$b_1$ via $x_1$ to $b_0$ and a Thisbe-path from $b_1$ via $y_2$ and
$y_1$ to $b_0$, determining the lengths $p_1=2$ and $t_1=3$ of the
river. These are the intended lengths since local correctness
prescribes $p_1=2$ and we assume that
$\digamma(2)=3$.

Before anything more is produced, we need to know that $b_1$ is
a bridge:
$$
\mn{Start}(\dagger,b_0,x_1,y_1,y_2,b_1) \; \rightarrow \;\;
\mn{Bridge}(\dagger, b_1).$$
Now we are going to put our workhorses to work by adding, for
$0 \leq k <\gamma$:
$$
\begin{array}{r@{\;}c@{\;}l}
\mn{Bridge}(\dagger, b) &\rightarrow& \exists x_1,\ldots, x_{\beta_i}, y_1, \ldots, y_{\alpha_i}\;\\[1mm]
&&\WH^{\beta_k}_k(\dagger, b, x_1,\ldots, x_{\beta_i}, b, y_1,\ldots,
y_{\alpha_1} )
\end{array}
$$
and for $0 \leq k,i < \gamma$:
$$
\begin{array}{l}
 \WH^i_k(\dagger, x_0, x_1,\ldots, x_{\beta_k}, y_0, y_1,\ldots, y_{\alpha_k})
\rightarrow\\[1mm]
\hspace*{8mm}\exists z_1,\ldots, z_{\beta_k}, u_1,\ldots, u_{\alpha_k}\\[1mm]
\hspace*{8mm}\WH^{i\pp \beta_k}_k(\dagger, x_{\beta_k} , z_1,\ldots, z_{\beta_k}, y_{\alpha_k}, u_1,\ldots, u_{\alpha_k}).
\end{array}
$$
The above rules patiently produce Pyramus- and Thisbe-paths that lead
to~$b$.  Via $T_{\mn{proj}}$, applying a rule with relation
$\mn{WH}^i_k$ in the body produces a Pyramus-path of length $\beta_k$
from $x_{\beta_k}$ to $x_0$ and a Thisbe-path of length $\alpha_k$
from $y_{\alpha_k}$ to $y_0$.  The superscript $\cdot^i$ is used to
remember how many Pyramus-edges were produced since the last bridge,
modulo~$\gamma$, and the subscript $\cdot_k$ is used to choose a
remainder class, that is, it expresses the promise that the
Pyramus-path between the two bridges is of length~$n$, for some number
$n$ with $n\ \mn{mod} \ \gamma=k$.
 
Then, at some point, the next bridge can be reached:
$$
\begin{array}{l}
 \WH^{k\mm \beta_k}_k(\dagger, x_0, x_1,\ldots, x_{\beta_k}, y_0, y_1,\ldots, y_{\alpha_k})\rightarrow\\[1mm]
 \hspace*{8mm} \exists z_1,\ldots, z_{\beta_k-1}, u_1, \ldots, u_{\alpha_k-1},b\\[1mm]
  \hspace*{8mm} \mn{BH}_k(\dagger, x_{\beta_k} , z_1,\ldots,
  z_{\beta_k-1},b, y_{\alpha_k}, u_1,\ldots,
  u_{\alpha_k-1},b)
\end{array}
$$
and
$$
\begin{array}{l}
\mn{BH}_k(\dagger, x_{\beta_k} , z_1,\ldots, z_{\beta_k-1},b, y_{\alpha_k}, u_1,\ldots, u_{\alpha_k-1},b)\rightarrow\\[1mm]
 \hspace*{8mm} \mn{Bridge}(\dagger, b).
\end{array}
$$
Note that in the upper rule, relation $\WH^{k\mm \beta_k}_k$ indicates
that we have seen $m$ Pyramus-edges, for some $m$ with
$m\ \mn{mod} \ \gamma={k\mm \beta_k}$, and that $\mn{BH}_k$ will generate
$k$ more Pyramus-edges, thus arriving at the promised remainder
of~$k$. It is also easy to see that if the chosen remainder class was
$k$ and the length of the Pyramus-path between two bridges produced by
the above rules is $n$, then the length of the Thisbe-path is
$\digamma(n) = n\alpha_k/\beta_k$. Thus, Point~2 of local correctness
is satisfied. 

 Finally, we want to be able to produce the last\footnote{Orographically the first, as we produce the river from the 
 mouth to the source.} segment of the river:
$$
\mn{Bridge}(\dagger, b) \rightarrow \exists b' \; \mn{End}(\dagger,b,b')
$$
This will produce direct Pyramus- and Thisbe-edges from $b'$ to $b$
(recall that $\digamma(1)=1$).

\subsection{Constructing $T_1$: Projecting on  $\querysig$.}\label{sec:projections}

We now generate the actual rivers as projections of the template
produced by~$T_{\mn{rec}}$. We start at the mouth:
$$
\begin{array}{l}
    \mn{Start}(\dagger, b_0,x_1,y_1,y_2,b_1) \rightarrow \\[1mm]
\hspace*{8mm}
                                                           \mn{Mouth}(b_0),
  \\[1mm]
\hspace*{8mm}
                                                           \mn{Pyramus}(x_1,b_0), 
                                                           \mn{Pyramus}(b_1,x_1), 
                                                           \mn{Pyramus}(x_1,e_2),
  \\[1mm]
\hspace*{8mm}
                                                           \mn{Thisbe}(y_1,b_0), 
                                                           \mn{Thisbe}(y_2,y_1), 
                                                           \mn{Thisbe}(b_1,y_2),
  \\[1mm]
\hspace*{8mm}
                                                           \mn{Thisbe}(y_1,e_1), 
                                                           \mn{Thisbe}(y_2,e_1),
  \\[1mm]
\hspace*{8mm}
                                                           \mn{Channel}(c,x_1), 
                                                           \mn{Channel}(c,y_1), 
                                                           \mn{Channel}(c,y_1),
  \\[1mm]
  \hspace*{8mm} \mn{Channel}(e_1,e_1), \mn{Pyramus}(e_1,e_1), \mn{Thisbe}(e_1,e_1) 
  \\[1mm]
  \hspace*{8mm} \mn{Channel}(e_2,e_2), \mn{Pyramus}(e_2,e_2), \mn{Thisbe}(e_2,e_2).
\end{array}
$$
The rules for $\WH^i_k$ are then as expected:
$$
\begin{array}{l}
\WH^i_k(\dagger, x_0, x_1,\ldots, x_{\beta_k}, y_0, y_1,\ldots, y_{\alpha_k}) \rightarrow \\[1mm]
\hspace*{8mm}\mn{Pyramus}(x_1,x_0),\ldots, \mn{Pyramus}(x_{\beta_k}, x_{\beta_k-1}),\\[1mm]
\hspace*{8mm} \mn{Pyramus}(x_1,e_2),\ldots, \mn{Pyramus}(x_{\beta_k}, e_2),\\[1mm]
\hspace*{8mm} \Tsibe(y_1,y_0),\ldots, \Tsibe(y_{\alpha_k}, y_{\alpha_k-1}),\\[1mm]
\hspace*{8mm} \Tsibe(y_1,e_1), \ldots, \Tsibe(y_{\alpha_k},e_1), \\[1mm]
\hspace*{8mm} \mn{Channel}(c,x_1),\ldots, \mn{Channel}(c,x_{\beta_k} ),\\[1mm]
\hspace*{8mm} \mn{Channel}(c,y_1),\ldots, \mn{Channel}(c,y_{\alpha_k}
).
\end{array}
$$
%
Rules for the relations $\mn{BH}_i$ are analogous, so we skip them.
There are also rules for projecting relations
$\mn{Bridge}$ and  $\mn{End}$:
$$
\begin{array}{@{}r@{\;}c@{\;}l}
  \mn{Bridge}(\dagger, b) &\rightarrow& \mn{Channel}(c,b),
  \mn{Pyramus}(b,e_1), \Tsibe(b,e_2) \\[1mm]
\mn{End}(\dagger,b,b') &\rightarrow& \mn{Pyramus}(b',b), \Tsibe(b',b), \mn{Encounter}(b,b').
\end{array}
$$
In the appendix, we show that:
\begin{lemma}
  $\digamma$ stops iff $T_2=T_1\cup T_{\mn{myth}}$ is
  $\Sigma_Q,\Sigma_D$-hom-conservative over
  $T_1 = T_{\mn{rec}} \cup T_{\mn{proj}}$.
\end{lemma}
This establishes Point~1 of Theorem~\ref{thm:undecidable}. For the
``if'' direction, one shows that if
$\chase_{T_2}(\emptyset) \rightarrow_{\Sigma_Q}
\chase_{T_1}(\emptyset)$, then every locally correct river is
incorrect, and thus $\digamma$ stops.  Since rivers may be long, but
are finite, it actually suffices that
$\chase_{T_2}(\emptyset) \rightarrow^{\mn{lim}}_{\Sigma_Q}
\chase_{T_1}(\emptyset)$ for $\digamma$ to stop, which by
Theorem~\ref{thm:firstchar} gives Point~2 of
Theorem~\ref{thm:undecidable}.

For Point~3 of Theorem~\ref{thm:undecidable}, we again want to use the
toolkit above, in particular $T_\mn{myth}$ and Observation~\ref{see}.
But the situation is a bit different now. In the above reduction, we
had at our disposal $T_1$ which was able to produce, from nothing, all
the rivers we needed.  So we could afford to have $\dbsig=\emptyset$.
Now, however, we no longer have $T_1$, but only $T_2$, and our
strategy is as follows. Recall that $\digamma$ stops if and only if
there is a locally correct $\mn{River}_\kappa$ that is correct, and
that $\mn{River}_\kappa$ is correct if there is no database-preserving
$\querysig$-homomorphism from
$\chase_{T_\mn{myth}}(\mn{River}_{\kappa})$ to $\mn{River}_{\kappa}$.
We use the database $D$ to guess a $\mn{River}_{\kappa}$ that admits
no such homomorphism. More precisely, we design $T_2$ so that it
verifies the existence of a (single) locally correct river in $D$ and
only if successful generates a chase with $T_\mn{myth}$ at the
\mn{Encounter} fact of that river. Details are in the appendix.

\pagebreak

\section{Triviality for Linear TGDs}
   \label{sect:linear}

   We show that for linear TGDs, $\dbsig,\querysig$-triviality is
   decidable and \PSpace-complete, and that it is only \coNP-complete when
   the arity of relation symbols is bounded. The upper bounds
   crucially rely on the observation that non-triviality is always
   witnessed by a \emph{singleton database}, that is, a database that
   contains at most a single fact.
   This was first noted (for CQ-conservative extensions) in the
   context of the decription logic DL-Lite
   \cite{DBLP:conf/aaai/KonevKLSWZ11}.
   \begin{restatable}{lemma}{lemlinTGDsingleton}
     \label{lem:linTGDsingleton}
     Let $T$ be a set of linear TGDs and $\Sigma_D,\Sigma_Q$
     schemas. Then
%
%
     $T$ is $\Sigma_D,\Sigma_Q$-trivial iff
     $\chase_{T}(D) \rightarrow_{\Sigma_Q} D$ for all singleton
     $\Sigma_D$-databases $D$.
   \end{restatable}
   So an important part of deciding triviality is to decide, given a set
   of TGDs $T$ and a singleton database $D$, whether
   $\chase_T(D) \not\rightarrow_{\Sigma_Q} D$. The basis for this is the
   subsequent lemma. 
%
   \begin{restatable}{lemma}{lemtwofacts}
     \label{lem:twofacts}
     Let $T$ be a set of linear TGDs and $D$ a singleton
     database. Then
     $\chase_T(D) \not\rightarrow_{\Sigma_Q} D$ implies
     that there is a connected database
     $C \subseteq \chase_T(D)$ that contains at most two facts
     and such that $C \not\rightarrow_{\Sigma_Q} D$.
   \end{restatable}
   Lemmas~\ref{lem:linTGDsingleton} and~\ref{lem:twofacts} provide us
   with a decision procedure for triviality for linear TGDs.  Given a
   finite set of linear TGDs $T$ and finite schemas $\Sigma_D$ and
   $\Sigma_Q$, all we have to do is iterate over all singleton
   $\Sigma_D$-databases $D$ and over all $C \subseteq \chase_T(D)$ that
   contain at most two facts and check (in polynomial time) whether
   $C \rightarrow_{\Sigma_Q} D$. To identify the sets~$C$, we can iterate
   over all exponentially many candidates and check for each of them
   whether $D,T \models q_C$ where $q_C$ is $C$ viewed as a Boolean
   CQ. This entailment check is possible in \PSpace
   \cite{DBLP:conf/ijcai/GottlobMP15}. 
   This yields the \PSpace upper bound in the following result.
   \begin{theorem}
     For linear TGDs, triviality is \PSpace-complete.  It is
     \coNP-complete if the arity of relation symbols is bounded by a
     constant.
   \end{theorem}
   For the \coNP upper bound, the crucial observation is that when the
   arity of relation symbols is bounded by a constant, then the
   entailment check `$D,T \models q_C$' is in \NPclass
   \cite{DBLP:journals/ai/GottlobKKPSZ14}. To decide non-triviality, we
   may thus guess $D$ and $C$ and verify in polynomial time that
   $C \not\rightarrow_{\Sigma_Q} D$ and in \NPclass that
   $D,T \models q_C$.  For the lower bounds, we reduce entailments of the
   form $D,T \models \exists x \, A(x)$, with $T$ a set of linear TGDs,
   to non-triviality for linear TGDs. This problem is \PSpace-hard in general
   \cite{DBLP:journals/jcss/CasanovaFP84} and \NPclass-hard when the
   arity of relation symbols is bounded by a constant. The reduction goes
   as follows. Let $D$, $T$, and $\exists x \, A(x)$ be given. Introduce
   a fresh binary relation symbol~$R$, set
   $\Sigma_D = \Sigma_Q = \{ R \}$, and let $T'$ be the extension of $T$
   with the TGDs
   $$
   \begin{array}{rcl}
     &\rightarrow& q_D \\[1mm]
     A(u) &\rightarrow& \exists x \exists y \exists z \, R(x,y), R(y,z)
   \end{array}
   $$
   where $q_D$ is $D$ viewed as a Boolean CQ. Note that there is no
   homomorphism from $R(x,y),R(y,z)$ into the singleton
   $\Sigma_D$-database $\{ R(c,c') \}$. Based on this, it is easy to verify that $T'$
   is $\Sigma_D,\Sigma_Q$-trivial iff $D,T \not\models \exists x \, A(x)$.

%

   \section{Frontier-One TGDs}
   \label{sect:frontierone}

   The purpose of this section is to show the following.
   \begin{theorem} \label{thm:frontierone}  
     For frontier-one TGDs, CQ-conservativity and hom-conservativity are
     decidable in \ThreeExpTime (and \TwoExpTime-hard).
   \end{theorem}
   \TwoExpTime lower bounds carry over from the description
   logic~\ELI, see \cite{DBLP:conf/ijcai/Gutierrez-Basulto18} for
   hom-conservativity and \cite{DBLP:journals/jair/JungLMS20} for
   CQ-conservativity. They already apply when only unary and binary
   relation symbols are admitted. In the remainder of the section, we
   thus concentrate on upper bounds.

   Both in the case of hom-conservativity and CQ-conservativity, we
   first provide a suitable model-theoretic characterization and then
   use it to find a decision procedure based on tree
   automata. The case of CQ-conservativity is significantly
   more challenging because of the appearance of homomorphism limits.
   
   \subsection{Deciding Hom-Conservativity}
   \label{subsect:homchar}

   We show that to decide hom-conservativity, it suffices to consider
   databases of bounded treewidth. Instead of using the
   standard notion of a tree decomposition, however, it is more
   convenient for us to work with what we call tree-like databases. We
   define these first.

%
%
   A \emph{$\Sigma$-instance tree} is a triple $\mathcal{T}=(V,E,B)$ with
   $(V,E)$ a directed tree and $B$ a function that assigns a
   $\Sigma$-database $B(v)$ to every $v\in V$ such that the following
   conditions are satisfied:
   \begin{enumerate}

     \item for every $a\in\bigcup_{v\in V}\adom(B(v))$, the restriction of
       $(V,E)$ to the nodes $v \in V$ such that $a \in \adom(B(v))$ is a
       tree of depth at most one;

     \item for every $(u,v)\in E$, $\adom(B(u))\cap
       \adom(B(v))$ contains at most one constant.

   \end{enumerate}
   The \emph{width} of the instance tree is the supremum of the
   cardinalities of $\adom(B(v))$, $v \in V$.  A $\Sigma$-instance tree
   \Tmc defines an associated instance
   $I_{\mathcal{T}} = \bigcup_{v \in V} B(v)$.
   A $\Sigma$-instance $I$ is \emph{tree-like of width $k$} if there is
   a $\Sigma$-instance tree $\mathcal{T}$ of width $k$ with $I=I_{\mathcal{T}}$.

   Instance trees of width $k$ are closely related to tree
   decompositions of width $k$ in which the bags overlap
   in at most one constant. For example, Condition~1 strengthens the usual
   connectedness condition and is
   more closely tailored towards our purposes.

   %
   \begin{restatable}{theorem}{thmhomchar}\label{thm:homchar}
     Let $T_1$ and $T_2$ be sets of frontier-one TGDs, and $\dbsig$ and
     $\querysig$ schemas. Let $k$ be the body width of~$T_1$. Then the following are equivalent: 
     \begin{enumerate}

       \item $T_1 \models^{\text{hom}}_{\dbsig,\querysig} T_2$;

       \item $\chase_{T_2}(D) \rightarrow_{\Sigma_Q} \chase_{T_1}(D)$,
	 for all tree-like $\dbsig$-databases $D$ of width at
	 most~$k$. 

     \end{enumerate}
   \end{restatable}
   The ``$1 \Rightarrow 2$''-direction is a direct consequence of the definition
   of hom-conservativity. For the ``$2 \Rightarrow 1$''-direction, let $D$ be a
   $\dbsig$-database witnessing
   \mbox{$T_1 \not\models^{\text{hom}}_{\dbsig,\querysig} T_2$}, that is,
   $\chase_{T_2}(D) \not \rightarrow_{\Sigma_Q} \chase_{T_1}(D)$. We
   show in the appendix that the unraveling $U$ of $D$ into a
     tree-like $\dbsig$-instance of width $k$ also satisfies
   $\chase_{T_2}(U) \not \rightarrow_{\Sigma_Q}
   \chase_{T_1}(U)$. Compactness then yields 
   a finite subset $U'$ of $U$ that still satisfies $\chase_{T_2}(U')
   \not \rightarrow_{\Sigma_Q} \chase_{T_1}(U')$.

   \smallskip We show in the appendix how Theorem~\ref{thm:homchar}
   can be used to reduce $\Sigma_D,\Sigma_Q$-hom-conservativity to the
   (\ExpTime-complete) emptiness problem of two-way alternating tree
   automata ({\ata}s) and in this way obtain a 3\ExpTime upper
   bound. Here, we only give a sketch. Let $T_1$ and $T_2$ be sets of
   frontier-one TGDs, $\dbsig$ and $\querysig$ schemas, $k$ the body
   width of $T_1$, and $\ell$ the head width of~$T_1$.

   The \ata works on input trees that encode a tree-like database $D$
   of width at most $k$ along with a tree-like model $I_0$ of $D$ and
   $T_1$ of width at most $\max\{k,\ell\}$. It verifies that
   $\chase_{T_2}(D) \not\rightarrow_{\Sigma_Q} I_0$. If such an $I_0$
   is found, then
   $\chase_{T_2}(D) \not\rightarrow_{\Sigma_Q} \chase_{T_1}(D)$ because
   $\chase_{T_1}(D) \rightarrow I_0$.  The converse is also true since
   $\chase_{T_1}(D)$ is tree-like of width $\max\{k,\ell\}$. In fact,
   the instance $\chase_{T_1}(D)|^\downarrow_c$ that the chase
   generates below each $c \in \mn{adom}(D)$ (see
   Section~\ref{sect:prelim}) is tree-like of width~$\ell$.

   Since our homomorphisms are database-preserving and $T_2$ is a set
   of frontier-one TGDs,
   $\chase_{T_2}(D) \not\rightarrow_{\Sigma_Q} I_0$ 
   if and only if there is a $c \in \mn{adom}(D)$ such that
   \mbox{$\chase_{T_2}(D)|^\downarrow_c \not\rightarrow_{\Sigma_Q}
   I_0$}. The {\ata} may thus check this
   latter condition, which it does by relying on the
   notion of a type. Since types also play a role in the subsequent
   sections, we make this precise.



   Let $T$ be a set of frontier-one TGDs. We use $\bodyCQ(T)$ to
   denote the set of unary or Boolean CQs that can be obtained by
   starting with the Boolean CQ
      $\exists x \exists\bar y\, \phi(x,\bar y)$ with
   $\phi(x,\bar y)$ the body of some TGD in $T$, 
   then dropping any number of atoms, and then identifying
   variables.  Finally, we may choose a variable as the answer
   variable or stick with only quantified variables. A
   \emph{$T$-type} is a subset $t \subseteq
   \bodyCQ(T)$ such that for some instance $I$ that is a model of
   $T$ and some $c \in \adom(I)$,
 \begin{enumerate}

   \item $q(x) \in t$ iff $c \in q(I)$ for all unary $q(x) \in
     \bodyCQ(T)$~and

   \item $q \in t$ iff $I \models q$ for all Boolean $q \in \bodyCQ(T)$.

 \end{enumerate}
 We then also use $\tp_T(I,c)$ to denote $t$.  We assume that every
 type contains the additional formula $\mn{true}(x)$.  We may then
 view $t$ 
 as a unary CQ with free variable
 $x$ and thus as a (canonical) database. For brevity, we use
 $t$ also to denote both of these.
 $\TP(T)$ is the set of all
 $T$-types. Note that the number of types is double exponential
 in $||T||$. The type
 $\tp_T(\chase_T(D),c)$ tells us everything we need to know about
 $c$ in the chase of database~$D$ with $T$, as follows.
\begin{lemma}
  \label{lem:typedet}
  Let $T$ be a set of frontier-one TGDs, $I$ an instance, and
  $c \in \mn{adom}(I)$. Then $\chase_T(I)|^\downarrow_c$ and
  $\chase_T(J)|^\downarrow_c$ are homomorphically equivalent,
  where $J$ is obtained from $\mn{tp}_T(\chase_T(I),c)$ by replacing the
  free variable $x$ with~$c$.
\end{lemma}
The proof of Lemma~\ref{lem:typedet} is straightforward by
reproducing chase steps from the construction of $\chase_T(I)$
in $\chase_T(J)$ and vice versa. Details are omitted.

So to verify that $\chase_{T_2}(D) \not\rightarrow_{\Sigma_Q}
I_0$, a \ata may guess a constant $c$ in the database
$D$ represented by the input tree, and it may also guess the type
$\tp_{T_2}(\chase_{T_2}(D),c)$. It then goes on to verify that
$\tp_{T_2}(\chase_{T_2}(D),c)$ was guessed correctly (which is not
entirely trivial as
$\chase_{T_2}(D)$ is \emph{not} encoded in the input). Exploiting
Lemma~\ref{lem:typedet}, it then starts from type
$\tp_{T_2}(\chase_{T_2}(D),c)$ to construct `in its states' the
instance
$\chase_{T_2}(D)|^\downarrow_c$, simultaneously walking through the
instance 
$I_0$ encoded by the input tree to verify that, as desired,
$\chase_{T_2}(D)|^\downarrow_c \not\rightarrow_{\Sigma_Q}
I_0$ (we actually build a \ata for verifying
$\chase_{T_2}(D)|^\downarrow_c \rightarrow_{\Sigma_Q}
I_0$ and then complement).

   \subsection{Deciding CQ-Conservativity}
   \label{subsect:boundedhoms}

   We start with showing that, also for deciding CQ-conservativity,
   it suffices to consider tree-like databases. In addition, 
   it suffices to consider CQs $q$ of arity $0$ or~$1$.
%
   \begin{restatable}{theorem}{thmcharone} \label{thm:cqcharone} Let
     $T_1$ and $T_2$ be sets of frontier-one TGDs, and $\dbsig$ and
     $\querysig$ schemas. Let $k$ be the body width of~$T_1$.
     Then the following are equivalent:
     \begin{enumerate}

       \item $T_1 \models^{\text{CQ}}_{\dbsig,\querysig} T_2$;

       \item \mbox{$q_{T_2}(D) \subseteq q_{T_1}(D)$}, for all
	 tree-like $\dbsig$-databases $D$ of width at most $k$ and
	 connected $\querysig$-CQs $q$ of arity 0 or~1.

     \end{enumerate} \end{restatable}
   The proof of Theorem~\ref{thm:cqcharone} first concentrates on
   restricting the shape of the database, using unraveling and compactness
   as in the proof of Theorem~\ref{thm:homchar}.  In a second step,
   it is then not difficult to restrict also the shape of the CQ.


   \smallskip
   
 The following refinement of Theorem~\ref{thm:firstchar} is a straightforward consequence of Theorem~\ref{thm:cqcharone}.
 \begin{theorem}
   \label{thm:intermediatechar}
   Let $T_1$ and $T_2$ be sets of frontier-one TGDs, $\dbsig$ and
   $\querysig$ schemas, and $k$ the body width of $T_1$. Then
   the following are equivalent:
   \begin{enumerate}
   \item 
     $T_1 \models^{\text{CQ}}_{\dbsig,\querysig} T_2$;
     
   \item $\chase_{T_2}(D) \rightarrow^{\lim}_{\querysig} \chase_{T_1}(D)$,
   for all
   tree-like $\dbsig$-databases $D$ of width at most $k$.
   \end{enumerate}
 \end{theorem}
 Although Theorem~\ref{thm:intermediatechar} looks very similar to
 Theorem~\ref{thm:homchar}, it does not directly suggest a decision
 procedure. In particular, it is not clear how tree automata can deal
 with homomorphism limits. We next work towards a more operative
 characterization that pushes the use of homomorphism limits to parts
 of the chase that are $\Sigma_Q$-disconnected from the database and
 regular in shape. As we shall see, this allows us to get to grips
 with homomorphism limits.

 \smallskip

 For a database $D$, with $\chase_{T}(D)|^\con_{\Sigma}$ we
 denote the union of all maximally $\Sigma$-connected components of
 $\chase_{T}(D)$ that contain at least one constant from
 $\adom(D)$.
 \begin{restatable}{theorem}{thmcharsecond}
   \label{thm:charsecond}
   Let $T_1$ and $T_2$ be sets of frontier-one TGDs, $\dbsig$ and
   $\querysig$ schemas, and $k$ the body width of $T_1$.  Then
   $T_1 \models^{\text{CQ}}_{\dbsig,\querysig} T_2$ iff for all
   tree-like $\dbsig$-databases $D$ of width at most $k$, the
   following holds:
   \begin{enumerate}
     \item $\chase_{T_2}(D)|^\con_{\querysig} \rightarrow_{\querysig}
       \chase_{T_1}(D)$;

     \item for all maximally $\querysig$-connected components $I$ of
       $\chase_{T_2}(D) \setminus \chase_{T_2}(D)|^\con_{\querysig}$, one
       of the following holds:
       \begin{enumerate}

	 \item $I \rightarrow_{\querysig} \chase_{T_1}(D)$;

	 \item $I \rightarrow^{\lim}_{\querysig}
           \chase_{T_1}(D)|^\downarrow_c$ 
           for
	   some $c \in \adom(D)$. 

       \end{enumerate}

   \end{enumerate} \end{restatable}
The subsequent example illustrates the theorem.
\begin{example}\label{ex:charsnd}
  Consider the sets of TGDs $T_1,T_2$ and the schemas $\dbsig,
  \querysig$ from Example~\ref{ex:basic}. Recall that $T_2$ is
  $\dbsig,\querysig$-CQ-conservative over $T_1$. Since $\dbsig$
  contains only the unary relation $A$, we may w.l.o.g.\ concentrate
  on the $\dbsig$-database $D=\{A(c)\}$. Clearly,
  Point~1 of
  Theorem~\ref{thm:charsecond} is satisfied.

  For Point~2, observe that $\mn{chase}_{T_2}(D) \setminus
  \mn{chase}_{T_2}(D)|^\con_{\querysig}$
  contains only one maximally $\querysig$-connected component,
  which is of the form
  \[I = \{ R(c_1,c_0), R(c_2,c_1),\dots \}.\]
  Moreover, 
  $I\to_{\querysig}^\lim \chase_{T_1}(D)|^\downarrow_c$  and thus Point~2(b)
  is satisfied. Point~2(a) is not satisfied since
  $I\not\to_{\querysig}\chase_{T_1}(D)$.
  %
\end{example}
%


The easier `if' direction of the proof of Theorem~\ref{thm:charsecond}
relies on the fact that, as per Theorem~\ref{thm:cqcharone}, we can
concentrate on connected CQs of arity~0 or~1.  The interesting
direction is `only if', distinguishing several cases and using several
`skipping homomorphism' arguments (see 
Lemma~\ref{lem:skippinghomsgeneral}).

  \smallskip

  Points~1 and~2(a) of Theorem~\ref{thm:charsecond} are amenable to
  the same tree automata techniques that we have used for
  hom-conservativity.  Point~2(b) achieves the desired expulsion of
  homomorphism limits, away from the database $D$ to instances of
  regular shape. In fact, the number of possible $T_1$-types is
  independent of $D$ and thus by Lemma~\ref{lem:typedet} the number of
  instances
  $\chase_{T_1}(D)|^\downarrow_c$
  in Point 2(b) that have
  to be considered is also independent of~$D$. Moreover, these
  instances are purely chase-generated and thus regular in shape. The
  same is true for the instances $I$ in Point~2. We next take a closer
  look at the latter.

  \smallskip 
  
  Let $T$ be a set of frontier-one TGDs.
    A \emph{$T$-labeled
    database} is a pair $A=(D,\mu)$ with $D$ a database and
  $\mu: \adom(D) \rightarrow \TP(T)$. 
  We associate $A$ with a database $D_A$ that is obtained by starting
  with $D$ and then adding, for each $c \in \adom(D)$, a disjoint copy
  $D'$ of the type $\mu(c)$ viewed as a database and glueing the copy of $x$ in
  $D'$ to $c$ in $D_A$. We use $T$-labeled databases
  to describe fragments of chase-generated instances, and thus
  assume that $D_A$ contains only null constants.
  We also associate $A$ with a Boolean CQ $q_A$,
  obtained by viewing $D_A$ as such a CQ. With
  $\chase_T(D_A)|^\mn{con}_\Sigma$, we denote the union of all maximally
  $\Sigma$-connected components of $\chase_{T}(D_A)$ that contain
  a constant from $\mn{adom}(D)$.

%
%
%
%
%

  \smallskip 

    A \emph{labeled $\Sigma$-head fragment of $T_2$} is a
    $T_2$-labeled database $(F,\mu)$ such that $F$ can be obtained by
    choosing a TGD $\phi(x,\bar y) \rightarrow \exists \bar z \,
    \psi(x,\bar z) \in T_2$ and taking a maximally $\Sigma$-connected
    component of $\psi$ that does not contain the frontier variable.
    The following lemma follows from an easy analysis of the chase
    procedure.  Proof details are omitted.
\begin{lemma} \label{lem:finiteI}
  Let $I$ be a maximally $\Sigma_Q$-connected component of
  $\chase_{T_2}(D) \setminus \chase_{T_2}(D)|^\con_{\querysig}$, as in
  Point~2 of Theorem~\ref{thm:charsecond}.  Then for some labeled
  $\querysig$-head fragment $A=(F,\mu)$ of $T_2$,
  %
  \begin{enumerate}

    \item $\chase_{T_2}(D) \models q_A$, and 
    
    \item $I$ is homomorphically equivalent to
      $\chase_T(D_A)|^\mn{con}_\querysig$.


  \end{enumerate}
\end{lemma}
Clearly, the number of labeled $\Sigma$-head fragments of $T_2$ is
independent of $D$, just like the number of $T_1$-types. It thus
follows from Lemma~\ref{lem:finiteI} and what was said before it that
the number of checks
`$I \rightarrow^{\lim}_{\querysig} \chase_{T_1}(D)|^\downarrow_c$' in
Point~2(b) of Theorem~\ref{thm:charsecond} does not depend on $D$:
there is at most one such check for every labeled $\Sigma$-head
fragment of $T_2$ and every $T_1$-type. We can do all these checks in
a preprocessing step, before starting to build {\ata}s 
for
CQ-conservativity that implement the characterization provided by
Theorem~\ref{thm:charsecond}. Whenever the \ata needs
to carry out a check
`$I \rightarrow^{\lim}_{\querysig} \chase_{T_1}(D)|^\downarrow_c$' to
verify Point~2(b), we
can simply look up the precomputed result and let the \ata reject
immediately if it is negative.
%
Thus, the automata are completely freed
from dealing with homomorphism limits. 

\subsection{Precomputation}

It remains to show how to actually achieve the precomputation of the tests
`$I \rightarrow^{\lim}_{\querysig} \chase_{T_1}(D)|^\downarrow_c$' in
Point~2(b) of Theorem~\ref{thm:charsecond}. This is where we finally
deal with homomorphism limits.
%
The following theorem makes precise the problem that we actually have to
decide. 
%
\begin{theorem} \label{thm:decidelimit} Given two sets of frontier-one
  TGDs $T_1$ and~$T_2$, a schema $\Sigma$, 
  a labeled $\Sigma$-head fragment $A=(D,\mu)$ for $T_2$,
  and a $T_1$-type $\widehat t$, it can be
  decided in time triple exponential in $||T_1||+||T_2||$ 
  whether
  $\chase_{T_2}(D_A)|^\con_\Sigma \rightarrow^{\lim}_\Sigma
  \chase_{T_1}(\widehat t)$. 
\end{theorem}
We invite the reader to compare the decision problem formulated in
Theorem~\ref{thm:decidelimit} with Point~2(b) of
Theorem~\ref{thm:charsecond} in the light of Lemmas~\ref{lem:typedet}
and~\ref{lem:finiteI}.  The decision procedure used to prove
Theorem~\ref{thm:decidelimit} is again based on tree automata. To
enable their use, however, we first rephrase the
decision problem in Theorem~\ref{thm:decidelimit} in a way that
replaces homomorphism limits with unbounded homomorphisms.

Let $T_1$, $T_2$, $\Sigma$, $A=(D,\mu)$, and $\widehat t$ be as in
Theorem~\ref{thm:decidelimit}.
%
Recall that $A$ is associated with a database $D_A$ and a Boolean CQ
$q_A$. Here, we additionally use unary CQs~$q^c_A$, for every
$c \in \mn{adom}(D)$, which are defined exactly like $q_A$ except that
$c$ is now the answer variable.

The main idea for proving Theorem~\ref{thm:decidelimit} is to replace
homomorphism limits into $\chase_{T_1}(\widehat t)$ with homomorphisms
into a class of instances $\Rmc(T_1,\widehat t)$ whose disjoint union
should be viewed as a 
relaxation of $\chase_{T_1}(\widehat t)$. In particular, this
relaxation admits a homomorphism limit to $\chase_{T_1}(\widehat t)$,
but not a homomorphism. Let us make this precise.

We again use instance trees. This time, however, they are
not based on directed trees, but on \emph{directed pseudo-trees}, that
is, finite or infinite directed graphs $G=(V,E)$ such that every node
$v \in V$ has at most one incoming edge and $G$ is connected
and contains no (directed) cycle. Note that infinite directed
pseudo-trees need not have a root. For example, a two-way infinite
path qualifies as a directed pseudo-tree.

A \emph{$T_1$-labeled instance tree} has the form $\Tmc=(V,E,B,\mu)$
with $\Tmc'=(V,E,B)$ an instance tree (based on a directed
pseudo-tree) and
$\mu:\adom(I_{\Tmc'})\to \TP(T_1)$ a function that assigns a
$T_1$-type to every element in~$I_{\Tmc'}$. For $v\in V$, we use
$\mu_v$ to denote the restriction of $\mu$ to $\adom(B(v))$. Moreover,
we set $I_{\Tmc}=I_{\Tmc'}$.  We say that \Tmc is
\emph{$\widehat t$-proper} if the following conditions are satisfied:
\begin{enumerate}

  \item for every $v\in V$, one of the following holds:
    \begin{enumerate}

    \item $v$ is the root of $(V,E)$, $B(v)$ has the form
      $\{\mn{true}(c_0) \}$, and $\mu(c_0)=\widehat t$;

\item there is a TGD $\vartheta$ in $T_1$
    such that $B(v)$ is isomorphic to the 
    head of $\vartheta$ and $\widehat t,T_1 \models q_{(B(v),\mu_v)}$;

  \end{enumerate}
  
  
\item for every $(u,v)\in E$ such that $B(u) \cap B(v)$ contains a
  (single) constant $c$, we have
  $\mu_u(c),T_1\models q_{(B(v),\mu_v)}^{c}(x)$. That is: the constant
  $x$ from the type $\mu_u(c)$ viewed as a database is an answer
  to unary CQ  $q_{(B(v),\mu_v)}^{c}$ w.r.t.\ $T_1$.

\end{enumerate}
%
The announced class $\Rmc(T_1,\widehat t)$ consists of all instances
$I$ such that $I=I_\Tmc$ for some $\widehat t$-proper $T_1$-labeled
instance tree~\Tmc. It is easy to see that
$\chase_{T_1}(\widehat t) \in \Rmc(T_1,\widehat t)$ as there is a
$\widehat t$-proper $T_1$-labeled instance tree \Tmc such
that $I_\Tmc=\chase_{T_1}(\widehat
t)$. 
However, there are also instances $I \in \Rmc(T_1,\widehat t)$ that do
not admit a homomorphism to $\chase(\widehat t,T_1)$. The following
example illustrates their importance.
%
\begin{example} 
  Consider $T_1,T_2,\dbsig,\querysig$ from
  Example~\ref{ex:basic} and $I$ and $\widehat t$ from
  Example~\ref{ex:charsnd}. Then $I \not\rightarrow \mn{chase}_{T_2}(\widehat t)$.  However, we find a $\widehat t$-proper $T_1$-labeled instance tree
  $\Tmc=(V,E,B,\mu)$ such that $I \rightarrow I_\Tmc$.

  We may construct \Tmc by starting with a
  single node~$v_0$, $B(v_0)=\{R(c_1,c_0)\}$, and
  $$\mu(c_0)=\mu(c_1)=\{ B(x), \exists y \, R(x,y), B(y) \}.$$
  Then, repeatedly add a predecessor $v_{i+1}$ of $v_i$,
  with $B(v_{i+1})=\{R(c_{i+1},c_i)\}$ and
  $$\mu(c_{i+1})=\{ B(x), \exists y \, R(x,y), B(y) \},$$ ad
  infinitum. The resulting tree \Tmc is $\widehat t$-proper and satisfies
  $I_\Tmc=I$. Note that it does not have a root.
\end{example}

\smallskip

The next lemma is the core ingredient to the proof of
Theorem~\ref{thm:decidelimit}. Informally, it states that when
replacing $\chase_{T_1}(\widehat t)$ with instances from
$\Rmc(T_1,\widehat t)$, we may also replace homomorphism limits with
homomorphisms.
\begin{restatable}{lemma}{lemlimtoinfinite}
  \label{lem:limtoinfinite}
  Let $I$ be a countable $\Sigma$-connected instance 
  such that
  $\mn{adom}(I)$ contains only nulls.  Then
  $I \rightarrow^{\lim}_\Sigma \chase_{T_1}(\widehat t)$ iff there is
  an $\widehat{I} \in \Rmc(T_1,\widehat t)$ with \mbox{$I\rightarrow
  \widehat I$}.
\end{restatable}
In the proof of Lemma~\ref{lem:limtoinfinite}, the laborious direction
is `only if', 
where one assumes that
$I \rightarrow^{\lim}_\Sigma \chase_{T_1}(\widehat t)$ and then uses
finite subinstances $J_1\subseteq J_2\subseteq \cdots$ of $I$ with
$I = \bigcup_{i\geq 1}J_i$ and homomorphisms $h_i$ from $J_i$ to
$\chase_{T_1}(\widehat t)$, $i \geq 1$, to identify the desired
instance $\widehat I \in \Rmc(T_1,\widehat t)$. This again involves several
`skipping homomorphisms' type of arguments.

Using Lemma~\ref{lem:limtoinfinite}, we give a decision procedure
based on {\ata}s that establishes Theorem~\ref{thm:decidelimit}. The
\ata accepts input trees that encode an instance
$I\in \Rmc(T_1,\widehat t)$ that admits a $\Sigma$-homomorphism from
$\chase_{T_2}(D_A)|^\con_{\Sigma}$. 

\section{Future Work}

It would be interesting to determine the exact complexity of
hom- and CQ-conservativity for frontier-one TGDs. We tend to
think 
that these problems are 3\ExpTime-complete. Note that in the
description logic \ELI, they are \TwoExpTime-complete \cite{DBLP:journals/jair/JungLMS20}.

It would also be interesting to study conservative extensions and
triviality for other classes of TGDs that have been proposed in
the literature. Of course, it would be of particular interest to
identify decidable cases. Classes for which undecidability does not
follow from the results in this paper include acyclic and sticky
TGDs, which exist in several forms, see for instance
\cite{DBLP:journals/pvldb/CaliGP10}.

\bibliographystyle{kr}
\bibliography{local}

\cleardoublepage
\appendix

\section{The Chase}

We introduce the chase. Let $I$ be an instance and $T$ a set of TGDs.  A 
TGD 
$\phi(\bar x,\bar y) \rightarrow \exists \bar z \, \psi(\bar x,\bar z) 
\in T$ is \emph{applicable} at a tuple $\bar c$ of constants 
in $I$ if $\phi(\bar c,\bar c') \subseteq I$ for some $\bar c'$ and there is no
homomorphism $h$ from $\psi(\bar x,\bar z)$ to $I$ such that
$h(\bar x)=\bar c$. In this case, the {\em 
  result of applying the TGD in $I$ at $\bar c$} is the instance 
$I \cup \{\psi(\bar c,\bar c'')\}$ where $\bar c''$ is the tuple 
obtained from $\bar z$ by simultaneously replacing each variable $z$
with a fresh null, that is, a null that does not occur in~$I$. We also 
refer to such an application as a \emph{chase step}.  

A {\em chase sequence for $I$ with $T$} is a sequence of instances
$I_0,I_1,\dots$ such that $I_0=I$ and each $I_{i+1}$ is the result of
applying some TGD from $T$ at some tuple $\bar c$ of
constants in $I_i$. The \emph{result} of this chase sequence is the
instance $J = \bigcup_{i \geq 0} I_i$. The chase sequence is
\emph{fair} if whenever a TGD from $T$ is applicable to a tuple
$\bar c$ in some $I_i$, then this application is a chase
step in the sequence. Every fair chase sequence for $I$ with $T$ has
the same result, up to homomorphic equivalence. Since for our purposes
all results are equally useful, we use $\chase_T(I)$ to denote the
result of an arbitrary, but fixed chase sequence for $I$ with $T$
and call  $\chase_T(I)$ the \emph{result of chasing $I$ with $T$.}
\begin{lemma}\label{pro:chase}
%
  Let $T$ be a finite set of TGDs and $I$ an instance. Then 
  for every model $J$ of $T$
  with $I \subseteq J$, there is a homomorphism $h$ from 
  $\chase_T(I)$ to $J$ that is the identity on $\adom(I)$. 
%
  %
\end{lemma}

\section{Proofs for Section~\ref{sect:consext}}

\lemskippinghomsgeneral*
\begin{proof} Assume that $I_1 \rightarrow^{\lim}_\Sigma I_2$.  We
  need to find a database-preserving $\Sigma$-homomorphism $h$ from
  $I_1$ to $I_2$. Let $\adom(I_1)=\{c_1,c_2,\dots\}$ (finite or
  infinite) and for $i \geq 1$, let $A_i$ be the set of the first
  $\min\{i,|\adom(I_1)|\}$ constants from the sequence $c_1,c_2,\dots$.
  Since $I_1 \rightarrow^{\lim}_\Sigma I_2$, we find for each
  $i \geq 0$ a database-preserving $\Sigma$-homomorphism $h_i$ from
  $I_1|_{A_i}$ to $I_2$. For $A \subseteq A_i$, we use $h_i|_A$ to
  denote the restriction of $h_i$ to domain $A$. We would be done if
  we knew that the sequence $h_1,h_2,\dots$ satisfied the following
  uniformity condition:
  \begin{itemize}

  \item[($*$)] $h_{j}|_{A_i}=h_i$ for $j > i > 0$,

  \end{itemize}
  as then we could simply take $h=\bigcup_{i \geq 1} h_i$. We show how
  to extract from $h_1,h_2,\dots$ a sequence $g_1,g_2,\dots$ that
  satisfies Condition~($*$), and then define
  $h=\bigcup_{i \geq 1} g_i$. During the construction of the sequence
  $g_1,g_2,\dots$, we shall take care that for every $g_i$, there are
  infinitely many $j > i$ with $h_j|_{A_i}=g_i$.

    \smallskip
  
    To start, we need a homomorphism from $I_1|_{A_1}$ to $I_2$. Since
    $A_1$ and $I_2$ are finite, there are only finitely many mappings
    $f:A_1 \rightarrow \adom(I_2)$ and thus there must be such a
    mapping $f$ such that $h_i|_{A_1}=f$ for infinitely many
    $i \geq 1$. Set $g_1=f$.

  Now assume that $g_1,\dots,g_n$ have already been defined.
  Then there is an infinite set $\Gamma$ of indices $j>i$ 
  such that for all $j \in \Gamma$, $h_j|_{A_n}=g_n$. Since $I_2$
  is finite, there are only finitely many extensions $f: A_{n+1}
  \rightarrow \adom(I_2)$ of $g_n$ an thus there must be such
  an extension $f$ such that $h_j|_{A_{n+1}}=f$ for infinitely
  many $j \in \Gamma$. Set $g_{n+1}=f$.

  \smallskip
  
  The resulting sequence clearly satisfies ($*$) and thus the proof
  is done.
\end{proof}

\section{Proofs for Section~\ref{sect:undec}}

\subsection{ Analysing $\chase_{T_1}()$ and $\chase_{T_2}()$.}\label{sec:proofT1}

As we have already said in Section \ref{sec:recursive}, $T_1$ is now
defined as $T_{\mn{rec}}\cup T_{\mn{proj}}$ and $T_2$ is  $T_1\cup T_{\mn{myth}}$.
It is easy to notice that:

\begin{lemma}
  $\chase_{T_2}()=\chase_{T_{\mn{myth}}}(\chase_{T_1}())$.
\end{lemma}

\begin{proof} 
  The claim follows since all the facts $T_{\mn{myth}}$ can possibly produce are from $\querysig$. And all the
  facts in all the bodies of rules from $T_1$ are from $\Sigma_\digamma$.
\end{proof}

Now we are going to study $\querysig$-homomorphisms from $\chase_{T_2}()$ to $\chase_{T_1}()$.
First of all notice that:

\begin{lemma}\label{lem:b-dist}
  Identity is the only $\querysig$-homomorphisms from
  $\chase_{T_1}()$ to $\chase_{T_1}()$.
\end{lemma}

\begin{proof} 
  Suppose $h$ is a $\querysig$-homomorphisms from $\chase_{T_1}()$  to $\chase_{T_1}()$.
  It is easy to notice that $h(x)=x$ if (*) $x$ is one of the special
  constants $c,e_1,e_2$, or if (**) $\mn{Mouth}(x)$ holds.

  Then we use (*) to prove that $x$ is a bridge if and only if $h(x)$
  is a bridge or, in other words, that (***) $h$ maps bridges to
  bridges and non-bridges to non-bridges.

  As the next step imagine two bridges $b$ and $b'$, such that there
  is a Pyramus path from $b'$ to $b$ which does not visit any other
  bridge. For such $b$ and $b'$ define $dist(b',b)$ as the length of
  this Pyramus path. Now it follows from the construction of $T_1$
  that if for bridges $b,b',b''$ it holds that
  $dist(b',b)=dist(b'',b)$ then $b'=b''$. Since (in our case) distance
  is preserved by homomorphisms, this implies that (****) if $x$ is a
  bridge then $h(x)=x$ (the detailed proof would be by induction, with
  (**) serving as induction basis and (***) used in the induction
  step).

  Then it easily follows from (****) that the lemma also holds true
  for the non-bridge constants.
\end{proof}

Let us denote with \ccc the chase of $\{\mn{Encounter(c_0,c_0'}\}$
shown in Figure~\ref{figA}. For an instance $I$ and a fact
$F=\mn{Encounter(c,c'})\in I$, we denote 
with $I\cup_F \ccc$ the instance
obtained from $I$ by adding 
a disjoint copy of \ccc to $I$ and identifying $c_0$ with $c$ and $c_0'$ with
$c'$. It is now easy to see that:
\begin{lemma}\label{lem:t2-union}
  Let $E = \chase_{T_1}()|_{\{\mn{Encounter}\}}$  be the set of all
  facts for the relation symbol $\mn{Encounter}$ in $\chase_{T_1}$. Then:
\[\chase_{T_2}() = \bigcup_{F\in
E} \chase_{T_1}() \cup_F \ccc\]
\end{lemma}


The lemma means that $\chase_{T_2}()$ is $\chase_{T_1}()$ with one new copy of the chase $\ccc$ 
from Figure \ref{figA}, attached 
to every fact $F$ of the relation $\mn{Encounter}$ in $\chase_{T_1}()$.

It now follows easily from Lemma \ref{lem:t2-union} and Lemma \ref{lem:b-dist} that:

\begin{lemma}\label{lem:decomposition}
  The following two conditions are equivalent:
  \begin{enumerate}[label=(\roman*)]

    \item There exists a $\querysig$-homomorphism $h$ from
      $\chase_{T_2}()$  to $\chase_{T_1}()$.

    \item For each fact $F$ of the relation $\mn{Encounter}$ in
      $\chase_{T_1}()$ there exists a $\querysig$-homomorphism $h_F$
      from $\chase_{T_1}() \cup_F \ccc$ to $\chase_{T_1}()$.

  \end{enumerate}

\end{lemma}

\begin{proof} 
  The $\textit{(i)}\Rightarrow \textit{(ii)}$ direction is trivial.
  For the opposite direction, assume (ii), consider all the
  homomorphisms $h_F$ and notice that (by Lemma \ref{lem:b-dist}) they
  all agree on $\chase_{T_1}()$. Now  take $h$ as the union of all
  $h_F$.
\end{proof}

Let us now try to analyze the structure of $\chase_{T_1}()$.  Notice
that for each\footnote{With the obvious exception of the single fact
of the relation $\mn{Start}$ in $\chase_{T_1}()$, which has no parent.  }
fact $F_1$ of $\chase_{T_1}()\arrowvert_{\Sigma_\digamma}$  there
exists exactly one {\em parent of} $F_1$, a fact    $F_0$ of
$\chase_{T_1}()\arrowvert_{\Sigma_\digamma}$ such that $F_1$ was
created directly from $F_0$ by an application of a TGD from
$T_{\mn{rec}}$.  We will write $F_0\rightarrow F_1$ to say that $F_0$ is
a parent of $F_1$. By $\stackrel{\ast}{\rightarrow}$, we will denote
the reflexive and transitive closure of $\rightarrow$.
Notice that if $F$ is a fact of relation $\mn{End}$ then $F\rightarrow G$ is never true.

For a fact $F\in \chase_{T_1}()\arrowvert_{\Sigma_\digamma}$ define: 
\[\mn{Ancestors}(F) =\{G\in \chase_{T_1}()\arrowvert_{\Sigma_\digamma}:
G\stackrel{\ast}{\rightarrow} F\}\]
In natural language, $\mn{Ancestors}(F)$ comprises all the facts of
$\chase_{T_1}()\arrowvert_{\Sigma_\digamma}$ which were necessary to
produce $F$ (with fact $F$ included). 

Finally, let $\mn{Ancestors}_{\querysig}(F)$ be the set of all
$\querysig$-facts which can be produced from some fact in
$\mn{Ancestors}(F)$ by using a rule from $T_{\mn{proj}}$.

Now, it follows from the construction of $T_1$ that:

\begin{lemma} \label{lem:str-of-T1-chase}
  $\chase_{T_1}()$ satisfies the following: 
  \begin{enumerate}[label=(\roman*)]

    \item For a fact $F\in \chase_{T_1}()\arrowvert_{\Sigma_\digamma}$
      the database $\mn{Ancestors}_{\querysig}(F)$ is a {\em locally
      correct} river if and only if $F$ is an fact of the relation
      $\mn{End}$.

    \item For each locally correct sequence $\kappa$ there exists a
      fact $F\in \chase_{T_1}()$  (of the relation $\mn{End}$) such
      that $\mn{Ancestors}_{\querysig}(F)$ is (isomorphic to)
      $\mn{River}_\kappa$.

\end{enumerate}
\end{lemma}

Let us remind the Reader that in order to prove
Theorem~\ref{thm:undecidable} Point~1, we need to to prove the
equivalence~$(\heartsuit.l)$ from Section~\ref{sec:source}, restated
here for convenience: 
\begin{description} \item[$(\heartsuit.l)$] $\digamma$ does not stop
      if and only if $T_2$ is  $\dbsig,\querysig$-hom-conservative
      over $T_1$.  \end{description}
The ``$\Rightarrow$''-direction of this equivalence is now easy to show.
If $\digamma$ does not stop then every locally correct river is
incorrect. This means (using Lemma~\ref{lem:str-of-T1-chase} (i) and
Observation~\ref{see}) that for each fact $G$ of relation
$\mn{Encounter}$ in $\chase_{T_1}()$, which was created by projecting
some fact $F$ of the relation $\mn{End}$, the instance
$\chase_{T_1}() \cup_G \ccc$ (connected to $\chase_{T_1}()$ via $G$) can be
homomorphically embedded in $\mn{Ancestors}_{\querysig}(F)$. So, by Lemma
\ref{lem:decomposition}, there exists a  $\querysig$-homomorphism $h$
from $\chase_{T_2}()$  to $\chase_{T_1}()$.

What about the  ``$\Leftarrow$''-direction? Suppose $\digamma$ stops.
Then we know, from Lemma \ref{lem:str-of-T1-chase} (ii), that there is
an $F$ in $\chase_{T_1}()$ such that $\mn{Ancestors}_{\querysig}(F)$ is
$\mn{River}_{\kappa_F}$ for the correct sequence $\kappa_F$. Hence
(using Observation \ref{see} once again) we know that there is no
homomorphism from $\mn{Ancestors}_{\querysig}(F)\cup_G \ccc$ to
$\mn{Ancestors}_{\querysig}(F)$ (where again, $G$ is a fact of relation
$\mn{Encounter}$ resulting from projecting $F$). By
Lemma~\ref{lem:decomposition}, it suffices to prove is that there is
no homomorphism from $\chase_{T_1}() \cup_G \ccc$  to
$\chase_{T_1}()$. Now, observe that $\chase_{T_1}()$ is a union of all
possible locally correct instances $\mn{River}_\kappa$ but it is
\textit{not} a disjoint union: they all share the mouth, and there is
a lot of overlap between them. So maybe it is possible that one could
homomorphically embed, in $\chase_{T_1}()$, the copy of $\ccc$
attached to $G$, using facts of $\chase_{T_1}()$ which are not in
$\mn{Ancestors}_{\querysig}(F)$? 

Our last lemma says that there is no such embedding:

\begin{lemma}
  If there exists a homomorphism from $\mn{Ancestors}_{\querysig}(F)\cup_G
  \ccc$  to $\chase_{T_1}()$ then there exists a homomorphism from
  $\mn{Ancestors}_{\querysig}(F)\cup_G \ccc$ to $\mn{Ancestors}_{\querysig}(F)$.
\end{lemma}

\begin{proof}
  Note that the homomorphism of $\ccc$ into $\chase_{T_1}()$
  starts at some fact of relation \mn{Encounter}, and then follows
  essentially facts $\mn{Pyramus}(c,c')$ or $\mn{Thisbe}(c,c')$ in
  this direction. So it remains to observe that sets
  $\mn{Ancestors}_{\querysig}(F)$ are closed under taking successors
  along relations $\mn{Pyramus}$ and $\mn{Thisbe}$, by construction of
  $T_1$.

  More precisely, suppose $x$ is a constant of $\mn{Ancestors}_{\querysig}(F)$ and $y$
  is a constant of $\chase_{T_1}()$.  Suppose also that the fact
  $\mn{Pyramus}(x,y)$ or the fact $\mn{Thisbe}(x,y)$ is in
  $\chase_{T_1}()$. Then $y$ is a constant of $\mn{Ancestors}_{\querysig}(F)$.
\end{proof}

In consequence, the ``$\Leftarrow$''-direction of $(\heartsuit.l)$ also
holds, and Point~1 of Theorem~\ref{thm:undecidable} is proven.


\subsection{Proof of Point~2 of Theorem~\ref{thm:undecidable}}

In order to prove Point~2 of Theorem \ref{thm:undecidable} let us just
notice that we can simply reuse the proof of Point~1 from the previous
section.

Clearly, if $\digamma$ does not stop then $T_2$ is
$\Sigma_D,\Sigma_Q$-hom-conservative over $T_1$ so it is also
$\Sigma_D,\Sigma_Q$-CQ-conservative. We need to notice that if
$\digamma$ stops then $T_2$ is not $\Sigma_D,\Sigma_Q$-CQ-conservative
over $T_1$. To this end, it will be enough to find a \emph{finite}
subinstance $Q$ of $\chase_{T_2}()$ which cannot be homomorphically
embedded in $\chase_{T_1}()$.

So suppose  $\digamma$ stops. Let $\mn{River}_{\kappa_F}$ be the
correct river and let $m\in \mathbb N$ be any natural number bigger
than the longest $\Tsibe$ or $\mn{Pyramus}$ path in
$\mn{River}_{\kappa_F}$ which does not visit an eternity.

Let $\ccc_m$ be the fragment of the chase $\ccc$ resulting from $m$
applications of existential TGDs in $T_{\mn{myth}}$,
(and then using the datalog projections)\footnote{The instance from Figure \ref{figA} can be seen as $\ccc_4$.}.
Let also $\mn{River}_{\kappa_F} \cup_G \ccc_m $ be the union of the
two instances, with the only fact of
the $\mn{Encounter}$ relation in $\mn{River}_{\kappa_F}$ identified with the only  $\mn{Encounter}$ fact in $\ccc_m $.

We can use the arguments that were used in the proof of the 
``$\Leftarrow$''-direction in the previous section to show that
$\mn{River}_{\kappa_F} \cup_G \ccc_m $ is the $Q$ we need.

\subsection{Proof of Point~3 of Theorem \ref{thm:undecidable}}
\label{sec:hmm}


We use the same schema $\querysig$ as before, except that
$\mn{Encounter}$ is not in $\querysig$. We will define $\dbsig$
so that
$\querysig \setminus \{\mn{Encounter}\}\subseteq \dbsig$.  Also,
$\mn{Start}$ is now in $\dbsig$. And, for each rule $\cal R$ from
$T_\mn{rec}$ of the form
 $$P(\bar{x}, \bar{y})\rightarrow \exists \bar z\; Q(\bar{y}, \bar{z})$$
 $\Sigma_D$ contains a new relation symbol ${\cal S}_{\cal R}$ of
 arity $|\bar{x}|+|\bar{y}|+|\bar{z}|$ in $\dbsig$.
 
 We now define a set of TGDs $T_0$ whose triviality we are interested
 in. $T_0$ contains all rules from $T_\mn{myth}$ as well as the new
 rule
 $$\mn{End}(\dagger,b,b'), \mn{Pyramus}(b',b), \Tsibe(b',b) \rightarrow \mn{Encounter}(b,b').$$
 Finally, for each rule $\cal R$ from $T_\mn{rec}$, as above,
 $T_0$ contains the following  rule:
 \begin{description}

 \item[$(\clubsuit)$] \mbox{${\cal S}_{\cal R}(\bar{x},\bar{y},\bar{z}), P(\bar{x}, \bar{y}), \chase_{T_\mn{proj}}(\{P(\bar{x}, \bar{y})\}) 
\,\rightarrow \, Q(\bar{y}, \bar{z})$.}
   
 \end{description}
 Notice that this rule is indeed guarded. This is because
 $\chase_{T_\mn{proj}}(\{P(\bar{x}, \bar{y})\})$ is finite
 \footnote{The set of all facts which can be produced from
   $P(\bar{x}, \bar{y})$ by projections from $T_\mn{proj}$.}, and
 contains only variables from $\bar{x} \cup \bar{y}$.

 Notice that the only $\querysig$-facts created when chasing with
 $T_0$ are the ones created by $T_\mn{myth}$ and, in order for them to
 be created, some $\mn{Encounter}$ fact $\mn{Encounter}(c,c')$ must
 be produced first. It is now easy to see that:

\begin{lemma}\label{lem:lastone}
  Let $D$ be a $\dbsig$-database and suppose that for each fact
  $F=\mn{Encounter}(c,c')$ in $\chase_{T_0}(D)$, there exists a
  database-preserving $\querysig$-homomorphism $h$ from $T_\mn{myth}(\{F\})$ to $D$, such
  that $h(c)=c$ and $h(c')=c'$.  Then there exists a
  database-preserving 
  $\querysig$-homomorphism from $\chase_{T_\mn{myth}}(D)$ to $D$.
\end{lemma}


To establish Point~3 of Theorem \ref{thm:undecidable}, we need to
prove that the equivalence $(\heartsuit.g)$ holds true.


So assume first that $\digamma $ stops, and let $\kappa_F$ be the correct sequence.
We need to produce a database $D$ such that $\chase_{T_0}(D)$ has no
database-preserving $\querysig$-homomorphism to $D$.

We know, from Lemma \ref{lem:str-of-T1-chase}(ii) that
 there exists a fact $F\in \chase_{T_1}()$  (of the relation $\mn{End}$) 
 such that $\mn{Ancestors}_{\querysig}(F)= \mn{River}_{\kappa_F}$. 
Let $D$ be the database that consists of the following:
\begin{itemize}

\item all facts of $\mn{Ancestors}_{\querysig}(F)$, with the exception of the $\mn{Encounter}$ fact, which is not in $\dbsig$
(so we almost have the entire $\mn{River}_{\kappa_F}$ in $D$, just the $\mn{Encounter}$ fact is missing);

\item the (only) fact of the relation $\mn{Start}$ from $\chase_{T_1}()$;

\item for any two facts $G=P(\bar{a},\bar{a'})$ and $G'=Q(\bar{a'}, \bar{a''})$ from $\mn{Ancestors}(F)$ such that 
$G\rightarrow G'$ and $G'$ was created from $G$ using the rule $\cal
R$ of $T_\mn{rec}$, the fact 
${\cal S}_{\cal R}(\bar{a},\bar{a'},\bar{a''})$.

\end{itemize}
 Now let us try to imagine what facts will be produced by
 $\chase_{T_0}(D)$.  There is this  $\mn{Start}$ fact in $D$, and it
 will match with the $P$ of one of the rules of the form $(\clubsuit)$
 in $T_0$, producing a next fact in $\mn{Ancestors}(F)$, which will
 again match with the $P$ of some other $(\clubsuit)$ rule, and so on.
 All the facts of $\mn{Ancestors}(F)$ will be produced in this way,
 with $F$ as the last one. But recall that $F$ is an fact of the
 relation $\mn{End}$, so we have a rule in $T_0$ that will now use $F$
 to produce the missing $\mn{Encounter}$ fact of
 $\mn{River}_{\kappa_F}$. Now $T_\mn{myth}$ will fire, producing the
 chase $\ccc$, which (by Observation \ref{see}, since ${\kappa_F}$ is
 correct) will not have a database-preserving $\querysig$-homomorphism
 to $D$.

\smallskip

For the converse direction, assume that $\digamma $ does not stop and
let $D$ be any $\dbsig$-database.

The next lemma says that if any $\mn{Encounter}$ fact is created in
$\chase_{T_0}(D)$ then $D$ must contain (a homomorphic image of)
an entire partially correct river:

\begin{lemma}\label{lem:reallylast}
  ~\\[-4mm]
\begin{itemize}
\item[(i)] If $F\in \chase_{T_0}(D)$ is a $\Sigma_\digamma$-fact, then
  there is a fact $F_0\in \chase_{T_1}()$ of the same relation
  such that there exists a homomorphism $h$ from
  $\mn{Ancestors}(F_0)\cup \mn{Ancestors}_{\querysig}(F_0) $ to
  $\chase_{T_0}(D)$ with $h(F_0)=F$.

\item[(ii)]  If $G\in \chase_{T_0}(D)$ is a fact of the relation
  $\mn{Encounter}$, then there exists a partially correct river
  $\mn{River}_\kappa$
  and a homomorphism $h$ from  $\mn{River}_\kappa$ to
  $\chase_{T_0}(D)$ such that $h(G_0)=G$, where $G_0$ is the
  $\mn{Encounter}$ fact of $\mn{River}_\kappa$.

\end{itemize}
\end{lemma}
\begin{proof}
(i) By induction of the number of steps of the
chase needed to create $F$. If it is zero, meaning that $F\in D$, then
$F$ must be a fact of the relation $\mn{Start}$ (which is the only
relation from $\Sigma_\digamma$ which is also in $\dbsig$), and the
claim holds true. If it is greater than zero, then $F $ was created by
chase using one of the $(\clubsuit)$ rules.  This required the fact $P$
in the body of the rule to be created first. Now apply the induction hypothesis
to~$P$.

(ii) $G$ can only be created in  $\chase_{T_0}(D)$ from some fact $F$
of relation $\mn{End}$. Now apply Claim~(i) to this $F$.
\end{proof}

Now recall that, in order to finish the proof, we only need to prove
that there is a database-preserving $\querysig$-homomorphism from
$\chase_{T_0}(D)$ to $D$. By Lemma~\ref{lem:lastone}, it
is enough to show that for each fact $G=\mn{Encounter}(c,c')$ in
$\chase_{T_0}(D)$, there exists a $\querysig$-homomorphism $h$ from
$T_\mn{myth}(\{G\})$ to $D$, such that $h(c)=c$ and $h(c')=c'$.  But
this easily follows from Point~(ii) of Lemma~\ref{lem:reallylast}, from the
assumption that $\digamma $ does not stop (and hence no partially
correct river is correct) and from Observation \ref{see}.

%
%
%
%

\section{Proofs for Section~\ref{sect:linear}}

For the proofs in this section, we need some knowledge about
the structure of the chase of a database with a set of linear TGDs.

\medskip

Let $T$ be a set of linear TGDs and $I$ an instance.  With every fact
$\alpha \in \adom(\chase_T(I))$, we associate a unique fact
$\src(\alpha) \in I$ that $\alpha$ was `derived from', as follows:
\begin{itemize}

\item if $\alpha \in I$, then $\src(\alpha)=\alpha$;

\item if $\alpha$ was introduced by applying a TGD from $T$, mapping
  the body of $T$ to a fact $\beta \in \chase_T(I)$, then
  $\src(\alpha)=\src(\beta)$.

\end{itemize}
We further associate, with every fact $\alpha \in \adom(I)$, the
subinstance $\chase_T(I)|^\downarrow_\alpha$ of $\chase_T(I)$ that
consists of all facts $\beta$ with $\src(\beta)=\alpha$. One should
think of $\chase_T(I)|^\downarrow_\alpha$ as the `tree-like instance'
that the chase of $I$ with $T$ generates `below $\alpha$'.  The
following lemma essentially says that the shape of
$\chase_T(I)|^\downarrow_\alpha$ only depends on $\alpha$, but not on
any other facts in $I$.

\begin{lemma}
  \label{lem:lineartrees}
  Let $I$ be an instance, $T$ a set of linear TGDs, and
  $\alpha \in I$. Then there is a homomorphism 
  from
  $\chase_T(I)|^\downarrow_{\alpha}$ to $\chase_T(\{\alpha\})$ that
  is the identity on all constants in $\alpha$. 
%
\end{lemma}
To prove Lemma~\ref{lem:lineartrees}, one considers a chase sequence
$I_0,I_1,\dots$ a for $I$ with $T$ and shows by induction on $i$ that
for all $i \geq 0$, there is a 
homorphism $h$ from
$I_i|^\downarrow_{\alpha}$ to $\chase_T(\{\alpha\})$ that is the
identity on all constants in $\alpha$. This is done by replicating the
application of the TGD that generated $I_i$ from $I_{i-1}$ in
$\chase_T(\{\alpha\})$. The 
homomorphism obtained in
the limit is as desired. 
Details are omitted.

\lemlinTGDsingleton*
\begin{proof}
  Since hom triviality and CQ triviality are equivalent, we may choose
  to work with hom triviality. The `only if' direction is immediate,
  so concentrate on the (contrapositive of the) `if' direction.

  Assume that $T$ is not $\Sigma_D,\Sigma_Q$-trivial. Then there is a
  $\Sigma_D$-database~$D$ 
  such that $\chase_{T}(D) \not\rightarrow_{\Sigma_Q} D$. If $D$ is
  empty, then we are done. To establish Point~2 it suffices to show
  that otherwise,
%
    there
  is a fact $\alpha \in D$ such that
  \mbox{$\chase_{T}(\{\alpha\}) \not\rightarrow_{\Sigma_Q} \{\alpha\}$}.
%
  %

  Assume to the contrary that there is no such \mbox{$\alpha \in D$}.  Then
  for every $\alpha \in D$, there is a $\Sigma_Q$-homomorphism $h_\alpha$ from
  $\chase_{T}(\{\alpha\})$ to $\{ \alpha \}$ that is the identity on
  all constants in~$\alpha$. By Lemma~\ref{lem:lineartrees}, there is
  a database-preserving homomorphism $g_\alpha$
  from
  $\chase_{T}(D)|^\downarrow_\alpha$ to $\chase_{T}(\{\alpha\})$. 
  Then $h= \bigcup_{\alpha \in D} h_\alpha \circ g_\alpha$ is a
  database-preserving $\Sigma_Q$-homomorphism from $\chase_{T}(D)$ to $D$, in
  contradiction to \mbox{$\chase_{T}(D) \not\rightarrow_{\Sigma_Q} D$}.
\end{proof}

\lemtwofacts*
\begin{proof}
  Let $D=\{ R(\bar c) \}$. If $\chase_T(D)$ contains a fact
  $S(\bar d)$ with $R \neq S \in \Sigma_Q$, then we may choose
  $C=\{S(\bar d) \}$.  Thus assume that the only relation symbol
  from $\Sigma_Q$ that
  occurs in $\chase_T(D)$ is $R$. Assume that for every
  connected database $C \subseteq \chase_T(D)$ that contains
  at most two facts, \mbox{$C \rightarrow_{\Sigma_Q} D$}.  We show that
  \mbox{$\chase_T(D) \rightarrow_{\Sigma_Q}
    D$}, that is, 
%
  %
  we have to construct a database-preserving $\Sigma_Q$-homomorphism
  $h$ from $I$ to $D$.

  Since we can clearly ignore facts in $I$ that use a relation symbol
  from outside of~$\Sigma_Q$, we only need to consider facts that use
  the relation symbol $R$. For each fact $\alpha = R(\bar d) \in I$,
  we have $\{\alpha \} \rightarrow_{\Sigma_Q} D$ and thus find a
  database-preserving homomorphism $h_\alpha$ from $\{ \alpha\} $ to
  $D$. We set $h = \bigcup_{\alpha \in I} h_\alpha$. To show that $h$
  is the desired database-preserving $\Sigma_Q$-homomorphism $h$ from
  $I$ to $D$, it suffices to show that $h$ is a function, that is, if
  $\alpha = R(\bar d) \in I$, $\beta = R(\bar e) \in I$, and $\bar d$
  and $\bar e$ share a constant $c$, then $h_\alpha(c)=h_\beta(c)$.
  We know that $\{ \alpha,\beta \} \rightarrow_{\Sigma_Q} D$. Take a
  witnessing homomorphism $h_{\alpha,\beta}$. Since every $R$-fact
  homomorphically maps \emph{in at most one way} into the single
  $R$-fact in $D$, the restriction og $h_{\alpha,\beta}$ to the
  variables in $\bar d$ is identical to $h_\alpha$, and likewise
  for $h_\beta$ and the variables in $\bar e$. Consequently,
  $h_\alpha(c)=h_\beta(c)$ as desired.
\end{proof}

\section{Proofs for Section~\ref{sect:frontierone}: Model Theory}

In this section, we provide the proofs of all model-theoretic results
from Section~\ref{sect:frontierone}, that is,
Theorem~\ref{thm:homchar}, Theorem~\ref{thm:cqcharone},
Theorem~\ref{thm:charsecond}, and Lemma~\ref{lem:limtoinfinite}. The
complexity upper bounds in Theorems~\ref{thm:frontierone}
and~\ref{thm:decidelimit} follow from the automata constructions in
the subsequent Section~F. We start with giving some auxiliary results. 

\bigskip
We first make explicit the structure of the chase for the case that $T$
is a set of frontier-one TGDs, in terms of tree-like databases. Note
that when a frontier-one TGD is applicable to a tuple
$\bar c$, then $\bar c$ is in fact a single constant. With every
$c \in \adom(\chase_T(I))$, we associate a unique constant
$\src(c) \in \adom(I)$ that $c$ was `derived from', as follows:
\begin{itemize}

\item $\src(c)=c$ for all $c \in \adom(I)$;
  
\item if $c$ is a null that was introduced by applying a TGD from $T$
  at $d$ then $\src(c)=\src(d)$.


\end{itemize}
We further associate, with every $c \in \adom(I)$, the subinstance
$\chase_T(I)|^\downarrow_a$ of $\chase_T(I)$ that is the restriction
of $I$ to constants $\{ d \in \adom(\chase_T(I)) \mid \src(d)=c \}$.

\medskip

\begin{lemma}
  \label{lem:chasenarrow}
  Let $T$ be a set of frontier-one TGDs of head width $\ell$. Then for
  every $c \in \adom(I)$, $\chase_T(I)|^\downarrow_c$ is a rooted tree-like
  instance of width at most $\ell$.
\end{lemma}
Informally, we can think of $\chase_T(I)$ as $I$ with rooted tree-like
instances of width at most $\ell$ attached to each constant.
We next define the unraveling of a database $D$ into a rooted tree-like
instance $U$ of width $k \geq 1$.  A $k$-\emph{sequence} takes the
form
$$v=S_0,c_{0},S_1,c_{1},S_2,\dots,S_{n-1},c_{n-1},S_n,$$
where each $S_i \subseteq \adom(D)$ satisfies $|S_i| \leq k$ and
$c_{i} \in S_i \cap S_{i+1}$ for $0 \leq i < n$. The empty
$k$-sequence is denoted by $\varepsilon$.  For every $c \in \adom(D)$,
reserve a countably infinite set of fresh constants that we refer to
as \emph{copies} of $c$. 

Now let $(V,E)$ be the infinite directed tree with $V$ the set of all
$k$-sequences and $E$ the prefix order on $V$. We choose
 a database $B(v)$ for every $v=S_{0}\cdots S_{n}\in V$, 
proceeding by induction on $n$:
\begin{enumerate}

\item $B(\varepsilon) = \emptyset$;

\item if $v=S_0$, then $B(v)$ is obtained from $D|_{S_0}$ by
  replacing every constant $c$ with a fresh copy of $c$;

\item if $v=S_{0}c_0\cdots c_{n-1}S_{n}$ with $n > 0$, then $B(v)$ is obtained
  from $D|_{S_0}$ by replacing
  \begin{itemize}

  \item $c_{n-1}$ with the copy of $c_{n-1}$ used in $B(v')$ where
    $v'=S_0 \cdots S_{n-1}$ is the predecessor of $v$ in $(V,E)$;

  \item every constant $c \neq c_{n-1}$with a fresh copy of $c$.

  \end{itemize}
\end{enumerate}
Set $\Tmc=(V,E,B)$ and $U = I_\Tmc$.
It is easy to see that the `uncopying'
map is a
homomorphism from $U$ to $D$.

\medskip

We next observe some properties of unraveled databases.
\begin{restatable}{lemma}{lemchaseone}
  \label{lem:chase1}
  Let $D$ be a database and $U$ its $k$-unraveling, $k \geq 1$, and
  let $T$ be a set of frontier-one TGDs with body width bounded by
  $k$. Then for every $c \in \mn{adom}(D)$ and copy $c'$ of $c$ in
  $U$, there is a homomorphism $h$ from
  $\mn{chase}_{T}(D)|^\downarrow_{c}$
  to $\mn{chase}_{T}(U)$ with $h(c)=c'$.
\end{restatable}

\begin{proof}
  Let $D$, $U$, $k$, and $T$ be as in the lemma.
  Let $I_0,I_1,\dots$ be a chase sequence for $D$ with $T$. The
  definition of \mn{src} extends to the instances $I_0,I_1,\dots$ in
  an obvious way and thus it is also clear what we mean by
  $I_i|^\downarrow_c$, for $i \geq 0$ and $c \in \mn{adom}(D)$.

  For all $i \geq 0$, $c \in \mn{adom}(D)$, and copies $c'$ of $c$ in
  $U$, we construct homomorphisms $h_{i,c,c'}$ from
  $I_i|^\downarrow_c$ to $\mn{chase}_{T}(U)$ with
  $h_{i,c,c'}(c)=c'$. Clearly, this suffices to prove the lemma
  because
  we obtain the desired homomorphism $h$ in the limit.

  
  The construction of the homomorphisms $h_{i,c,c'}$ proceeds by
  induction on $i$. The induction start is trivial as we may simply
  set $h_{i,c,c'}(c)=c'$ for every $c \in \mn{adom}(D)$ and copy $c'$
  of $c$ in $U$.  Now assume that $I_{i+1}$ was obtained from $I_i$ by
  applying a TGD
  $\vartheta = \phi(x,\bar y) \rightarrow \exists \bar z \,
  \psi(x,\bar z)$ from $T$ at some $d \in \mn{adom}(I_i)$. Let
  $\mn{src}(d)=c$. Then
  $I_{i+1}|^\downarrow_e = I_{i}|^\downarrow_e$ for all $e \in
  \mn{adom}(D) \setminus \{ c \}$, and thus the only homomorphisms
  that we need to take care of are $h_{i+1,c,c'}$ with $c'$ a copy of $c$ in $U$.
  Take any such $c'$.

  To apply $\vartheta$ at
  $d$, there must be a homomorphism $g$ from $\phi$ to $I_i$ with
  $g(x)=d$.  Let $d'=h_{i,c,c'}(d)$. We argue that there is also a homomorphism
  $g'$ from $\phi$ to $\mn{chase}_{T}(U)$ with $g'(x)=d'$. Let
  $S=(\mn{ran}(g) \cap \mn{adom}(D))$. By construction of $U$ and
  since $k$ is not smaller than then number of variables in $\phi$, we find
  an $S' \subseteq \mn{adom}(U)$ and an isomorphism $\iota$ from
  $D|_S$ to $U|_{S'}$ such that $\iota(c)=c'$ if $c \in S$. Moreover,
  the non-reflexive\footnote{A fact $R(c_1,\dots,c_n)$ is
    \emph{reflexive} if $c_1 = \cdots = c_n$.}  non-unary facts in
  $D|_S$ are identical to those in $I_i|_S$ because applying a
  frontier-one TGD can never add such facts. It follows that we can
  assemble the desired homomorphism $g'$ from $\iota$ and the
  homomorphisms $h_{i,e,e'}$ with $e \in S$ and $\iota(e)=e'$.

  We have just shown that $\vartheta$ is applicable in $\mn{chase}_{T}(U)$ at
  $d'$ or there is (already) a homomorphism $\widehat g$ from $\psi$
  to $\mn{chase}_{T}(U)$ with $\widehat g(x)=d'$. In either case, we
  can extend $h_{i,c,c'}$ to the desired homomorphism $h_{i+1,c,c'}$
  from $I_{i+1}|^\downarrow_c$ to $\mn{chase}_{T}(U)$ with
  $h_{i,c,c'}(c)=c'$ in an obvious way.
\end{proof}

\thmhomchar*

\begin{proof}
  The ``only if''-direction is immediate from the definition of
  hom-conservativity.

  For ``if'', assume that
  $T_1\not\models^{\textup{hom}}_{\dbsig,\querysig}T_2$. Then there is
  a $\dbsig$-database $D$ such that $\chase_{T_2}(D)
  \not\rightarrow_{\Sigma_Q} \chase_{T_1}(D)$. It suffices to show
  that $\mn{chase}_{T_2}(U) \not\rightarrow_{\Sigma_Q}
  \mn{chase}_{T_1}(U)$, $U$ the unraveling of $D$ of width $k$.  We
  prove the contrapositive.

  Thus assume that
  $\mn{chase}_{T_2}(U) \rightarrow_{\Sigma_Q} \mn{chase}_{T_1}(U)$.
  We have to show that
  $\mn{chase}_{T_2}(D) \rightarrow_{\Sigma_Q} \mn{chase}_{T_1}(D)$. We
  start with noting that, since $T_2$ is a set of frontier-one TGDs,
  $$
  \mn{chase}_{T_2}(D) = D \cup \bigcup_{c \in \mn{adom}(D')}
  \mn{chase}_{T_2}(D)|^\downarrow_c.
  $$
  As a consequence, it suffices to prove that
  $\mn{chase}_{T_2}(D)|^\downarrow_c \rightarrow_{\Sigma_Q} \mn{chase}_{T_1}(D)$
  for all $c \in \mn{adom}(D)$.

  Let $c \in \mn{adom}(D)$ and choose any copy $c'$ of $c$ in~$U$.  By
  Lemma~\ref{lem:chase1}, there is a homomorphism $h$ from
  $\mn{chase}_{T_2}(D)|^\downarrow_{c}$ to $\mn{chase}_{T_2}(U)$ with
  $h(c)=c'$. Together with
  $\mn{chase}_{T_2}(U) \rightarrow_{\Sigma_Q} \mn{chase}_{T_1}(U)$, this implies
  that there is a $\Sigma_Q$-homomorphism $h'$ from
  $\mn{chase}_{T_2}(D)|^\downarrow_{c}$ to $\mn{chase}_{T_1}(U)$ with
  $h'(c)=c'$. By construction of $U$, there is a homomorphism from $U$
  to $D$ that maps $c'$ to $c$. It is easy to extend this homomorphism
  to a homomorphism from $\mn{chase}_{T_1}(U)$ to
  $\mn{chase}_{T_1}(D)$ by following the application of chase rules.
  Thus   $\mn{chase}_{T_2}(D)|^\downarrow_c \rightarrow_{\Sigma_Q} \mn{chase}_{T_1}(D)$,
  as desired.

  It remains to show that there is a finite $U' \subseteq U$ with
  $\mn{chase}_{T_2}(U') \not\rightarrow_{\Sigma_Q}
  \mn{chase}_{T_1}(U')$. Since the TGDs in $T_2$ are frontier-one
  TGDs, an easy analysis of the chase procedure shows that
  $\mn{chase}_{T_2}(U) \not\rightarrow_{\Sigma_Q} \mn{chase}_{T_1}(U)$
  implies that, for some $c \in \mn{adom}(U)$,
  $\mn{chase}_{T_2}(U)|^\downarrow_c \not\rightarrow_{\Sigma_Q}
  \mn{chase}_{T_1}(U)$. It follows from Lemma~\ref{lem:typedet} that
  $\mn{chase}_{T_2}(U')|^\downarrow_c \not\rightarrow_{\Sigma_Q}
  \mn{chase}_{T_1}(U)$ for any $U' \subseteq U$ with
  $\mn{tp}_{T_2}(U,c)=\mn{tp}_{T_2}(U',c).$ By compactness of
  first-order logic, there is a finite such $U'$, as required.

\end{proof}

\thmcharone*
\begin{proof}     
  The ``only if''-direction is immediate from the definition of
  CQ-conservativity.

  \medskip
  For ``if'', assume that $T_1 \not\models^{\textup{CQ}}_{\dbsig,\querysig} T_2$.
  Then there is a $\dbsig$-database $D$ and a $\querysig$-CQ
  $q(\bar x)$ such that $q_{T_2}(D) \not\subseteq q_{T_1}(D)$. We
  first manipulate $q$ so that it has arity $0$ or $1$.

  We may assume w.l.o.g.\ that $q$ is connected because if it is not,
  then $p_{T_2}(D) \not\subseteq p_{T_1}(D)$ for some maximal
  connected component $p$ of $q$ and we can replace $q$ by $p$.
  Choose some $\bar c \in q_{T_2}(D) \setminus q_{T_1}(D)$ and let $h$
  be a homomorphism from $q$ to $\chase_{T_2}(D)$ such that
  $h(\bar x)=\bar c$. Let $q'(\bar x')$ be obtained from $q(\bar x)$
  in the following way:
  \begin{itemize}

  \item identify all variables $x_1,x_2 \in \var(q)$ in $q$ such that
    $h(x_1)=h(x_2)$;

  \item if $h(y) \in \adom(D)$ for some quantified variable $y$
    in $q$, then make $y$ an answer variable.

  \end{itemize}
  It is easy to see that $q'_{T_2}(D) \not\subseteq q'_{T_1}(D)$ and
  in fact
  $\bar c' := h(\bar x') \in q'_{T_2}(D) \setminus q'_{T_1}(D)$. Also
  note that $h$ is an injective homomorphism from $q'$ to
  $\chase_{T_2}(D)$.

%
   
  Let
  $C = \{ c \in \adom(D) \mid \exists x \in \var(q'):
  \src(h(x))=c\}$. For $c \in C$, let $q^c(\bar x^c)$ denote the restriction of
  $q$ to those variables $x$ such that $\src(h(x))=c$. The arity of
  each $q^c$ is 0 or 1 because $h$ maps all answers variables in $q^c$
  to $c$ and thus all such answer variables have been identified
  during the construction of $q'$. 
  By the following claim, we thus
  obtain a CQ $q$ of the required form
  by choosing one of the queries $q^c$.
  \\[2mm]
  {\it Claim.} There is a $c \in C$ such that
  $q^c_{T_2}(D) \not\subseteq q^c_{T_1}(D)$.
  \\[2mm]
  To prove the claim, assume to the contrary that
  $q^c_{T_2}(D) \subseteq q^c_{T_1}(D)$ for all $c \in C$.  Let
  $c \in C$. Since $h$ is a homomorphism from $q^c$ to
  $\chase_{T_2}(D)$, $h(\bar x^c) \in q^c_{T_2}(D)$ and thus
  $h(\bar x^c) \in q^c_{T_1}(D)$.  Consequently, there is a
  homomorphism $h_c$ from $q^c$ to $\chase_{T_1}(D)$ with
  $h_c(\bar x^c)=h(\bar x^c)$. Set $h' = \bigcup_{c \in C} h_c$ and
  note that $h'$ is functional since the queries $q^c$ do not share
  any variables (this is because $h$ is an injective homomorphism from
  $q'$ to $\chase_{T_2}(D)|^{\downarrow}_c$ and by construction of the
  queries $q^c$). By construction of $h'$, we have
  $h'(\bar x)=h(\bar x)=\bar c$. It thus remains to argue that $h'$ is
  a homomorphism from $q'$ to $\chase_{T_1}(D)$ as this contradicts
  $\bar c \notin q_{T_1}(D)$. Let $R(\bar z)$ be an atom in
  $q'$. First assume that there is a $c \in C$ such that
  $\src(h(z)) =c$ for all $z \in \bar z$. Then
  $h'(\bar z)=h_c(\bar z)$ and thus $R(h(\bar z)) \in
  \chase_{T_1}(D)$
  by definition of $h'$.  Now assume that there are $z_1,z_2$ in
  $\bar z$ with $\src(h(z_1)) \neq \src(h(z_2))$. Since the TGDs in
  $T$ are frontier-one, an easy analysis of the chase shows that this
  implies $R(h(\bar z)) \in D$, that is, the fact $R(h(\bar z))$ was
  in the original database as no such fact is ever added by the chase.
  Thus $h(z) \in \adom(D)$ for all variables $z$ in $\bar z$. By
  construction of $q'$, it follows that $\bar z$ consists only of
  answer variables. This implies $h'(\bar z)=h(\bar z)$ by definition
  of $h'$, and thus $R(h'(\bar z)) \in D \subseteq \chase_{T_1}(D)$.

  \medskip At this point, we know that $q$ is connected, of arity~0.
  We next argue that the
  database $D$ can be replaced by its $k$-unraveling $U$. We
  concentrate on the case that $q$ is unary. The Boolean case is very
  similar. We have to show that there is some $c' \in \mn{adom}(U)$
  such that $c' \in q_{T_2}(U)$, but $c' \notin q_{T_1}(U)$.

  By choice of $q$, there is a $c \in \mn{adom}(D)$ and a homomorphism
  $h$ from $q(x)$ to $\chase_{T_2}(D)|^{\downarrow}_c$ such that
  $h(x)=c$. Choose any copy $c'$ of $c$ in~$U$.  By
  Lemma~\ref{lem:chase1}, there is a homomorphism $g$ from
  $\mn{chase}_{T_2}(D)|^\downarrow_{c}$ to $\mn{chase}_{T_2}(U)$ with
  $g(c)=c'$. Composing $h$ and $g$, we obtain a homomorphism $h'$ from
  $q(x)$ to $\mn{chase}_{T_2}(U)$ with $h'(x)=c'$ and thus
  $c' \in q_{T_2}(U)$. It remains to show that $c' \notin
  q_{T_1}(U)$. But this follows from the facts that
  $c \notin q_{T_1}(D)$ and that there is a homomorphism from $U$ to
  $D$ that maps $c'$ to $c$, which easily extends to a homomorphism
  from $\mn{chase}_{T_1}(U)$ to $\mn{chase}_{T_1}(D)$.

  We may now finish the proof by argueing that for some finite $U'
  \subseteq U$, we have $q_{T_2}(U') \not\subseteq q_{T_1}(U')$.
  This, however, is a direct consequence of the compactness of
  first-order logic.
\end{proof}

For $c \in \adom(I)$ and $i \geq 0$, we use
  $I|^c_i$ to denote the restriction of $I$ to the constants that are
  reachable in the Gaifman graph of $I$ on a path of length at most $i$.
  Note that when we chase a finite database with a set of frontier-one
  TGDs, then the resulting instance has finite degree. This fails when
  frontier-one TGDs are replaced with guarded TGDs.

\begin{lemma}
\label{lem:skippinghoms}
  Let $I_1,I_2$ be instances of finite degree with $I_1$
  $\Sigma$-connected, for a schema~$\Sigma$. If there are
  $a_0 \in \adom(I_1)$ and $b_0 \in \adom(I_2)$ such that for each
  $i \geq 0$ there is a database-preserving $\Sigma$-homomorphism
  $h_i$ from $I_1|^{a_0}_i$ to $I_2$ with $h_i(a_0)=b_0$, then
  $I_1 \rightarrow_{\Sigma} I_2$.
\end{lemma}
\begin{proof} We are going to construct a database-preserving
  $\Sigma$-homomorphism $h$ from $I_1$ to $I_2$ step by step,
  obtaining in the limit a homomorphism that shows
  $I_1 \rightarrow_{\Sigma} I_2$.  We will take care that, at all
  times, the domain of $h$ is finite and
  \begin{itemize}

  \item[($*$)] there is a sequence $h_0,h_1,\dots$ with $h_i$ a
    database-preserving $\Sigma$-homo\-morphism from $I_1|^{a_0}_i$ to
    $I_2$ such that whenever $h(c)$ is already defined, then
    $h_i(c)=h(c)$ for all $i \geq 0$.

  \end{itemize}
  Start with setting $h(a_0)=b_0$. The original sequence of
  homomorphisms $h_0,h_1,\ldots$ from the lemma witnesses~($*$). Now
  consider the set $\Lambda$ that consists of all constants
  $c \in \adom(I_1)$ with $h(c)$ is undefined and such that there is a
  $d \in \adom(I_1)$ with $h(d)$ defined and that co-occurs with $c$
  in some $\Sigma$-fact in $I_1$.  Since the domain of $h$ is finite
  and $I_1$ has finite degree, $\Lambda$ is finite. By ($*$) and since
  $I_2$ has finite degree, for each $c \in \Lambda$, there are only
  finitely many $c'$ such that $h_i(c)=c'$ for some $i$. Thus, there
  must be a function $\delta:\Lambda \rightarrow \adom(I_2)$ such
  that, for infinitely many~$i$, we have $h_i(c)=\delta(c)$ for all
  $c \in \Lambda$. Extend $h$ accordingly, that is, set
  $h(c)=\delta(c)$ for all $c \in \Lambda$. Clearly, the sequence
  $h_0,h_1,\dots$ from ($*$) before the extension is no longer
  sufficient to witness ($*$) after the extension. We fix this by
  skipping homomorphisms that do not respect $\delta$, that is, define
  a new sequence $h'_0,h'_1,\dots$ by using as $h'_i$ the restriction
  of $h_j$ to the domain of $I_1|^{a_0}_i$ where $j\geq i$ is smallest
  such that $h_j(c)=\delta(d)$ for all $c \in \Lambda$.  This finishes
  the construction. The lemma follows from the fact that, due to the
  $\Sigma$-connectedness of $I_1$, every element is eventually
  reached. Note that $h$ is database-preserving since all the
  homomorphisms
  in the original sequence $h_0,h_1,\dots$ are.
\end{proof}

\thmcharsecond*
\begin{proof}
  ``$\Leftarrow$''. Assume that
  $T_1 \not\models^{\text{CQ}}_{\dbsig,\querysig} T_2$. By
  Theorem~\ref{thm:cqcharone}
  there is a tree-like
  $\dbsig$-databases $D$ of width at most $k$ and a connected
  $\Sigma_Q$-CQ $q$ of arity 0 or 1 such that
  $q_{T_2}(D) \not \subseteq q_{T_1}(D)$.

  First assume that $q(x)$ is of arity 1.  Then there is a constant
  $c \in \adom(D)$ such that $a \in q_{T_2}(D) \setminus
  q_{T_1}(D)$. Take a homomorphism $h$ from $q$ to
  $\chase_{T_2}(D)$ such that $h(x)=a$. Since $q$ is connected
  and uses only symbols from $\querysig$, $h$ is actually a
  homomorphism
  from $q$ to $\chase_{T_2}(D)|^\con_{\querysig}$.
  We show that
  $\chase_{T_2}(D)|^\con_{\querysig} \not\rightarrow_{\querysig}
    \chase_{T_1}(D)$ and thus Point~1 of
  Theorem~\ref{thm:charsecond} is violated. Assume to the contrary
  that $\chase_{T_2}(D)|^\con_{\querysig} \rightarrow_{\querysig}
    \chase_{T_1}(D)$ and take a
  witnessing $\querysig$-homomorphism $g$. Then
  $g \circ h$ is a homomorphism from $q$ to
  $\chase_{T_1}(D)$ that maps $x$ to $a$, implying $a \in q_{T_1}(D)$
  and thus a contradiction.

  Now assume that $q()$ is of arity 0 and take a homomorphism from $q$
  to $\chase_{T_2}(D)$. If the range of $h$ falls within
  $\chase_{T_2}(D)|^\con_{\querysig}$, then we can argue as above that
  Point~1 of Theorem~\ref{thm:charsecond} is violated. Thus assume
  that this is not the case. Then the connectedness of $q$ implies
  that the range of $h$ does not overlap with
  $\chase_{T_2}(D)|^\con_{\querysig}$. Moreover, there must be a
  maximally $\querysig$-connected component $I$ of
  $\chase_{T_2}(D) \setminus \chase_{T_2}(D)|^\con_{\querysig}$ such
  that the range of $h$ falls within $I$. We may again argue as above
  to show that $I \not\rightarrow^n_{\querysig} \chase_{T_1}(D)$, with
  $n$ the number of variables in $q$, as otherwise we find a
  homomorphism from $q$ to $\chase_{T_1}(D)$. This implies
  $I \not\rightarrow_{\querysig} \chase_{T_1}(D)$ and
  $I \not\rightarrow^{\lim}_{\querysig}
  \chase_{T_1}(D,c)|^\downarrow_c$ for all $c \in \adom(D)$. Thus
  Point~2 of Theorem~\ref{thm:charsecond} is violated.

  \medskip ``$\Rightarrow$''. Assume that
  $T_1 \models^{\text{CQ}}_{\dbsig,\querysig} T_2$ and let $D$ be a
  tree-like $\dbsig$-database of width at most~$k$. We have to show
  that Points~1 and~2 of Theorem~\ref{thm:charsecond} hold.

  \smallskip 
  
  We start with Point~1. By Theorem~\ref{thm:intermediatechar},
  $T_1 \models^{\text{CQ}}_{\dbsig,\querysig} T_2$ implies
  $\chase_{T_2}(D) \rightarrow^{\lim}_{\querysig} \chase_{T_1}(D)$.
  Let $I_1,\dots,I_k$ be the maximally connected components of
  $\chase_{T_2}(D)|^\con_{\querysig}$. It suffices to show that
  $I_i \rightarrow_{\querysig} \chase_{T_1}(D)$ for $1 \leq i \leq k$.
  Fix such an $i$. By definition of
  $\chase_{T_2}(D)|^\con_{\querysig}$, $I_i$ must contain some
  constant $c \in \adom(D)$.  Since
  $\chase_{T_2}(D) \rightarrow^{\lim}_{\querysig} \chase_{T_1}(D)$, we
  find a sequence $h_0,h_1,\dots$ where $h_\ell$ is a
  database-preserving $\querysig$-homomorphism from $I_i|^{c}_\ell$ to
  $\chase_{T_1}(D)$. In particular, $h_\ell(c)=c$ for all
  $\ell$. Thus, Lemma~\ref{lem:skippinghoms} yields
  $I_i \rightarrow_{\querysig} \chase_{T_1}(D)$. In summary, as
  required we obtain
  $\chase_{T_2}(D)|^\con_{\querysig} \rightarrow_{\querysig}
  \chase_{T_1}(D)$.

    \smallskip 

    Now for Point~2. Let $I$ be a maximally $\querysig$-connected
    component of
    $\chase_{T_2}(D) \setminus \chase_{T_2}(D)|^\con_{\querysig}$.
    Theorem~\ref{thm:intermediatechar} yields
    $I \rightarrow^{\lim}_{\querysig} \chase_{T_1}(D)$. Consequently,
    we find a sequence $h_0,h_1,\dots$ where $h_i$ is a
    $\querysig$-homomorphism from $I|^{c_0}_i$ to
    $\chase_{T_1}(D)$, for some $c_0 \in \adom(I)$. We distinguish
    two cases.

    First assume that there is a $d_0 \in \adom(\chase_{T_1}(D))$ such
    that $h_i(c_0)=d_0$ for infinitely many $i$. Construct a new
    sequence $h'_0,h'_1,\dots$ with $h'_i$ a $\querysig$-homomorphism
    from $I|^{c_0}_i$ to $\chase_{T_1}(D)$ by skipping homomorphisms
    that do not map $c_0$ to $d_0$, that is, $h'_i$ is the restriction
    of $h_j$ to the domain of $I|^{c_0}_i$ where $j\geq i$ is smallest
    such that $h_j(c_0)=d_0$.  Lemma~\ref{lem:skippinghoms} yields
    $I \rightarrow_{\querysig} \chase_{T_1}(D)$ and thus Point~2a is
    satisfied.

    Otherwise, there is no $d_0 \in \adom(\chase_{T_1}(D))$ such that
    $h_i(c_0)=d_0$ for infinitely many~$i$. We can assume that there
    is an $a_0 \in \adom(D)$ such that $\src(h_i(c_0))=a_0$ for all
    $i$; in fact, there must be an $a_0$ such that
    $\src(h_i(c_0))=a_0$ for infinitely many $i$ and we can again skip
    homomorphisms to achieve this for all $i$.  For brevity, let
    $J= \chase_{T_1}(D)|^\downarrow_{a_0}$. By
    Lemma~\ref{lem:chasenarrow}, $J$ is tree-like of width $\ell$,
    where $\ell$ is the head width of $T_1$. Thus, there is a rooted
    instance tree $\mathcal{T} = (V,E,B)$ of width $k$ that is
    finitely branching and satisfies $I_\Tmc=J$. Since there is no
    $d_0 \in \adom(\chase_{T_1}(D))$ such that $h_i(c_0)=d_0$ for
    infinitely many~$i$, it follows that for all $i,n \geq 0$ we must
    find a $j\geq i$ such that $h_j(c_0)$ is a domain element whose
    distance from $a_0$ in the Gaifman graph of $J$ exceeds~$n$.
    Based on this observation, we construct a sequence of
    homomorphisms $h'_0,h'_1,\dots$ as follows.  For all $i\geq 0$,
    let $h'_i$ be the restriction of $h'_j$ to the domain of
    $I|^{c_0}_i$ where $j \geq i$ is smallest such that the distance
    of $h_j(c_0)$ from $a_0$ exceeds $i$. Note that each $h_i'$ is a
    $\querysig$-homomorphism from $I|^{c_0}_i$ to $J$.  Since $I$ is
    connected, it is not hard to verify that this implies
    $I \rightarrow^{\lim}_{\querysig} J$. 
    Thus Point~2b is satisfied.
\end{proof}

The proof of the following lemma is somewhat technical. We recommend to 
read it with Example~\ref{ex:basic} in mind.

\lemlimtoinfinite*
\begin{proof}
  $(\Leftarrow)$ Assume that $I\rightarrow_{\Sigma} \widehat I$ for
  some $\widehat I\in \Rmc(T_1,\widehat t)$, that is,
  $I\rightarrow_\Sigma I_{\Tmc}$ for some
  $\widehat t$-proper $T_1$-labeled instance tree $\Tmc=(V,E,B,\mu)$. To show that
  $I \rightarrow^{\lim}_\Sigma \chase_{T_1}(\widehat t)$, it
  clearly suffices to prove that $I_{\Tmc}\rightarrow^{\lim}
  \chase_{T_1}(\widehat t)$.  

  Let $n \geq 1$ and $I'$ an induced subinstance of $I_\Tmc$ with
  $|\adom(I')| \leq n$. We have to show that
  $I' \rightarrow \chase_{T_1}(\widehat t)$.  Let $V'$ be the minimal
  subset of $V$ such that $v\in V'$ whenever $\adom(B(v)) \cap \adom(I') \neq \emptyset$
  and for \mbox{$E'=E \cap (V' \times V')$}, the graph $(V',E')$ is
  connected (and thus a tree). Let $\Tmc'=(V',E',B')$ with $B'$ the
  restriction of $B$ to $V'$. It is enough to prove that
  $I_{\Tmc'} \rightarrow \chase_{T_1}(\widehat t)$. 

  We start to construct a homomorphism $h$ from $I_{\Tmc'}$ to
  $\chase_{T_1}(\widehat t)$ as follows. Let $v$ be the root of
  $(V',E')$ or a non-root such that
  $\adom(B(v))\cap\adom(B(v')) = \emptyset$, $v'$ the predecessor of
  $v$. We know from Condition~1 of properness that $B(v)$ has the form
  $\{\mn{true}(c_0) \}$ and $\mu(c_0)=\widehat t$ or there is a TGD
  $\vartheta$ in $T_1$ such that $B(v)$ is isomorphic to
  the 
  head of $\vartheta$ and $\widehat t,T_1 \models q_{(B(v),\mu_v)}$.
  In both cases, we find a homomorphism $h_v$ from $D_{(B(v),\mu_v)}$
  to $\chase_{T_1}(\widehat t)$. The initial $h$ is the union of all
  the homomorphisms $h_v$.

  We now extend $h$ in a step-wise fashion. Let $(v,v') \in E'$ such
  that $h$ already covers $\adom(B'(v))$, but not $\adom(B'(v'))$.
  Then $h$ is a homomorphism from $D_{(B(v),\mu_{v})}$ to
  $\chase_{T_1}(\widehat t)$. Since $h$ does not yet cover
  $\adom(B'(v'))$, $\adom(B(v))\cap\adom(B(v')) \neq \emptyset$.  By
  Condition~2 of properness and because $\chase_{T_1}(\widehat t)$ is
  a model of $T_1$, we can extend $h$ to
  $\adom(D_{(B(v'),\mu_{v'})})$.

  \medskip $(\Rightarrow)$ Let $I$ be a countable $\Sigma$-connected
  instance 
  such that $I\rightarrow^{\lim}_\Sigma \chase_{T_1}(\widehat t)$, and
  let $\alpha_0, \alpha_1,\dots$
  be a (finite or infinite) enumeration of the non-unary facts in $I$ (we assume that
  there is at least one such fact). Consider the (finite or infinite)
  sequence of instances
  $$I_0 \subseteq I_1 \subseteq \cdots$$ with $I_i$ the restriction
  of $I$ to the constants in $\{ \alpha_0,\dots,\alpha_i \}$.  Since
  $I$ is $\Sigma$-connected, we may clearly choose
  $\alpha_0, \alpha_1,\dots$ so that $I_i$ is $\Sigma$-connected for
  all $i \geq 0$. 
  Once more since $I$ is $\Sigma$-connected (and thus there are no
  isolated unary facts),  \mbox{$\bigcup_{i\geq0} I_i = I$}.

  Since $I\rightarrow^{\lim}_\Sigma \chase_{T_1}(\widehat t)$ there is
  a sequence
  $$ h_0,h_1,\dots $$ with $h_i$ a homomorphism from
  $I_i$ to $\chase_{T_1}(\widehat t)$ for all  $i\geq 0$. 
  %
  %
  We have to identify a $\widehat t$-proper $T_1$-labeled instance tree $\widehat \Tmc$
  with $I\rightarrow I_{\widehat \Tmc}$. We do this by identifying a
  sequence
  $$
  \Tmc_0,\Tmc_1,\dots \text{ with } \Tmc_i =(V_i,E_i,B_i,\mu_i)
  $$
  of finite $\widehat t$-proper $T_1$-labeled instance trees that are monotonically growing
  in the sense that for all $i \geq 0$, $V_{i} \subseteq V_{i+1}$,
  $E_{i} \subseteq E_{i+1}$, and $B_{i}(v)=B_{i+1}(v)$ as well as
  $\mu_{i}(v)=\mu_{i+1}(v)$ for all $v \in V_i$. The desired instance
  tree $\widehat \Tmc$ is then obtained in the limit. In particular,
  each $\Tmc_i$ is constructed such that $I_i \rightarrow I_{\Tmc_i}$
  and along with the construction of the trees $\Tmc_i$ we construct
  a sequence of homomorphisms
  $$
  g_0,g_1,\dots
  $$
  witnessing this. Also this sequence is monotonically growing in the
  sense that $g_0 \subseteq g_1 \subseteq \cdots$ and we obtain the
  desired homomorphism $g$ from $I$ to $I_{\widehat \Tmc}$ in the
  limit.

  As a `guide' for the construction of the two
  sequences $\Tmc_0,\Tmc_1,\dots$ and $g_0,g_1,\dots$, we use the
  homomorphisms $h_i$ from $I_i$ to $\chase_{T_1}(\widehat{t})$.
  During the process, we also uniformize the sequence 
  $h_0,h_1,\dots$ be removing `unsuitable' homomorphisms from it, 
  similar to what has been done in the proof of 
  Lemma~\ref{lem:skippinghoms}. 
  For the construction, we make more precise the synchronization
  between the sequence $\Tmc_0,\Tmc_1,\dots$ and the sequence
  $h_0,h_1,\dots$.
  As described after the definition of $T_1$-labeled instance trees, the construction of $\chase_{T_1}(\widehat{t})$
  gives rise to a ($\widehat t$-proper) $T_1$-labeled instance tree
  $\Tmc = (V,E,B,\mu)$ such that $I_\Tmc = \chase_{T_1}(\widehat{t})$.
  Moreover, the width of \Tmc is bounded by the head width of TGDs
  in $T_1$.
  An \emph{embedding} of a $T_1$-labeled instance tree
  $\Tmc_i$ into $\Tmc$ is a pair of mappings
  $f:V_i\to V,\iota:\adom(I_{\Tmc_i})\to\adom(I_{\Tmc})$ such that $f$
  is an injective homomorphism from $(V_i,E_i)$ to $(V,E)$ and $\iota$
  satisfies the following conditions:
  \begin{enumerate}

    \item for every $v\in V_i$, the restriction of $\iota$ to
      $\adom(B_i(v))$ is an isomorphism from $B_i(v)$ to
      $B(f(v))$;

    \item for every $c\in \adom(I_{\Tmc_i})$, we have
      $\mu_i(c)=\mu(\iota(c))$.

  \end{enumerate}
  %
%
  We remark for further use that actually
  \begin{enumerate}

  \item[3.]  $\iota$ is an isomorphism from $I_{\Tmc_i}$ to
    $\bigcup_{u \in \mn{ran}} B(u)$.

  \end{enumerate}
  This is easy to show using the injectivity of $f$ and the definition
  of instance trees.
  
  Now, along with the sequences $\Tmc_0,\Tmc_1,\dots$ and
  $g_0,g_1,\dots$, we also construct embeddings
  $$
    f_{i,j}, \iota_{i,j} \text{ with } 0 \leq i \leq j
  $$
  where each $f_{i,j}, \iota_{i,j}$ is an embedding of $\Tmc_i$ into
  \Tmc. Note that there are infinitely many embeddings
  $f_{i,j}, \iota_{i,j}$ for each $i$ instead of only a single
  one. The reason is that these embeddings achieve a synchronization
  of each $\Tmc_i$ with the entire sequence $h_i,h_{i+1},\dots$ in
  the sense that, for $0 \leq i \leq j$, we shall take care that
  \begin{itemize}

  \item[$(\dagger)$] $h_j(c) = \iota_{i,j} \circ g_i(c)$ for all
    $c \in \mn{adom}(I_i)$.

%
%
%
%

  \end{itemize}
  Informally, $(\dagger)$ states that all homomorphisms $h_j$,
  $j\geq i$, map $I_i$ into $\chase_{T_1}(\widehat t)$ in the same way as
  $g_i$ maps $I_i$ into $\Tmc_i$.

  \smallskip 

  Now for the actual construction.  To define $\Tmc_0$ and~$g_0$,
  choose for every $k \geq 0$ a node $u_k \in V$ from~\Tmc such that
  $h_k(I_0) \subseteq B(u_k)$. Such a $u_k$ must exist since all constants in
  $I_0$ co-occur in a single fact in $I_0$. Consider the sequence
  $$(B_0,\lambda_0,\bar d_0), (B_1,\lambda_1,\bar d_1),\dots$$ with
  $(B_i,\lambda_i,\bar d_i)=(B(u_k),\mu(u_k),h_k(\bar c))$. Since the width of
  \Tmc is bounded, there are only finitely many isomorphism types of
  these triples, where $(B_i,\lambda_i,\bar d_i)$ and
  $(B_j,\lambda_j,\bar d_k)$
  are isomorphic if there is an isomorphism $\iota$ from $B_i$ to
  $B_j$ with $\iota(\bar d_i)=\bar d_j$ and
  $\lambda_i(c)=\lambda_j \circ \iota(c)$ for all $c \in \adom(B_i)$.
  Thus we may choose an isomorphism type that occurs infinitely
  often. We skip all homomorphisms $h_i$ such that
  $(B_i,\lambda_i,\bar d_i)$
  is not of that type, that is, we replace each $h_i$ with $h_j$ where
  $j\geq i$ is minimal such that $(B_j,\lambda_j,\bar d_j)$ is of the chosen
  isomorphism type.  Now define $\Tmc_0$ by taking
  \[V_0 = \{v_0\}, \quad E_0=\emptyset, \quad B_0(v_0) = B_0,\quad \mu_0
  = \lambda_0, \]
and set $g_0(c) = h_0(c)$ for all constants $c \in \adom(I_0)$ and
$f_{0,j}(v_0)=u_{k}$ for all $j \geq 0$. As $\iota_{0,j}$, we use the
isomorphism that witness that $(B_0,\lambda_0,\bar d_0)$ and
$(B_j,\lambda_j,\bar d_j)$ have the same isomorphism type. Based on
this choice, it can be verified that ($\dagger$) is satisfied.

  \smallskip For the inductive step, assume that we have already
  constructed $\Tmc_i$, $g_i$, as well as the embeddings
  $f_{i,j},\iota_{i,j}$ for all $j \geq i$.
  We obtain $\Tmc_{i+1},
  g_{i+1}$ from $\Tmc_i,g_i$ by starting with $\Tmc_{i+1} = \Tmc_i$
  and $g_{i+1} = g_i$ and then extending as follows. 

  By construction of $I_{i+1}$, there is a non-unary fact
  $R(\bar c) \in I_{i+1}$ such that $I_{i+1}$ is the restriction of
  $I$ to $\adom(I_i) \cup \bar c$. For every $k > i$, there is thus a
  $u_k\in V$ with $R(h_{k}(\bar c))\in B(u)$. In fact, each $u_k$ is
  unique by definition of instance trees. By skipping homomorphisms
  from $h_{i+1},h_{i+2},\dots$ along with the associated functions
  $f_{i,i+1},f_{i,i+2},\dots$ and isomorphisms
  $\iota_{i,i+1},\iota_{i,i+2},\dots$, we can achieve that one of the
  following two cases applies:
  \begin{enumerate}

    \item $u_k$ is in the range of $f_{i,k}$ for all $k>i$, or

    \item $u_k$ is not in the range of $f_{i,k}$ for all $k>i$.

  \end{enumerate}
  In Case~1, since $V_i$ is finite we can once more
  skip homomorphisms and
  achieve that there is some $v\in V_i$ such that $f_{i,k}(v)=u_k$, for
  all $k>i$. For every $k>i$, by Property~3 of embeddings we
  may define
  \[\bar d_k = \iota_{i,k}^-(h_k(\bar c)).\] 
  %
  The choice of $v$ and Property~1 of embeddings yield
  $\bar d_k\subseteq \adom(B_i(v))$ for all $k > i$.  Since
  $\adom(B_i(v))$ is finite, there are only finitely many possible
  choices for the $\bar d_k$. By skipping homomorphism, we may thus
  achieve that they are all identical. Extend $g_{i+1}$ by setting
  $g_{i+1}(c) = \iota_{i,k}^-(h_k(c))$ for some (equivalently: all)
  $k> i$, for every $c\in \bar c$ such that $g_{i+1}(c)$ is not yet
  defined. Define $f_{i+1,j}=f_{i,j}$ and $\iota_{i+1,j}=\iota_{i,j}$
  for all $j > i$.  One may verify that~$(\dagger)$ is satisfied.

  We argue that, as required, $g_{i+1}$ is a homomorphism from
  $I_{i+1}$ to $I_{\Tmc_{i+1}}$. Take any fact
  $S(\bar e) \in I_{i+1}$. 
  The definition of $g_{i+1}$ and
  ($\dagger$) yield $g_{i+1}(c) = \iota_{i,i+1}^-(h_{i+1}(c))$
  for all $c \in \bar e$. It remains to note that $h_{i+1}$ is
  a homomorphism from $I_{i+1}$ to $\mn{chase}_{T_1}(\widehat t)$,
  thus $S(h_{i+1}) \in \mn{chase}_{T_1}(\widehat t)$, and
  $\iota_{i,i+1}^-$ is an isomorphism.

  \smallskip
  
  We consider now Case~2, that is, $u_k$ is not in the range of $f_{i,k}$,
  for all $k>i$. Now, observe that since $I_{i+1}$ is connected, there
  is some $d\in \bar c\cap \adom(I_i)$, and thus $g_i(d)$ is already
  defined. Choose $v\in V_i$ with $g_i(d)\in\adom(B(v))$ and consider
  the sequence
  \[w_k = f_{i,k}(v),\quad k>i.\]
  Since $u_k$ is not in the range of $f_{i,k}$, we have $u_k\neq w_k$, for
  all $k$. However, $h_k(d)\in \adom(B(w_k))\cap \adom(B(u_k))$. 
  Due to Property~1 from the definition of instance trees, we can
skip homomorphisms to reach one of the following situations:
  \begin{enumerate}[label=(\alph*)]

    \item $u_k$ is the predecessor of $w_k$, for all $k>i$, 

    \item $u_k$ is a successor of $w_k$, for all $k>i$, or

    \item $u_k,w_k$ are siblings, for all $k>i$.

  \end{enumerate}
  In all cases, we start as follows. Similarly to the induction start,
  consider the sequence $$(B_{i+1},\lambda_{i+1},\bar d_{i+1}),
  (B_{i+2},\lambda_{i+2},\bar d_{i+2}),\dots$$ with $(B_k,\lambda_k,\bar
  d_k)=(B(u_k),\mu(u_k),h_k(\bar c))$, for all $k>i$. Since the width of \Tmc is
  bounded, there are only finitely many isomorphism types of these
  triples.  Thus we may choose an isomorphism type that occurs
  infinitely often. By skipping homomorphisms, we can achieve that all
  $(B_k,\lambda_k,\bar d_k)$ are of the same type. We further select a
  triple $(B,\lambda,\bar d)$ that is of the same isomorphism type to
  be used as a bag in the tree $\Tmc_{i+1}$. We make this choice such
  that $g_i(d)\in \adom(B)$ and $\adom(B)\setminus\{g_i(d)\}$ consists
  only of fresh constants, that is, constants not used in
  $I_{\Tmc_i}$.  Define $\tau_{k}$ to be an isomorphism that witnesses
  that $(B,\lambda,\bar d)$ and $(B_k,\lambda_k,\bar d_k)$ have the
  same isomorphism type, for all $k > i$.

  We now extend $\Tmc_{i+1},g_{i+1}$ distinguishing Cases~(a)--(c). 

  \smallskip In Case~(a), let us first argue that $v$ is the root of
  $\Tmc_i$. If not, then $f_{i,k}$ maps the predecessor $v'$ to $u_k$,
  for all $k>i$, in contradiction to the fact that $u_k$ is not in the
  range of $f_{i,k}$. Then, add a predecessor $v'$ of $v$ to
  $\Tmc_{i+1}$, set
  \[B_{i+1}(v') = B,\quad \mu_{i+1} = \mu_{i+1}\cup \lambda, \]
  and set, for all $c\in \adom(I_{i+1})$ such that $g_{i+1}(c)$ is not
  yet defined, $g_{i+1}(c) = \iota_{i+1,k}^-(h_k(c))$, for some
  (equivalently: all) $k>i$. It can be verified that setting, for all
  $j>i$, 
  \begin{align*}
    f_{i+1,j} & =f_{i,j}\cup \{(v',u_j)\}, \text{ and} \\
    \iota_{i+1,j} & = \iota_{i,j}\cup \tau_j
  \end{align*}
  witnesses $(\dagger)$ for $i+1$.

  \smallskip In Case~(b), we do exactly the same as in Case~(a) with
  $v'$ a fresh \emph{successor} of $v$ (instead of predecessor). 
  
  \smallskip In Case~(c), we make a final case distinction. If
  $v$ has a predecessor $v_0$ in $\Tmc_i$, then proceed exactly as in
  Case~(a), but make $v'$ a fresh successor of $v_0$. Otherwise, let
  $u_k'$ be the predecessor of $u_k$, for all $k>i$, and 
%
%
%
  consider the sequence $$(B_{i+1}',\lambda_{i+1}',d_{i+1}),
  (B_{i+2}',\lambda_{i+2}', d_{i+2}),\dots$$ with $(B_k',\lambda_k',
  d_k)=(B(u_k'),\mu(u_k'),h_k(d))$, for all $k>i$ (recall that we
  fixed $d$ in the beginning of Case~2). We can again
  skip homomorphisms and achieve that all the $(B_k',\lambda_k',
  d_k)$ are of the same isomorphism type. We further select a
  triple $(B',\lambda',d')$ that is of the same isomorphism type to
  be used as a bag in the tree $\Tmc_{i+1}$. We make this choice such
  that $g_i(d)\in \adom(B')$ and $\adom(B')\setminus\{g_i(d)\}$ consists
  only of fresh constants, that is, constants not used in
  $I_{\Tmc_i}$.  Define $\tau_{k}'$ to be an isomorphism that witnesses
  that $(B',\lambda',d')$ and $(B_k',\lambda_k',d_k)$ have the
  same isomorphism type, for all $k > i$. Then, add a predecessor
  $v'$ of $v$ and a fresh successor $v''$ of $v'$ to $\Tmc_{i+1}$, set 
  \[B_{i+1}(v') = B', \quad B_{i+1}(v'') = B, \quad \mu_{i+1} =
  \mu_{i+1}\cup \lambda\cup \lambda', \]
  and set, for all $c\in \adom(I_{i+1})$ such that $g_{i+1}(c)$ is not
  yet defined, $g_{i+1}(c) = \iota_{i+1,k}^-(h_k(c))$, for some
  (equivalently: all) $k>i$. It can be verified that setting, for all
  $j>i$, 
  \begin{align*}
    f_{i+1,j} & =f_{i,j}\cup \{(v',u_j'),(v'',u_j)\}, \text{ and} \\
    \iota_{i+1,j} & = \iota_{i,j}\cup \tau_j\cup \tau_j'
  \end{align*}
  witnesses $(\dagger)$ for $i+1$.

  \medskip This finishes the construction of the sequences
  $\Tmc_0,\Tmc_1,\dots$ and $g_0,g_1,\dots$. Recall that both the 
  $\Tmc_i$ and the $g_i$ are monotonically growing and that we are interested in the
  limits $\widehat \Tmc$ and $g$ of the sequences, that is, \[\widehat \Tmc =
  \left(\bigcup_{i\geq 0} V_i,\bigcup_{i\geq 0} E_i,\bigcup_{i\geq
      0}B_i, \bigcup_{i\geq 0}\mu_i\right)  \] and
$$
g = \bigcup_{i \geq 0} g_i.
$$
Since each $g_i$ is a homomorphism from $I_i$ to $I_{\Tmc_i}$, it is
clear that $g$ is a homomorphism from $I$ to $I_{\widehat \Tmc}$.
Moreover, $\widehat \Tmc$ is $\widehat t$-proper since each $\Tmc_i$ is.
\end{proof}

\section{Proofs for Section~\ref{sect:frontierone}: Decision Procedures}

We prove the decidability results from Section~\ref{sect:frontierone}
using the characterizations provided in that section and tree
automata. More precisely, to prove the \ThreeExpTime upper bounds for hom-conservativity and
CQ-conservativity in Theorem~\ref{thm:frontierone}, we show how to
construct, given sets $T_1,T_2$ of frontier-one TGDs and signatures
\dbsig and \querysig, a tree automaton \Amf such that $L(\Amf) \neq \emptyset$
iff $T_1 \not\models_{\dbsig,\querysig}^{\textup{hom}} T_2$ resp.\
$T_1 \not\models_{\dbsig,\querysig}^{\textup{CQ}} T_2$. The use of
tree automata is sanctioned by the characterizations of
hom-conservativity and CQ-conservativity in terms of tree-shaped
witnesses provided by Theorem~\ref{thm:homchar} and
Theorem~\ref{thm:charsecond}.

We start with giving the necessary details on tree automata.

\subsection{Tree Automata}

A \emph{tree} is a non-empty (and potentially infinite) set of words $W \subseteq (\Nbbm
\setminus 0)^*$ closed under prefixes.  
We assume that trees
are finitely branching, that is, for every $w \in W$, the set $\{ i >0
\mid w \cdot i \in W \}$ is finite.
For $w \in (\Nbbm \setminus 0)^*$, set $w \cdot
0 := w$. For $w=n_0n_1 \cdots n_k$, $k>0$, set $w
\cdot -1 := n_0 \cdots n_{k-1}$, and call $w$ a \emph{successor} of
$w\cdot -1$ and $w\cdot -1$ a \emph{predecessor} of $w$.  For an alphabet $\Theta$, a
\emph{$\Theta$-labeled tree} is a pair $(W,L)$ with $W$ a tree and
$L:W \rightarrow \Theta$ a node labeling function.

A \emph{two-way alternating tree automaton (\ata)} is a tuple
$\Amf = (Q,\Theta,q_0,\delta,\Theta)$ where $Q$ is a finite set of
{\em states}, $\Theta$ is the {\em input alphabet}, $q_0\in Q$ is the
{\em initial state}, $\delta$ is a {\em transition function}, and
$\Theta:Q\to \mathbb{N}$ is a {\em priority function}.  The transition
function $\delta$ maps every state $q$ and input letter $a \in \Theta$
to a positive Boolean formula $\delta(q,a)$ over the truth constants
$\mn{true}$ and $\mn{false}$ and \emph{transition atoms} of the form
$q$, $\Diamond^- q$, $\Box^- q$, $\Diamond q$ and $\Box q$.  A
transition $q$ expresses that a copy of \Amf is sent to the current
node in state $q$; $\Diamond^- q$ means that a copy is sent in state
$q$ to the predecessor node, which is required to exist; $\Box^- q$
means the same except that the predecessor node is not required to
exist; $\Diamond q$ means that a copy of $q$ is sent to some successor
and $\Box q$ means that a copy of $q$ is sent to all successors.  The
semantics of \ata is given in terms of runs as usual.

Let $(W,L)$ be a $\Theta$-labeled tree and 
$\Amf=(Q,\Theta,q_0,\delta,\Omega)$ a \ata. A {\em run of \Amf over 
  $(W,L)$} is a $W\times Q$-labeled tree $(W_r,r)$ such that 
$\varepsilon\in W_r$, $r(\varepsilon)=(\varepsilon,q_0)$, and for all 
$y\in W_r$ with $r(y)=(x,q)$ and $\delta(q,V(x))=\theta$, there is an 
assignment $v$ of truth values to the transition atoms in $\theta$
such that $v$ satisfies $\theta$ and:
    \begin{itemize}

      \item if $v(q')=1$, then $r(y')=(x,q')$ for some successor 
        $y'$ of $y$ in $W_r$;

      \item if $v(\Diamond^- q')=1$, then $x \neq \varepsilon$ and 
        $r(y')=(x\cdot -1,q')$ for some successor $y'$ of $y$ in $W_r$;

      \item if $v(\Box^- q')=1$, then $x=\varepsilon$ or 
        $r(y')=(x\cdot -1,q')$ for some successor $y'$ of $y$ in $W_r$;

      \item if $v(\Diamond q')=1$, then there is some $j$
        and a successor $y'$ of $y$ in $W_r$ with 
        $r(y')=(x\cdot j,q')$;

      \item if $v(\Box q')=1$, then for all successors $x'$ of $x$, there is a 
        successor $y'$ of $y$ in $W_r$ with 
        $r(y')=(x',q')$. 


\end{itemize}
Let $\gamma=i_0i_1\cdots$ be an infinite path in $W_r$ and denote, for
all $j\geq 0$, with $q_j$ the state such that $r(i_j)=(x,q_j)$. The
path $\gamma$ is {\em accepting} if the largest number $m$ such that
$\Omega(q_j)=m$ for infinitely many $j$ is even.  A run $(W_r,r)$ is
accepting, if all infinite paths in $W_r$ are accepting.  \Amf accepts
a tree if \Amf has an accepting run over it.  We use $L(\Amf)$ to
denote the set of $\Theta$-labeled trees accepted by $\Amf$. 

It is not hard to show that \ata are closed under intersection
and that the intersection automaton can be constructed in polynomial
time, see for example \cite{tata2007}. 
The \emph{emptiness problem} for \ata
means to decide, given a \ata~\Amf, whether
$L(\Amf)=\emptyset$. Emptiness of 2ATA can be solved in time
single exponential in the number of states and the maximal priority, and
polynomial in all other components. This was proved for 2ATAs on
ranked trees in~\cite{Vardi98} and it was shown in
\cite{DBLP:journals/jair/JungLMS20} that the result carries over to
the particular version of 2ATAs used here, which run on trees of
arbitrary finite degree.

\subsection{Upper Bound for Hom-Conservativity}

To decide hom-conservativity via Theorem~\ref{thm:homchar} it suffices
to devise a \ata \Amf such that
\begin{itemize}

  \item[$(\ast_\Amf)$] $\Amf$ accepts all tree-like instances $I$ of width
    $\max(k,\ell)$ that are models of $T_1$ and some tree-like
    $\querysig$-databases $D$ of width $k$ such that
    $\chase_{T_2}(D)\not\rightarrow_{\querysig} I$, where $k$ and
    $\ell$ are the body and head width of $T_1$. 

\end{itemize}
However, 2ATAs cannot run directly on tree-like databases or instances
because the potential labels of the underlying trees (the bags) may
use any number of constants and do not constitute a finite alphabet.
We therefore use an appropriate encoding of
tree-like databases that reuses constants so that we can make do with
finitely many constants overall, similar to what has been done, for
example, in~\cite{GradelW99}.

\paragraph{Encoding of tree-like instances.} Let $m=\max(k,\ell)$ with $k$ the
body width and $\ell$ the head width of $T_1$. Fix a set $\Delta$ of
$2m$ constants and define $\Theta_0$ to be the set of all
$\fullsig$-databases $B$ with
\mbox{$\adom(B)\subseteq \Delta$} and $|\adom(B)|\leq m$, where $\fullsig$ is
the union of $\dbsig$ and~$\mn{sig}(T_1)$, that is, all relation
symbols that occur in~$T_1$.

Let $(W,L)$ be a $\Theta_0$-labeled tree. For convenience, we use
$B_w$ to refer to the database $L(w)$ at node $w$. For a constant
$c\in \Delta$, we say that $v,w\in W$ are \emph{$c$-equivalent} if
$c\in\adom(B_u)$ for all $u$ on the unique shortest path from $v$
to~$w$. Informally, this means that $c$ represents the same constant in
$B_v$ and in $B_w$. In case that $c\in \adom(B_w)$, we use $[w]_c$ to
denote the set of all $v$ that are $c$-equivalent to~$w$.  We call
$(W,L)$ \emph{well-formed} if it satisfies the following counterparts
of Conditions~1 and~2 of instance trees:
\begin{enumerate}[label=\arabic*$'$.]

  \item for every $w\in W$ and every $c\in\adom(B_w)$, the restriction
    of $W$ to $[w]_c$ is a tree of depth at most $1$;

  \item for every $w\in W$ and successor $v$ of $w$, $\adom(B_w)\cap
    \adom(B_v)$ contains at most one constant.

\end{enumerate}
Each well-formed $\Theta_0$-labeled tree $(W,L)$ \emph{represents}
a $\fullsig$-instance tree $\Tmc_{W,L}=(V,E,B)$ as follows. 
The underlying tree $(V,E)$ is the tree (described by) $W$. 
The active domain of $I_{\Tmc_{W,L}}$ is the
set of all equivalence classes $[w]_c$ with $w\in W$ and
$c\in\adom(B_w)$ and the labeling $B$ is defined by taking 
\[R([w]_{c_1},\ldots,[w]_{c_k})\in B(w)\quad \text{ iff \quad
}R(c_1,\ldots,c_k)\in B_w,\]
for all $w\in W$ and $c\in \adom(B_w)$. As a shorthand, we use
$I_{W,L}$ to denote the instance $I_{\Tmc_{W,L}}$.

Conversely, for every $\dbsig$-instance $I$ such that $I=I_\Tmc$ for a
instance tree $\Tmc=(V,E,B)$ of width $m$, we can find a
$\Theta_0$-labeled tree $(W,L)$ that represents $I$ in the sense that
$I_{W,L}$ is isomorphic to $I$.  Since $\Delta$ is of size $2m$, it is possible
to select a mapping $\pi:\adom(D)\to \Delta$ such that for each
$(v,w)\in E$ and each $d,e\in \adom(B(w))\cup\adom(B(v))$, we have
$\pi(d) = \pi(e)$ iff $d=e$.  Define the $\Theta_0$-labeled tree
$(W,L)$ by setting $W=(V,E)$, and for all $w\in W$, $B_w$ to the image
of $B(w)$ under $\pi$. Clearly, $(W,L)$ satisfies the desired
properties.

\paragraph{Automata Constructions}

We construct a \ata~\Amf that satisfies~$(\ast_\Amf)$, for given
$T_1,T_2,\dbsig,\querysig$. We may assume without loss of generality
that all symbols from \dbsig and \querysig occur in $T_1$. The desired
{\ata} runs over $\Theta$-labeled trees with
$\Theta=\Theta_0\times \Theta_0\times \Theta_1$ where $\Theta_0$ is
defined as above, and $\Theta_1$ is the set of all mappings
$\mu:\Delta'\to \TP(T_2)$ for some $\Delta'\subseteq \Delta$ with
$|\Delta'|\leq m$. Intuitively, the first component will represent a
$\dbsig$-database~$D$, the second component will represent a model $I$
of $T_1$ and~$D$, and the last component will represent the $T_2$
chase of $D$, restricted to $\adom(D)$.

For a $\Theta$-labeled tree $(W,L)$, we set
$L(w)=(L_0(w),L_1(w),L_2(w))$ for all $w\in W$ and thus may use
$L_i(w)$ to refer to the $i$-th component of the label of $w$, for $i
\in \{0,1,2\}$. For the sake of readability, we may use $\mu_w$ to
denote $L_2(w)$. A $\Theta$-labeled tree $(W,L)$ is called
\emph{well-typed} if, for all $w\in W$: 
\begin{enumerate}

  \item the domain of $\mu_w$ is $\adom(L_0(w))$ and 

  \item for every successor $v$ of $w$ and every $c\in
    \adom(L_0(w))\cap\adom(L_0(v))$, we have $\mu_w(c)=\mu_v(c)$.

\end{enumerate}
The desired \ata \Amf is constructed as the intersection of the five
{\ata}s $\Amf_0,\Amf_1,\Amf_2,\Amf_3,\Amf_4$ provided by the following lemma.
\begin{lemma} \label{lem:hom-automata}
  There are {\ata}s $\Amf_0,\Amf_1,\Amf_2,\Amf_3,\Amf_4$ such that:
  \begin{itemize}

    \item[--] $\Amf_0$ accepts $(W,L)$ iff it is well-typed and
      $(W,L_0)$ and $(W,L_1)$ are
      well-formed;

    \item[--] $\Amf_1$ accepts $(W,L)$ iff $I_{W,L_0}$ is a
      $\dbsig$-database of width~$k$; 

    \item[--] $\Amf_2$ accepts $(W,L)$ iff $I_{W,L_1}$ is a model of
      $I_{W,L_0}$ and $T_1$;

    \item[--] $\Amf_3$ accepts $(W,L)$ iff 
      for every $w\in W$ and every $c\in \adom(L_0(w))$, 
      \[\mu_w(c) = \tp_{T_2}(\chase_{T_2}(I_{W,L_0}),[w]_c).\]

    \item[--] $\Amf_4$ accepts $(W,L)$ iff
      $\mn{chase}_{T_2}(I_{W,L_1}) \not\rightarrow I_{W,L_1}$.

  \end{itemize}
  The number of states of
  \begin{itemize}

    \item[--] $\Amf_0$ is exponential in $||T_1||$ (and independent of
      $T_2$);

    \item[--] $\Amf_1$ does not depend on the input;

    \item[--] $\Amf_2$ is exponential in $||T_1||$ (and independent of
      $T_2$);

    \item[--] $\Amf_3$ is exponential in $||T_2||$ (and independent of~$T_1$);

      \item[--] $\Amf_4$ is double exponential in $||T_2||$ (and independent
	of $T_1$).

    \end{itemize}
    All automata can be constructed in time triple exponential in
    $||T_1||+||T_2||$
    and have
    maximum priority one. 
  %
\end{lemma}
It can be verified that \Amc satisfies ($\ast_\Amf$) and thus
$L(\Amf)\neq\emptyset$ iff
$T_1 \not\models_{\dbsig,\querysig}^{\textup{hom}} T_2$.
The rest of
this section is devoted to proving Lemma~\ref{lem:hom-automata}.

\subsubsection*{Automaton $\Amf_0$.} This automaton is straightforward
to construct.

\subsubsection*{Automaton $\Amf_1$.} This automaton simply verifies
that all databases $L_0(w)$ use only symbols from $\dbsig$ and at most
$k$ constants, and that on every path there are only finitely many
non-empty databases.  Constantly many states suffice for this purpose.

\subsubsection*{Automaton $\Amf_2$.} First note that $I_{W,L_1}$ is a
model of $I_{W,L_0}$ iff $L_0(w)$ is a subset of $L_1(w)$, for every
$w\in W$. This check can easily be done by a \ata with constantly many
states. In order to verify that $I_{W,L_1}$ is a model of $T_1$, it is
essential to realize that the employed encoding allows a \ata to do the
following:
\begin{itemize}

  \item[$(\dagger)$] given some $w\in W$ and $c\in \adom(L_1(w))$, and
    a unary CQ $q(x)$, verify that there is a homomorphism
    $h$ from $q$ to $I_{W,L_1}$ with $h(x)=[w]_c$. 

\end{itemize}
Since (parts of) \ata{}s can easily be complemented by dualization,
they are also able to verify that there is no such homomorphism.  The
\ata $\Amf_2$ may thus visit all $w\in W$ and all
$c\in \adom(L_1(w))$ and verify that, for every TGD
$\phi(x,\bar y)\to \exists \bar z\, \psi(x,\bar z)$ in $T_1$, there is
no homomorphism $h$ from $q_\phi(x)$ to $I_{W,L_1}$ with $h(x)=[w]_c$
or there is a homomorphism $g$ from $q_\psi(x)$ to $I_{W,L_1}$ with
$g(x)=[w]_c$.

Informally, a \ata can achieve $(\dagger)$ by memorizing (in its
state) a CQ $p$ for which it still has to check the existence of a
homomorphism, plus the target constant of the free variable of $p$ (if
any). If the automaton visits a given node $w\in W$ in such a state,
it guesses the variables $y_1,\ldots,y_n$ that the homomorphism will
map to $\adom(L_1(w))$ and also the corresponding homomorphism targets
$e_1,\ldots,e_n \in \adom(L_1(w))$.  It verifies that the guess indeed
give rise to a partial homomorphism to database $L_1(w)$ and proceeds
with the parts of $p$ that have not been mapped to the current
database $L_1(w)$.

To formalize this idea, we use \emph{instantiated CQs} in which all answer
variables are replaced with constants, writing $q(\bar c)$ to indicate
that $\bar c$ are precisely the constants that occur in $q$ and that
all variables are quantified.  We will mostly drop the word
`instantiated' and only speak of CQs.  

Let $q(\bar c)$ be an (instantiated) CQ.  A \emph{$\Delta$-splitting}
of $q(\bar c)$ is obtained by first replacing any number of variables
in $q(\bar c)$ with constants\footnote{Different variables may be
  replaced
  with the same constant.} from $\Delta$ and then partitioning the
(atoms of the) resulting CQ into CQs
$q_0(\bar c_0),q_1(\bar c_1),\ldots,q_n(\bar c_n)$ such
that:
\begin{enumerate}

\item $q_0$ has no quantified variables;

  \item for all $i>0$, $\bar c_i$ is empty or a single constant
    from $\bar c_0$;

  \item for all $j>i>0$, $q_i$ and $q_j$ share no variables.

\end{enumerate}
For a set $T$ of frontier-one TGDs, the \emph{$\Delta$-closure
  $\mn{cls}(T,\Delta)$ of $T$} is the smallest set of CQs such that: 
\begin{itemize}

  \item For every TGD $\phi(x,\bar y)\to\exists \bar z\,\psi(x,\bar z)
    \in T$ and every $c\in \Delta$, the CQs $q_\phi(c)$ and
    $q_{\psi}(c)$ are contained in $\mn{cls}(T,\Delta)$; 

  \item if $q(\bar c)\in \mn{cls}(T,\Delta)$ and
    $q_0(\bar c_0),q_1(\bar c_1),\ldots,q_n(\bar c_n)$ is a $\Delta$-splitting of
    $q(\bar c)$, then
    $q_1(\bar c_1),\ldots,q_k(\bar c_k)\in \mn{cls}(T,\Delta)$.

\end{itemize}

It is important to note that:  

\begin{lemma} \label{lem:cls-size}
  The cardinality of $\mn{cls}(T,\Delta)$ is bounded by
  $|T|\cdot 2^m\cdot (|\Delta|+1)^{m}$, with $m$ the maximum of body and head
  width of $T$.
\end{lemma}

\begin{proof}
  Every query in $\mn{cls}(T,\Delta)$ can be obtained by starting with
  a Boolean CQ $q()\leftarrow \exists x\, q_{\psi}(x,\bar z)$ or
  $q()\leftarrow \exists x\, q_{\phi}(x,\bar y)$ 
  for some TGD $\phi(x,\bar y)\to\exists \bar z\,\psi(x,\bar z) \in
  T$, restricting it to some subset of its variables,
  and then possibly replacing any number of variables with constants
  from $\Delta$.
\end{proof}
We can now describe the \ata achieving~$(\dagger)$ more
formally. It uses all members of $\mn{cls}(T_1,\Delta)$ as states. If
it visits $w\in W$ in state $q(\bar c)$, it non-deterministically
chooses a $\adom(L_1(w))$-splitting $q_0(\bar c_0),q_1(\bar c_1),\ldots,q_n(\bar c_n)$
of $q(\bar c)$, verifies that $q_0(\bar c)$ (viewed as a database) is
contained in $L_1(w)$ and, for each $i$ with $1\leq i\leq n$: 
\begin{itemize}

  \item if $q_i(\bar c_i)$ is unary and $\bar c_i=c$, then the \ata sends a copy in
    state $q_i(\bar c_i)$ to some $v\in [w]_c$; 

  \item if $q_i(\bar c_i)$ is Boolean, then the \ata sends a copy in
    state $q_i(\bar c_i)$ to some $v\in W$.

\end{itemize}
Using the priorities we can make sure that the process terminates,
that is, at some point the splitting takes the form of $q_0(\bar
c_0)=q(\bar c)$. Using Lemma~\ref{lem:cls-size} we can verify that
$\Amf_2$ uses exponentially many states.

\subsubsection*{Automaton $\Amf_3$.} Before we can describe the idea,
we need to establish some necessary preliminaries. A TGD $\phi(\bar
x,\bar y)\to\exists \bar z\,\psi(\bar x,\bar z)$ is \emph{full} if
$\bar z$ is empty. A \emph{monadic datalog program} is a set of full
frontier-one TGDs. In such programs, however, we also admit nullary
relation symbols and empty frontiers.

Let $T$ be a set of frontier-one TGDs. We construct from $T$ a monadic
datalog program $T'$ as follows. Recall that all CQs in $q \in
\bodyCQ(T)$ are Boolean or unary.  For every CQ $q(\bar x) \in
\bodyCQ(T)$, introduce a relation symbol $A_{q(\bar x)}$ of arity
$|\bar x|$.  With $\bodyCQ^+(T)$, we denote the set of (Boolean or
unary) CQs that can be obtained from a CQ $q \in \bodyCQ(T)$ by adding
any number of atoms $A_{p(\bar x)}(\bar y)$ with $p(\bar x)\in
\bodyCQ(T)$ and $\bar y$ a tuple of variables from $q$ of length
$|\bar x|$.

Given a CQ $q \in \bodyCQ^+(T)$, we denote with $q^\downarrow$ the
CQ obtained from $q$ by replacing every nullary atom $A_{p(\bar{x})}$
with a copy $p'(\bar x)$ of $p(\bar x)$ that uses only fresh variable
names. In addition, if $\bar x=x$ is non-empty, the copy of $x$ in
$p'(\bar x)$ is identified with $x$.
Now, $T'$ consists of all rules

\begin{enumerate}[label=(\roman*)]

    \item $q(\bar x) \to A_{p(\bar x)}(\bar x)$ such that
      $q(\bar x) \in \bodyCQ^+(T)$ and $p(\bar x) \in \bodyCQ(T)$ have the same
      arity and $D_{q^\downarrow},T\models p(\bar x)$, 

    \item $q(x)\to A_{p(\bar x)}(\bar x)$ such that
      $p(\bar x) \in \bodyCQ(T)$ and $q$ is a conjunction of nullary
      atoms $A_{p'}$ and unary atoms $A_{p'}(x)$, $p' \in \bodyCQ(T)$,
      such that $D_{q^\downarrow},T\models p(\bar x)$.

\end{enumerate}

\begin{lemma} \label{lem:datalog-rewriting} Let $T$ be a set of
  frontier-one TGDs and $T'$ the corresponding monadic datalog
  program. Then, for every database $D$, $q(\bar x)\in
  \bodyCQ(T)$, and every $\bar c\in\adom(D)^{|\bar x|}$,
  \[D,T\models q(\bar c)\quad\text{iff}\quad D,T'\models A_{q(\bar
  x)}(\bar c).\]
\end{lemma}
\begin{proof} For the ``if''-direction, suppose that $D,T\not\models
  q(\bar c)$, for some database $D$, $q(\bar x)\in
  \bodyCQ(T)$, and $\bar c\in\adom(D)^{|\bar x|}$, 
  that is,
  $\bar c\notin q(\chase_T(D))$.  Obtain an instance $I$ from
  $\chase_{T}(D)$ by interpreting the fresh symbols $A_{p(\bar y)}$ in
  the expected way, that is, for all $p(\bar y)\in \bodyCQ(T)$, and 
  $\bar d\in
  \adom(\chase_T(D))^{|\bar y|}$, we have: 
  \[ A_{p(\bar y)}(\bar d)\in I\quad\text{ iff }\quad\bar d\in
  p(\chase_T(D)).\]
  It is readily verified that $I$ is a model of $D$ and $T'$. But
  since $\bar c\notin q(\chase_T(D))$, we have $A_{q(\bar
  x)}(\bar c)\notin I$ and thus $D,T'\not\models A_{q(\bar
  x)}(\bar c)$.
  
  For the ``only if''-direction, let $D,T'\not\models A_{q(\bar
  x)}(\bar c)$, for some database $D$, $q(\bar x)\in \bodyCQ(T)$, and
  $\bar c\in\adom(D)^{|\bar x|}$, that is, $A_{q(\bar x)}(\bar
  c)\notin I$ for some model $I$ of $D$ and $T'$.  We associate with
  every $d\in \adom(I)$ a $T$-type $t_d$ by taking:
  \begin{align*}
    t_d =~{} & \{p(x)\in \bodyCQ(T)\mid A_{p(x)}(d)\in I\}\cup{}\\ &
    \{p()\in\bodyCQ(T) \mid A_{p()}\in I\}.
  \end{align*}
  Now obtain an instance $I'$ from $I$ by adding, for every $d\in
  \adom(I)$ a disjoint copy of $\chase_T(t_d)$ identifying its root
  with $d$. It remains to show that $I'$ is a model of $T$ and that
  $\bar c\notin q(I')$. For both statements, we will need the
  following auxiliary claim.

  \smallskip\noindent\textit{Claim~1.} For all $p(\bar x)\in
  \bodyCQ(T)$ and all $\bar d\in \adom(I)^{|\bar x|}$, we have
  \[ \bar d\in p(I')\quad\text{ iff }\quad A_{p(x)}(\bar d)\in I.\]
      
  \smallskip\noindent\textit{Proof of Claim~1.}  The ``if''-direction
  is immediate from the construction of $I'$, so we concentrate on
  ``only if''. The proof is by induction on the number of variables in
  $p$.

  Let $\bar d\in p(I')$, that is, there is a
  homomorphism $h$ from $p$ to $I'$ with $h(\bar x) = \bar d$.  Let us
  define $h^\downarrow(x)=e$ in case $h(x)$ is in the copy of
  $\chase_T(t_e)$, for all variables $x$ in $p$.  Set
  $H^\downarrow = \{h^\downarrow(x)\mid \text{ $x$ variable in
  $p$}\}$. We decompose $p$ guided by $h$ into queries $p_0$, $p_e$ with
  $e\in H^\downarrow$ as follows: 
  \begin{itemize}

    \item For every $e\in H^\downarrow$, $p_e$ is the restriction of
      $p$ to all variables $y$ in $p$ with $h^\downarrow(y)=e$. 

   \item $p_0$ consists of the remaining atoms and has answer variable
     $x$ in case $p$ has an answer variable $x$.

  \end{itemize}
  Note that all obtained queries are contained in $\bodyCQ(T)$. We
  distinguish three cases. Observe that Case~1 applies if $p$ contains
  at most one variable and thus establishes the
  induction base.

  \smallskip\textit{Case 1: $p_e=p$, for some $e\in H^\downarrow$.} If
  $p$ is Boolean, then there is a homomorphism $g$ from $p$ to
  $\chase_T(t_e)$. If $p$ has answer variable $x$, then $\bar d = e$
  and $g(x)$ is the constant corresponding to the free variable of
  $t_e$.
  Thus, $t_e,T\models p(\bar x)$ and we find in $T'$ the rule $\widehat q(x)\to
  A_{p(\bar x)}(\bar x)$ where $\widehat q$ is the conjunction of all
  $A_{p'(x)}(x)$ with $A_{p'(x)}(e)\in I$ and all $A_{p'(x)}\in I$.
  By definition of $\widehat q$, $e\in \widehat q(I)$ and hence $A_{p(\bar x)}(\bar d)\in I$.

  \smallskip\textit{Case 2: $p_0=p$.} Then $h$ witnesses
  that $\bar d\in p(I)$. Since trivially $D_{p^\downarrow},T\models p(\bar x)$, $T'$ contains the rule
  $p(\bar x)\to A_{p(\bar x)}(\bar x)$. This implies that $A_{p(\bar
  x)}(\bar d)\in I$.

  \smallskip\textit{Case 3: Otherwise.} Then all obtained queries
  $p_0$, $p_e$ have less variables than $p$. We obtain a CQ $\widehat p$
  from $p_0$ by doing the following, for every $e\in H^\downarrow$:
  \begin{itemize}

    \item if $p_0$ and $p_e$ do not share any variable, then add the
      nullary atom $A_{p_e}$, and

    \item if $p_0$ and $p_e$ share a variable, then pick such a
      variable $x_e$, make it an answer variable in $p_e$, and add the
      unary atom $A_{p_e(x_e)}(x_e)$. 

  \end{itemize}
  Observe that $h$ witnesses that $\bar d\in \widehat p(I)$ since
  $\bar d\in p_0(I)$ and, for all $e\in H^\downarrow$, we
  have: 
  \begin{itemize}

    \item If $\widehat p$ contains $A_{p_e}$, then $h$ witnesses that 
      $()\in p_e(I')$, and thus, by induction, $A_{p_e}\in I$. 

    \item If $\widehat p$ contains $A_{p_e(x_e)}(x_e)$, then $h$
      witnesses that $h(x_e)\in p_e(I')$, and thus, by induction,
      $A_{p_e(x_e)}(h(x_e))\in I$. 

  \end{itemize}
  Moreover, it should be clear that $D_{\widehat
  p^\downarrow},T\models p(\bar x)$, and thus we find the rule $\widehat
  p(\bar x)\to A_{p(\bar x)}(\bar x)$ in $T'$, hence \mbox{$A_{p(\bar
  x)}(\bar d)\in I$.}
 
  \medskip This finishes the proof of Claim~1. Claim~1 immediately
  implies that $\vec c\notin q(I')$ since $A_{q(\bar x)}(\bar c)\notin
  I$. It remains to verify the following. 

  \smallskip\noindent\textit{Claim~2.} $I'$ is a model of $T$.
 
  \smallskip\noindent\textit{Proof of Claim~2.} To see that $I'$ is a
  model of $T$ let $\vartheta = \phi(x,\bar y)\to\exists \bar z\,\psi(x,\bar z)$ be
  a TGD in $T$ and suppose that $d\in q_\phi(I')$, that is, there is a
  homomorphism $h$ from $q_{\phi}$ to $I'$ with $h(x) = d$. We
  show that $d\in q_\psi(I')$. We distinguish cases.

  \smallskip\textit{Case~1: $d\in\adom(I)$.} In this case, Claim~1
  implies that $A_{q_\phi}(d)\in I$, and hence $q_{\phi}(x)\in t_d$.
  Since $\chase_T(t_d)$ is a model of both $q_{\phi}(x)$ and
  $\vartheta$, we have that $x\in q_\psi(\chase_T(t_d))$.
  The construction of $I'$ ensures that $x\in q_\psi(I')$.

  \smallskip\textit{Case~2: $d \notin \adom(I)$.} Then $d$ is in
  the copy of $\chase_T(t_e)$, for some $e\in \adom(I)$.  We obtain
  CQs $q_1,q_2$ from 
  $q_\phi$ as follows:
  \begin{itemize}

    \item $q_1$ is the restriction of $q_\phi$ to all variables $y$ such
      that $h(y)$ is in the copy of $\chase_T(t_e)$.

    \item $q_2$ is obtained by starting from the remaining atoms and then
      identifying all variables shared with $q_1$. (Note that every
      such variable $y$ satisfies $h(y)=e$.) If there is none
      such variable, then $q_2$ is Boolean. Otherwise, the variable
      obtained in the identification process is the answer variable. 

  \end{itemize}
  The homomorphism $h$ witnesses that $e\in q_2(I')$ if $q_2$ is unary
  and $()\in q_2(I')$ otherwise.  If $q_2$ is unary, Claim~1 yields
  that $A_{q_2(y)}(e)\in I$ and thus $q_2(x)\in t_e$.  Otherwise,
  Claim~1 yields that $A_{q_2}\in I$ and thus $q_2\in t_e$.  By
  definition of $\chase_T(t_e)$ we find a homomorphism $g$ from $q_2$
  to $\chase_T(t_e)$ that maps the answer variable of $q_2$ (if any)
  to the constant corresponding to the free variable of $t_e$.  Let
  $h'$ be the copy of $h$ that maps $q_1$ to $\chase_T(t_e)$ (instead
  of the copy of $\chase_T(t_e)$ in $I'$). But then $g\cup h'$ is a
  homomorphism from $q_\phi$ to $\chase_T(t_e)$, and hence $d'\in
  q_\phi(\chase_T(t_e))$ where $d'$ is the copy of $d$ in
  $\chase_T(t_e)$. Since $\chase_T(t_e)$ is a model of $\vartheta$, we
  have $d'\in q_\psi(\chase_T(t_e))$ and thus $d\in q_\psi(I')$.  This
  finishes the proof of Claim~2.
\end{proof}

\begin{lemma} \label{lem:datalog-rewriting-cmp}
  Let $T$ be a set of frontier-one TGDs of body width $k$. Then, $T'$ consists of: 
  \begin{itemize}

    \item at most exponentially many rules of type~(i), and 

    \item at most double exponentially many rules of type~(ii).
      
  \end{itemize}
  Moreover, rules of type~(i) have at most $k$ variables and rules of
  type~(ii) have only one variable. $T'$ can be computed in
  time triple exponential in $||T||$. 
\end{lemma}

\begin{proof}
  First note that there are at most exponentially many queries in
  $\bodyCQ^+(T)$. Indeed, by construction, there are only
  exponentially many queries in $\bodyCQ(T)$, and each query has at
  most $k$ variables, $k$ the body width of $T$. These at most $k$
  variables are now labeled with the fresh concept names $A_{\psi(\bar
  x)}$ with $\psi(\bar x)$ in $\bodyCQ(T)$. It follows that there are
  at most exponentially many queries in $\bodyCQ^+(T)$. Overall, there
  are at most exponentially many candidates for rules of type~(i) and
  at most double exponentially many candidates for rules of type~(ii)
  in $T'$.

  Also note that, for each query $q(\bar x)$ that can occur in a rule
  body in~(i) or~(ii), the query $q^\downarrow$ is a $T$-type, and
  thus of size exponential in $||T||$. Moreover, note that the checks
  $D_{q^{\downarrow}},T\models p(\bar x)$ that have to be made in oder
  to decide whether a candidate rule is included in $T'$ are instances
  of query evaluation w.r.t.\ frontier-one TGDs. 
  Since query evaluation w.r.t.\ frontier-one TGDs is
  \TwoExpTime-complete~\cite{BMRT11}, all these checks can be made in
  triple exponential time. 
\end{proof}

We are now in a position to describe the automaton $\Amf_3$. Let
$T_2'$ be the monadic datalog program obtained from $T_2$. The
automaton uses $T_2'$ to verify the correctness of the labeling
$\mu_w$ by visiting every node  $w\in W$ and doing the
following for every $c\in\adom(L_0(w))$ and every $q(\bar x)\in
\bodyCQ(T_2)$: 
\begin{enumerate}

  \item if $q(x)\in \mu_w(c)$ is unary, then verify that
    $I_{W,L_0},T_2\models \Amc_{q(x)}([w]_c)$;

  \item if $q\in \mu_w(c)$ is Boolean, then verify that
    $I_{W,L_0},T_2\models \Amc_{q}$;

  \item if $q(x)\notin \mu_w(c)$ is unary, then verify that
    $I_{W,L_0},T_2\not\models \Amc_{q(x)}([w]_c)$;

  \item if $q\notin \mu_w(c)$ is Boolean, then verify that
    $I_{W,L_0},T_2\not\models \Amc_{q}$.

\end{enumerate}
By Lemma~\ref{lem:datalog-rewriting}, the automaton may use $T'_2$ in
place of $T_2$. For Points~1 and~2, the automaton guesses a
\emph{derivation}, as commonly used to define the semantics of
datalog; for details, we refer to~\cite{AbHV95}. For Points~3 and~4,
it needs to verify that there is no derivation, which is easy by
dualizing the subautomaton for Points~1 and~2. We thus concentrate on
Points~1 and~2.

To verify that $I_{W,L_0},T_2'\models A_{q(x)}([w]_c)$ (resp.,
$I_{W,L_0},T_2'\models A_{q}$), the automaton non-deterministically
chooses a derivation of $A_{q(\bar x)}([w]_c)$ (resp., $A_{q}$) in
$I_{W,L_0}$ under $T_2'$. For doing so, it uses states from
$\mn{cls}(T_2'',\Delta)$ where $T_2''$ is the fragment of $T_2'$
consisting only of the rules of type~(i) and where $\mn{cls}$ defined
as in the description of $\Amf_2$. It starts in state $A_{q(\bar
x)}(c)$ (resp., $A_q$). (Recall that in world $w$, the element
$[w]_c$ of $I_{W,L_0}$ is represented by constant $c$.)

Intuitively, if the automaton visits $w\in W$ in a state
$q(\bar c)$, then this represents the obligation to find a derivation for
$q(\widehat c)$ in $I_{W,L_0}$ under $T_2'$, where $\widehat
c=[w]_{c}$ if $\bar c=c$ consists of a single constant and $\widehat
c$ is empty otherwise. We distinguish cases depending on the shape of
$q(\bar c)$. 

\smallskip\textit{Case~(1)} If $q(\bar c)$ is of shape $A_{p(\bar x)}(\bar
c)$, then the automaton non-deterministically does one of the
following: 
\begin{itemize}

  \item non-deterministically choose a rule $q'(\bar x)\to
    A_{p(\bar x)}(\bar x)$ of type~(i) in $T_2'$ and proceed in state
    $q'(\bar c)$, or

  \item non-deterministically choose a rule $q'(x)\to A_{p(\bar
    x)}(\bar x)$ of type~(ii) in $T_2'$ and:  
    \begin{itemize}

      \item if $\bar c=c$ is a single constant, then (using
	alternation) the automaton proceeds in states $A_{p'}(c)$, for
	all unary atoms $A_{p'}(x)$ that occur in $q'$, and in
	$A_{p'}$, for all nullary atoms $A_{p'}$ that occur in $q'$;

      \item if $\bar c$ is empty, the automaton navigates
	(non-deterministically) to some $w\in W$, picks a constant
	$c\in \adom(L_0(w))$ and proceeds as in the previous item
	(again using alternation).

    \end{itemize}

\end{itemize}

\smallskip\textit{Case~(2)} If $q(\bar c)$ is not of shape $A_{p(\bar
x)}(\bar c)$, then the
automaton non-deterministically chooses an
$\mn{adom}(B(w))$-splitting $q_0(\bar c_0),q_1(\bar c_1),\ldots,q_n(\bar c_n)$
of $q(\widehat c)$. It then obtains $q_0'(\bar c_0)$ from $q_0(\bar c_0)$
by dropping all atoms of the form $A_{p(x)}(x)$ and $A_{p}$ and
proceeds to verify that $q_0'(\bar c)$ (viewed as a database) is contained in
$L_0(w)$. Additionally, for each $i$ with $1\leq i\leq n$: 
\begin{itemize}

  \item if $q_i(\bar c_i)$ is unary with $\bar c_i = c$, then the \ata
    sends a copy in state $q_i(\bar c_i)$ to some $v\in [w]_c$; 

  \item if $q_i(\bar c_i)$ is Boolean, then the \ata sends a copy in
    state $q_i(\bar c_i)$ to some $v\in W$.

\end{itemize}
Finally, the dropped atoms are processed as follows. 
\begin{itemize}

 \item if $A_p(c)$ is a unary atom in $q_0(\bar c_0)$, the automaton sends
   a copy in state $A_p(c)$ to $w$; 

  \item if $A_p$ is a Boolean atom in $q_0(\bar c_0)$, the automaton sends
   a copy in state $A_p$ to $w$.

\end{itemize}
Using the priorities, we can make sure that the process terminates.
Combining Lemma~\ref{lem:cls-size} and
Lemma~\ref{lem:datalog-rewriting-cmp}, one can verify that
$\mn{cls}(T_2'',\Delta)$ (recall that $T_2''$ is the subset of $T_2'$
consisting only of rules of type~(i)) contains exponentially many
queries and thus $\Amf_3$ uses at most exponentially many states. By
Lemma~\ref{lem:datalog-rewriting-cmp}, $\Amf_3$ can be computed in
triple exponential time.

\subsubsection*{Automaton $\Amf_4$.} To construct automaton $\Amf_4$,
first note that
$\mn{chase}_{T_2}(I_{W,L_1}) \not\rightarrow I_{W,L_1}$ if for some
$w\in W$ and some $c\in \adom(L_0(w))$, there is no
$\querysig$-homomorphism $h$ from $\chase_{T_2}(\mu_w(c))$ to
$I_{W,L_1}$ with $h(x) = [w]_c$. It thus suffices to check the latter.

For convenience, we concentrate on the complement and build an
automaton that is capable of verifying that, given $w\in W$ and
$c\in \adom(L_0(w))$,
\begin{itemize}

  \item[$(\dagger)$] there is a $\querysig$-homomorphism $h$ from
    $\chase_{T_2}(\mu_w(c))$ to $I_{W,L_1}$ with $h(x)=[w]_c$.

\end{itemize}
The automaton $\Amf_4$ then non-deterministically guesses a $w\in W$
and a $c\in \adom(L_0(w))$ and uses the complement/dualization of the
automaton that verifies $(\dagger)$.

We rely on the representation of $\chase_{T_2}(\mu_w(c))$ as a
(rooted!)  $\mu_w(c)$-proper $T_2$-labeled instance tree
$\Tmc = (V,E,B,\mu)$, see the discussion that preceeds
Lemma~\ref{lem:limtoinfinite}.  It is important to realize that the
instance $B(v)$ at some node $v\in V$ together with the labeling
$\mu_v$ of $\adom(B(v))$ with $T_2$-types completely determine the
successors $v'$ of $v$ and their labeling $B(v')$ and $\mu_{v'}$. More
precisely, the type $\mu(c)$ of a constant $c\in \adom(B(v))$
determines all successors $v'$ of $v$ that have $c$ in their domain
$\adom(B(v'))$. Moreover, $v$ has a successor with label
$B(v'),\mu_{v'}$ iff $\mu_v(c),T_2\models q_{(B(v'),\mu_{v'})}^c(x)$
 (c.f. Condition~2 of properness).
Since query evaluation w.r.t.\
frontier-one TGDs is \TwoExpTime-complete~\cite{BMRT11} and the size
of the input is exponential in $||T||$, this check is possible in
triple exponential time. Hence, all possible successors can be
computed in triple exponential time.

For achieving~$(\dagger)$, the automaton proceeds as follows. It
memorizes (in its states) the database $B(v)$ at the current node
$v\in V$ of $\Tmc$ and the type labeling $\mu_v$.  It then guesses a
partial $\Sigma_Q$-homomorphism from $B(v)$ to the currently visited node
$w\in W$. Each variable that is mapped to the current state gives rise
to successors $v'$ of $v$ with associated $B(v')$ and $\mu_{v'}$
labelings, and the automaton spawns copies of itself that generate
these successor (as states), moves to neighboring nodes in the input
tree, and proceeds there. As in the encoding of $\Tmc$ as a labeled
tree, it remaps the constants in the instances $B(\cdot)$ to ensure
that only finitely many states are used; this is possible since every
instance $B(v)$ is isomorphic to the head of some TGD in $T_2$.
%
%

More formally, the automaton uses as states pairs $\langle q(\bar
c),\mu\rangle$ where:
\begin{itemize}

  \item $q(\bar c)$ is an element of $\mn{cls}(T_2,\Delta)$, and 

  \item $\mu$ assigns a $T_2$-type to every variable in $q(\bar c)$.

\end{itemize}
When the automaton visits a node $w\in W$ in state $\langle q(\bar
c),\mu\rangle$, this represents the obligation to verify that there is
a $\querysig$-homomorphism $h$ from $q$ to $I_{W,L_1}$ such that:
\begin{itemize}

  \item for every constant $c\in \bar c$, $h(c)=[w]_c$, 
    and 

  \item for every variable $x$ in $q$, there is a
    $\querysig$-homomorphism $g$ from $\chase_{T_2}(\mu(x))$ to
    $I_{W,L_1}$ with $g(x)=h(x)$.

\end{itemize}
For doing so, the automaton non-deterministically chooses an
$\mn{adom}(B(w))$-splitting $q_0(\bar c_0), \ldots, q_n(\bar c_n)$ of
$q(\bar c)$ and proceeds as follows: 
\begin{itemize}

  \item it verifies that the $\querysig$-restriction of $q_0(\bar c_0)$
    is a subset of $L_1(w)$;


      \item for every $i$ with $1\leq i\leq n$, we let $\mu_i$ be the restriction of
	$\mu$ to the variables in $q_i$, then 

    \begin{itemize}

      \item if $q_i(\bar c_i)$ is unary with $\bar c_i = c$, then the \ata
	sends a copy in state $\langle q_i(\bar c_i),\mu_i\rangle$ to
	some $v\in [w]_c$.

      \item if $q_i(\bar c_i)$ is Boolean, then the \ata sends a copy in
	state $\langle q_i(\bar c_i),\mu_i\rangle$ to some $v\in W$.

    \end{itemize}

  \item for every variable $x$ in $q$ that was replaced by a constant
    $d$ in the splitting, consider any node $v$ in \Tmc and any
    $e \in \mn{adom}(B(v))$ with $\mu_v(e)=\mu(x)$,\footnote{Choosing
      different $v$ and $e$ leads to exactly the same result provided
      that $\mu_v(e)=\mu(x)$.} and all successors $v'$ of $v$ with
    $\mn{adom}(B(v)) \cap \mn{adom}(B(v')) \subseteq \{ e \}$. Let $q'$
    be $B(v')$ viewed as a CQ which is Boolean with $e$ viewed as
    the answer variable if $e \in \mn{adom}(B(v'))$ and Boolean
    otherwise. Further let $\mu'=\mu_{v'}$. The automaton
    does the following:
    \begin{itemize}

      \item if $q'$ is unary, then it sends a copy in state
	$\langle q'(d),\mu'\rangle$ to some $w'\in [w]_d$;

      \item if $q'$ is Boolean, then it sends a copy in state
	$\langle q',\mu'\rangle$ to some $w'\in W$.

    \end{itemize}





\end{itemize}
Overall, one can verify that the number of states is at most double
exponential in $||T_2||$. There are doubly exponentially many types
and, by Lemma~\ref{lem:cls-size}, the size of $\mn{cls}(T_2,\Delta)$ is
bounded exponentially in the size of $||T_2||$. Since all queries in
$\mn{cls}(T_2,\Delta)$ have at most $||T_2||$ variables, the triple
exponential bound follows. As argued, the automaton can 
be computed in time triple exponential in $||T_2||$. 

\subsection{Upper Bounds for CQ-Conservativity}

We actually work with a refinement of the characterization given in
Theorem~\ref{thm:charsecond}; its proof is based on
Lemma~\ref{lem:finiteI}. The formulation of this refinement is
somewhat more technical than the formulation of
Theorem~\ref{thm:charsecond}, and in fact we decided to go in these
two steps for didactic reasons.
%
%
%
%
%
\begin{restatable}{theorem}{thmcharsecondrefined}
  \label{thm:charsecondrefined}
  Let $T_1$ and $T_2$ be sets of frontier-one TGDs, $\dbsig$ and
  $\querysig$ schemas, and $k$ the body width of $T_1$.  Then
  $T_1 \models^{\text{CQ}}_{\dbsig,\querysig} T_2$ iff for all
  tree-like $\dbsig$-databases $D$ of width at most $k$ and all
  tree-like models $I$ of $T_1$ and $D$ of width $\max(k,\ell)$, the
  following holds:
  \begin{enumerate}[label=\arabic*.]

    \item $\chase_{T_2}(D)|^\con_{\querysig} \rightarrow_{\querysig}
      I$;

    \item for every labeled $\querysig$-head fragment $A=(F,\mu)$ of
      $T_2$ with $\chase_{T_2}(D) \models q_A$, one
      of the following holds:
      \begin{enumerate}

	\item $\chase_{T_2}(D_A)|^\con_{\querysig} \rightarrow_{\querysig} I$;

	\item $\chase_{T_2}(D_A)|^\con_{\querysig} \rightarrow^{\lim}_{\querysig}
          \chase_{T_1}(\tp_{T_1}(I,c))$
	  for some $c \in \adom(D)$. 

      \end{enumerate}

  \end{enumerate} 
\end{restatable}
\noindent
\emph{Proof.}
It suffices to show that for every tree-like $\dbsig$-database~$D$ of
width $k$, Conditions~1 and~2 of Theorem~\ref{thm:charsecond} are
satisfied if and only if for all tree-like models $I$ of $T_1$ and $D$ of width
$\max(k,\ell)$, Conditions~1 and 2 above are satisfied.

First assume that for all tree-like models $I$ of $T_1$ and $D$ of
width $\max(k,\ell)$, Conditions~1 and~2 above are satisfied. Since
$\mn{chase}_{T_1}(D)$ is such a model, Condition~1 of
Theorem~\ref{thm:charsecond} is also satisfied. Now for
Condition~2. By Lemma~\ref{lem:finiteI}, for all maximally
\querysig-connected components $J$ of
$\chase_{T_2}(D)\setminus \chase_{T_2}(D)|^\con_\querysig$,
Condition~2(a) or~2(b) above is satisfied when
$\chase_{T_2}(D_A)|^\con_{\querysig}$ is replaced by $J$. In the
former case, also Condition~2(a) of Theorem~\ref{thm:charsecond} is
satisfied. In the latter case, it follows that
$J \rightarrow^{\lim}_{\querysig} \chase_{T_1}(D)$.  One can then show
exactly as in the proof of Theorem~\ref{thm:charsecond} that either
Condition~2(a) or~2(b) of that theorem is satisfied.

Conversely, suppose that Conditions~1 and~2 of
Theorem~\ref{thm:charsecond} are satisfied for $D$. Since
$\mn{chase}_{T_1}(D) \rightarrow I$ for every model $I$ of $T_1$ and
$D$, Condition~1 of Theorem~\ref{thm:charsecond} implies that for all
tree-like models $I$ of $T_1$ and $D$ of width $\max(k,\ell)$,
Condition~1 above is satisfied. It remains to argue that Condition~2
above is satisfied.  Assume to the contrary that is is not. Then there
is some tree-like model $I$ of $T_1$ and $D$ of width $\max(k,\ell)$
and some labeled $\querysig$-head fragment $A=(F,\mu)$ of $T_2$ such
that both 2(a) and 2(b) above are violated. Since
$\mn{chase}_{T_1}(D) \rightarrow I$, these conditions are still
violated when $I$ is replaced by $\mn{chase}_{T_1}(D)$.

We distinguish the following cases:
  \begin{itemize}
  \item 
%
$\chase_{T_2}(D_A)|^\con_{\querysig}\not\to
  \chase_{T_2}(D)|^\con_{\querysig}$.

  Then $\chase_{T_2}(D_A)|^\con_{\querysig}$ is an induced subinstance
  of a maximally $\querysig$-connected component $I$ of
  $\chase_{T_2}(D) \setminus \chase_{T_2}(D)|^\con_{\querysig}$. Thus,
  $\chase_{T_2}(D_A)|^\con_{\querysig}\not\to_{\querysig}
  \chase_{T_1}(D)$ implies $I\not\to_{\querysig} \chase_{T_1}(D)$ and
  Condition~2(a) of Theorem~\ref{thm:charsecond} is not satisfied.
  Moreover,
  $\chase_{T_2}(D_A)|^\con_{\querysig}\not\to^\lim_{\querysig}
  \chase_{T_1}(\tp_{T_1}(\chase_{T_1}(D),c))$ for all $c\in \adom(D)$
  implies $I\not\to^\lim_{\querysig} \chase_{T_1}(D)|^\downarrow_c$,
  for all $c\in \adom(D)$.  This is because
  $\chase_{T_2}(D_A)|^\con_{\querysig}$ is a subinstance of $I$ and,
  due to Lemma~\ref{lem:typedet}, $\chase_{T_1}(D)|^\downarrow_c$ is a
  subinstance of $\chase_{T_1}(\tp_{T_1}(\chase_{T_1}(D),c))$.  Thus,
  both Condition~2(a) and~2(b) of Theorem~\ref{thm:charsecond} are
  violated, a contradiction.


\item
  $\chase_{T_2}(D_A)|^\con_{\querysig}\to
  \chase_{T_2}(D)|^\con_{\querysig}$.

  Then $\chase_{T_2}(D_A)|^\con_{\querysig}
  \not\rightarrow_{\querysig} \chase_{T_1}(D)$ implies 
  $\chase_{T_2}(D)|^\con_{\querysig} \not\rightarrow_{\querysig}
  \chase_{T_1}(D)$. Hence, Condition~1 of
  Theorem~\ref{thm:charsecond} is
  not satisfied, which is again a contradiction. \hfill $\Box$
  \end{itemize}

\medskip

Let $T_1,T_2,\dbsig,\querysig$ be given. We may again assume without
loss of generality that all symbols from \dbsig and \querysig occur in
$T_1$.  Let $k$ and $\ell$ be the body and head width of $T_1$.  
It suffices to devise a \ata \Bmf such
that
\begin{itemize}

  \item[$(\ast_\Bmf)$] $\Bmf$ accepts all tree-like instances $I$ of width
    $\max(k,\ell)$ that are a model of $T_1$ and of some tree-like
    $\querysig$-database $D$ of width $k$ such that Condition~1
    and Condition~2 of Theorem~\ref{thm:charsecondrefined} are
    violated.

\end{itemize}
In order to represent tree-like instances of bounded width as the
input to \ata{}s, we use exactly the same encoding of infinite
instances as for hom-conservativity, and in fact, the constructed
automata run over the same alphabet
$\Theta = \Theta_0\times\Theta_0\times \Theta_1$.  Recall that an
input tree over this alphabet represents a $\dbsig$-database $D$ in
the first component, a model $I$ of $T_1$ and $D$ in the second
component, and the chase of $D$ with $T_2$, restricted to $\adom(D)$,
in the last component.

The desired \ata \Bmf is constructed as the intersection of {\ata}s
$\Bmf_0,\Bmf_1,\Bmf_2,\Bmf_3$, and \Bmf' where \Bmf' in turn is the union of
{\ata}s $\Bmf_4$ and $\Bmf_5$, all of them provided by the following
lemma.
\begin{lemma} \label{lem:cq-automata}
  There are {\ata}s $\Bmf_0,\Bmf_1,\Bmf_2,\Bmf_3,\Bmf_4$ such that:
  \begin{itemize}

    \item[--] $\Bmf_0$ accepts $(W,L)$ iff it is well-typed and
      $(W,L_0)$ and $(W,L_1)$ are
      well-formed;

    \item[--] $\Bmf_1$ accepts $(W,L)$ iff $I_{W,L_0}$ is a
      $\dbsig$-database of width~$k$; 

    \item[--] $\Bmf_2$ accepts $(W,L)$ iff $I_{W,L_1}$ is a model of
      $I_{W,L_0}$ and $T_1$;

    \item[--] $\Bmf_3$ accepts $(W,L)$ iff 
      for every $w\in W$ and every $c\in \adom(L_0(w))$, 
      \[\mu_w(c) = \tp_{T_2}(\chase_{T_2}(I_{W,L_0}),[w]_c).\]

    \item[--] $\Bmf_4$ accepts $(W,L)$ iff Condition~1 of
      Theorem~\ref{thm:charsecondrefined} is violated;
       with `$I$' replaced
       with `$I_{W,L_1}$' is violated;
      
    \item[--] $\Bmf_5$ accepts $(W,L)$ iff Condition~2 of
      Theorem~\ref{thm:charsecondrefined} with `$I$' and `$D$'
      replaced with `$I_{W,L_1}$' and `$I_{W,L_0}$', respectively, is
      violated;

  \end{itemize}
  The number of states
  \begin{itemize}

    \item[--] of $\Bmf_0$ is exponential in $||T_1||$ (and independent
      of $T_2$);

    \item[--] of $\Bmf_1$ does not depend on the input;

    \item[--] of $\Bmf_2$ is exponential in $||T_1||$ (and independent
      of $T_2$);

    \item[--] of $\Bmf_3$ is exponential in $||T_2||$ (and independent
      of $T_1$); 

    \item[--] of $\Bmf_4$ is exponential in $||T_2||$ (and independent
      of $T_1$);

    \item[--] of $\Bmf_5$ is double
      exponential in both $||T_1||$ and $||T_2||$. 

  \end{itemize}
  All automata can be constructed in time triple exponential in
    $||T_1||+||T_2||$
%
  and have maximum priority
  one.  
  %
\end{lemma}
It can be verified that \Bmf satisfies ($\ast_\Bmf$) and thus $L(\Bmf)\neq\emptyset$ iff $T_1
\not\models_{\dbsig,\querysig}^{\textup{CQ}} T_2$.  The rest of this
section is devoted to proving Lemma~\ref{lem:cq-automata}. Automata
$\Bmf_0,\Bmf_1,\Bmf_2,\Bmf_3$ are exactly as
$\Amf_0,\Amf_1,\Amf_2,\Amf_3$ in
Lemma~\ref{lem:hom-automata}, so we concentrate on $\Bmf_4$ and
$\Bmf_5$.

\subsubsection{Automaton $\Bmf_4$.} The task of $\Bmf_4$ is to verify
$\chase_{T_2}(I_{W,L_0})|^\con_{\querysig} \not\rightarrow_{\querysig}
I_{W,L_1}$. Note that this is very similar to what is achieved by
automaton $\Amf_4$ from Lemma~\ref{lem:hom-automata}, which verifies
that $\chase_{T_2}(I_{W,L_0}) \not\rightarrow_{\querysig} I_{W,L_1}$.
In fact, it can be solved using essentially the same construction and
thus has the same size and can be computed in the same time as
$\Amf_4$. More precisely, the gist of the construction of $\Amf_4$ is
to find an automaton 
that verifies, given
$w\in W$ and $c\in \adom(L_0(w))$, that there is a
$\querysig$-homomorphism $h$ from $\chase_{T_2}(\mu_w(c))$ to
$I_{W,L_1}$ with $h(x) = [w]_c$. This is done by constructing
$\chase_{T_2}(\mu_w(c))$ `in the states'. $\Bmf_4$ does exactly the
same, but disregards $\Sigma_Q$-disconnected parts of $\chase_{T_2}(\mu_w(c))$.

\subsubsection{Automaton $\Bmf_5$.} The task of $\Bmf_5$ is to verify
that for all
$\querysig$-labeled head fragments $A=(F,\mu)$ of $T_2$ such that
$\chase_{T_2}(I_{W,L_0}) \models q_A$, the
following hold:
\begin{enumerate}

  \item $\chase_{T_2}(D_A)|^\con_{\querysig} \not\rightarrow_{\querysig} I_{W,L_1}$;

  \item $\chase_{T_2}(D_A)|^\con_{\querysig} \not\rightarrow^{\lim}_{\querysig}
    \chase_{T_1}(\mn{tp}_{T_1}(I_{W,L_0},[w]_c))$ for all $[w]_c \in \adom(I_{W,L_0})$.

\end{enumerate}
%
%
Note that the condition
$\chase_{T_2}(I_{W,L_0}) \models q_A$ is satisfied iff for some
$[w]_c\in\adom(I_{W,L_0})$, the type
$t=\tp_{T_2}(\chase_{T_2}(I_{W,L_0},[w]_c))$
satisfies $t,T_2\models q_A$. Since we are considering the intersection
with $\Bmf_3$, we can assume that
$\mu_w(c) = \tp_{T_2}(\chase_{T_2}(I_{W,L_0},[w]_c))$ and thus, the
latter condition is satisfied 
iff $\mu_w(c),T_2\models q_A$ for some
$w\in W$ and $c\in \adom(L_0(w))$. 

Thus, the automaton can identify all relevant labeled
$\querysig$-head fragments $A=(F,\mu)$ of $T_2$
by visiting all
$w\in W$, all $c\in\adom(L_0(w))$, and testing for each
whether $\mu_w(c),T_2
\models q_A$ is satisfied.  The result of all possible such tests
can be computed in time triple exponential in $||T_2||$
already during the construction of $\Bmf_5$, since query evaluation
w.r.t.\ frontier-one TGDs is \TwoExpTime-complete~\cite{BMRT11}.

If the test $\mu_w(c),T_2
\models q_A$ is succesfull, the automaton has
to verify Points~1 and~2 above for
$\chase_{T_2}(D_A)|^\con_{\querysig}$. There is
once more a lot of similarity between Point~1 and what is achieved by
automaton $\Amf_4$ from Lemma~\ref{lem:hom-automata}. Constructing an
automaton that verifies Point~1 is thus another variation of the
construction of $\Amf_4$, the main difference being that instead of
chasing a single type we chase $D_A$ with $T_2$ in the states of the
automaton. In particular, the automaton starts in state $\langle
q_F,\mu\rangle$ (note that $q_F\in\mn{cls}(T_2,\Delta)$).

For Point~2, we invoke Theorem~\ref{thm:decidelimit} for every
labeled $\querysig$-head fragment $A=(F,\mu)$ of $T_2$ identified above. The automaton
memorizes $A$ in its states and visits (again) all $w\in W$, and all
$c\in \adom(I_{W,L_0})$ in order to verify that 
\[\chase_{T_2}(D_A)|^\con_{\querysig} \not\rightarrow^{\lim}_{\querysig}
  \chase_{T_1}(\tp_{T_1}(I_{W,L_1},[w]_c)).\]
Recall that these tests have been precomputed via
Theorem~\ref{thm:decidelimit}. Hence, all the automaton has to do at
this point is to guess\footnote{Recall that the $T_1$-type is not
represented in the input.} the $T_1$-type $t$ of $[w]_c$ in
$I_{W,L_1}$, verify that it is the correct type using the monadic
datalog rewriting $T_1'$ of $T_1$ as in automaton~$\Amf_3$ of
Lemma~\ref{lem:hom-automata}, and lookup the result of 
$\chase_{T_2}(D_A)|^\con_\querysig\rightarrow^{\lim}_{\querysig} \chase_{T_1}(t)$ in
the precomputated table.

Overall, the resulting automaton is of size double exponential in both $||T_2||$
(for Point~1) and $||T_1||$ (for guessing a $T_1$-type and verifying it 
in Point~2). It can be computed in triple exponential time.
In particular, the computation of the lookup table for Point~2 is
possible in triple exponential time, by Theorem~\ref{thm:decidelimit}.

\subsection{Proof of Theorem~\ref{thm:decidelimit}}

We prove Theorem~\ref{thm:decidelimit} via the characterization of
bounded homomorphisms in terms of standard (unbounded) homomorphisms
given by Lemma~\ref{lem:limtoinfinite}. Let two sets of frontier-one
TGDs $T_1,T_2$, a schema $\Sigma$, a labeled $\Sigma$-head fragment $A =
(D,\mu)$ of $T_2$, and a $T_1$-type $\widehat t$ be given. It suffices to devise a
\ata $\Cmf$ such that:
\begin{itemize}

\item[$(\ast_\Cmf)$] \Cmf accepts all encodings of $\widehat t$-proper
  $T_1$-labeled instance trees $\Tmc$ of width~$m$ such
  that $\chase_{T_2}(D_A)|_{\Sigma}^\con\to I_\Tmc$.

\end{itemize}
Here we need to work with possibly non-rooted instance trees, and thus we
slightly modify our encoding of instance trees as input to the \ata.  The
input alphabet is $\Theta' = \Theta_0 \times \{0,1\} \times \Theta_1'$
where $\Theta_0$ is defined as above and $\Theta'_1$ is the set of all
mappings $\mu:\Delta'\to \TP(T_1)$ for some $\Delta'\subseteq \Delta$
with $|\Delta'|\leq m$. Note that, in contrast to the alphabet
$\Theta_1$ employed before, here we use $T_1$-types in place of
$T_2$-types. For a $\Theta'$-labeled tree $(W,L)$ and $w \in W$ with
$L(w)=(B,i,\mu)$, we use $L_0(w)$ to denote $B$, $i_w$ to denote $i$,
and $\mu_w$ to denote $\mu$.  Our aim is that every $\Theta'$-labeled
tree $(W,L)$ represents a $T_1$-labeled instance tree
$\Tmc=(V,E,B,\mu)$ where the $(V,E,B)$-part is represented by
$(W,L_0)$ as before and the $\mu$-part is represented by $(W,\mu_w)$.

The additional labeling with the 0/1-marker $i_w$ is necessary because
$T_1$-labeled instance trees need not have a root and thus may contain
an infinite predecessor path. This path will be represented as a
\emph{downward} path in the (rooted!) $\Theta'$-labeled trees, but
marked with a 1-marker for identification purposes. 

A $\Theta'$-labeled tree $(W,L)$ is \emph{well-typed} if, for all
$w\in W$, the domain of $\mu_w$ is $\adom(L(w))$, and for all
successors $v$ of $w$, and all $d\in
\adom(L_0(w))\cap \adom(L_0(v))$, we have $\mu_w(d)=\mu_v(d)$. It is
\emph{well-formed} if $(W,L_0)$ satisfies the two conditions of
well-formedness for $\Theta_0$-labeled trees from the automata
constructions above plus the following additional condition:
\begin{itemize}

  \item there is a finite or infinite (and non-empty) path
    $\Pi=w_0,w_1,w_2,\ldots$ in $W$ that starts at the root such that
    all nodes $w\in W$ with $i_w=1$ lie on this path.

\end{itemize}
Every well-typed and well-formed $\Theta'$-labeled tree $(W,L)$ gives
rise to a $\fullsig$-instance tree $(V,E,B)$ and an
associated instance $I_{W,L}$ as follows. 
\begin{itemize}

  \item the set of nodes $V$ is $W$;

  \item the set of edges $E$ is defined as follows: 
    \begin{itemize}

      \item if $w'$ is a successor of $w$ and $i_{w'}=0$, then $(w,w')\in E$;

      \item if $w'$ is a successor of $w$ and $i_{w'}=1$, then (both $w,w'$ lie on the
	path $\Pi$) and $(w',w)\in E$.

    \end{itemize}

\end{itemize}
That is, the successor relation on the path $\Pi$ becomes the
predecessor relation; the remaining successor relations stay the same.
The definition of the labeling $B$ and consequently also of $I_{W,L}$
is exactly as in the preceeding encoding.  Setting $\mu=\bigcup_{w\in
W}\mu_w$, this extends to a $T_1$-labeled instance tree
$\Tmc_{W,L}=(V,E,B,\mu)$.

Conversely, for every $T_1$-labeled instance tree $\Tmc=(V,E,B,\mu)$
of width at most $m$, we can find a $\Theta'$-labeled tree $(W,L)$
that represents $\Tmc$ in the sense that $\Tmc_{W,L}$ is isomorphic to
$\Tmc$.  Since $\Delta$ is of size $2m$, it is possible to select a
mapping $\pi:\adom(I_\Tmc)\to \Delta$ such that for each edge $(v,w)\in E$
and all constants $c,c'\in \adom(B(w))\cup\adom(B(v))$, we have
$\pi(c) = \pi(c')$ iff $c=c'$.  Define the $\Theta'$-labeled tree
$(W,L)$ as follows:
\begin{itemize}

  \item If $(V,E)$ has a root, then $W=(V,E)$. Otherwise, there is an
    infinite path $v_0,v_1,\ldots$ in $(V,E)$ such that $v_{i+1}$ is a
    predecessor of $v_i$, for all $i\geq 0$. We make this path the
    infinite \emph{successor} path $\Pi$ starting from $v_0$ (and leave all
    other successor relations untouched).

  \item For all $w\in W$, $L(w)=(B(w),0,\mu_w)$ if $w\notin \Pi$ and
    $B_w=(B(w),1,\mu_w)$, if $w\in \Pi$.

\end{itemize}
Clearly, $(W,L)$ satisfies the desired properties.

\smallskip

The automaton $\Cmf$ is the intersection of the three \ata{}s
$\Cmf_0,\Cmf_1,\Cmf_2$ provided by the following lemma.
\begin{lemma} \label{lem:proper-automata}
  There are {\ata}s $\Cmf_0,\Cmf_1,\Cmf_2$ such that:
  \begin{itemize}

  \item[--] $\Cmf_0$ accepts $(W,L)$ iff $(W,L)$ is well-typed and
    well-formed;

  \item[--] $\Cmf_1$ accepts $(W,L)$ iff the $T_1$-labeled
    instance tree $\Tmc_{W,L}$ is $\widehat t$-proper;

    \item[--] $\Cmf_2$ accepts $(W,L)$ iff
      $\chase_{T_2}(D_2)|_{\Sigma}^\con\to I_{W,L}$. 

  \end{itemize}
  The number of states of 
  \begin{itemize}

    \item[--] $\Cmf_0$ is exponential in $||T_1||$ (and independent of
      $T_2$); 

    \item[--] $\Cmf_1$ is linear in $||T_1||$ (and independent of $T_2$); 

    \item[--] $\Cmf_2$ is double exponential in $||T_2||$ (and independent of
      $T_1$).
 
  \end{itemize}
  All automata can be constructed in time triple exponential in
    $||T_1||+||T_2||$
%
  and have maximum priority 
  one. 
  %
\end{lemma}
It can be verified that \Cmf satisfies $(\ast_\Cmf)$, and thus
$L(\Cmf)\neq \emptyset$ iff there is some $\widehat t$-proper $T_1$-labeled
instance tree $\Tmc$ with
$\chase_{T_2}(D_A)|_{\Sigma}^\con\to I_\Tmc$. The automaton $\Cmf_0$
is straightforward.

\subsubsection{Automaton $\Cmf_1$.} The automaton simply visits every
node $w \in W$ in the input tree $(W,L)$ and verifies locally at each
node that Conditions~1 and~2 of properness are satisfied. For
Condition~1, we have to check whether the labeling $L_1(w)$ of the
current node $w$ satisfies Condition~1 of Properness.  Condition~1(a)
is a simple lookup and for Condition~1(b), one has to decide (during
the construction of the automaton) whether $\widehat t,T_1\models
q_{(B(v),\mu_v)}$ for all possible labelings $(B(v),\mu(v))$. This is
possible in time triple exponential in $||T_1||$, since query
evaluation w.r.t.\ frontier-one TGDs is
\TwoExpTime-complete~\cite{BMRT11}.  For Condition~2 of properness,
the automaton needs to memorize the constant $c$ (if any) that is
shared between neighboring nodes in $W$. The
condition~$\mu_u(c),T_1\models q^c_{(B(w),\mu_v)}(x)$ that is part of
Condition~2 of properness can then be checked, again in triple
exponential time in $||T_1||$. Thus, $\Cmf_1$ can be computed in
triple exponential time. 

\subsubsection{Automaton $\Cmf_2$.} The check
$\chase_{T_2}(D_A)|_{\Sigma}^\con\to I_{W,L}$ is similar to what
$\Amf_4$ in Lemma~\ref{lem:hom-automata} achieves and, in fact,
exactly what the sub-automaton for Point~1 of $\Bmf_5$ in
Lemma~\ref{lem:cq-automata} achieves. We repeat it here for the sake
of convenience. The \ata $\Cmf_2$ behaves exactly as $\Amf_4$, but
starts in state $\langle q_F,\mu\rangle$. Recall that $A=(F,\mu)$,
that $q_F$ is $F$ viewed as Boolean CQ, and that
$q_F\in\mn{cls}(T_2,\Delta)$ and so $\langle q_F,\mu\rangle$ is a
state in $\Amf_4$.

\end{document}

%% file: fig1.tex
\begin{tikzpicture}[scale=0.08]
\usetikzlibrary{arrows.meta}
  
\newcommand{\nast}{\coordinate (d) at ($(g)+(d)$);}

\coordinate (d) at (0,0);
\coordinate (g) at (20,0);
\coordinate (m) at (-10,0);


\draw[color=red, very thick, densely dashed, -stealth] ($(d)+(0.5,25)$)--($(g)+(-1.2,30)$) ;
\draw[color=green,very thick, densely dotted, -stealth] ($(d)+(0.5,5)$)--($(g)+(-1.2,0)$);

\draw[color=zgnily, very thick, -stealth] ($(d)+(0,24.5)$)--($(d)+(0,6.0)$) ;
\draw [black, thick, fill=black ] ($(d)+(0,5)$) circle [radius=1];
\draw [black, thick, fill=black ] ($(d)+(0,25)$) circle [radius=1];

\nast
\draw[color=gray, thick, -stealth] ($(d)+(0.8,16)+(m)$)--($(g)+(-.8,29)$);
\draw[color=gray, thick, , -stealth] ($(d)+(0.8,14)+(m)$)--($(g)+(-.8,1)$);

\draw [black, very thick, fill=white ] ($(d)$) circle [radius=1];
\draw [black, very thick, fill=white ] ($(d)+(0,30)$) circle [radius=1];
\draw [black, very thick, fill=white ] ($(d)+(0,15)+(m)$) circle [radius=1];

\nast
\draw [black, very thick, fill=white ] ($(d)$) circle [radius=1];
\draw [black, very thick, fill=white ] ($(d)+(0,30)$) circle [radius=1];
\draw [black, very thick, fill=white ] ($(d)+(0,15)+(m)$) circle [radius=1];

\draw[color=gray, thick, , -stealth] ($(d)+(.8,16)+(m)$)--($(d)+(-.8,29)$);
\draw[color=gray, thick, , -stealth] ($(d)+(.8,14)+(m)$)--($(d)+(-.8,1)$);

\draw[color=red, very thick, densely dashed, -stealth] ($(d)+(1.3,30)-(g)$)--($(d)+(-1.2,30)$) ;
\draw[color=green,very thick, densely dotted, -stealth] ($(d)+(1.3,0)-(g)$)--($(d)+(-1.2,0)$);

\nast
\draw [black, very thick, fill=white ] ($(d)$) circle [radius=1];
\draw [black, very thick, fill=white ] ($(d)+(0,30)$) circle [radius=1];
\draw [black, very thick, fill=white ] ($(d)+(0,15)+(m)$) circle [radius=1];

\draw[color=gray, thick, , -stealth] ($(d)+(.8,16)+(m)$)--($(d)+(-.8,29)$);
\draw[color=gray, thick, , -stealth] ($(d)+(.8,14)+(m)$)--($(d)+(-.8,1)$);

\draw[color=red, very thick, densely dashed, -stealth] ($(d)+(1.3,30)-(g)$)--($(d)+(-1.2,30)$) ;
\draw[color=green,very thick, densely dotted, -stealth] ($(d)+(1.3,0)-(g)$)--($(d)+(-1.2,0)$);

\nast
\draw [black, very thick, fill=white ] ($(d)$) circle [radius=1];
\draw [black, very thick, fill=white ] ($(d)+(0,30)$) circle [radius=1];
\draw [black, very thick, fill=white ] ($(d)+(0,15)+(m)$) circle [radius=1];

\draw[color=gray, thick, , -stealth] ($(d)+(.8,16)+(m)$)--($(d)+(-.8,29)$);
\draw[color=gray, thick, , -stealth] ($(d)+(.8,14)+(m)$)--($(d)+(-.8,1)$);

\draw[color=red, very thick, densely dashed, -stealth] ($(d)+(1.3,30)-(g)$)--($(d)+(-1.2,30)$) ;
\draw[color=green,very thick, densely dotted, -stealth] ($(d)+(1.3,0)-(g)$)--($(d)+(-1.2,0)$);

\draw [black, thick, fill=black ] ($(d)+(10,0)$) circle [radius=0.2];
\draw [black, thick, fill=black ] ($(d)+(10,30)$) circle [radius=0.2];

\draw [black, thick, fill=black ] ($(d)+(15,0)$) circle [radius=0.2];
\draw [black, thick, fill=black ] ($(d)+(15,30)$) circle [radius=0.2];

\draw [black, thick, fill=black ] ($(d)+(20,0)$) circle [radius=0.2];
\draw [black, thick, fill=black ] ($(d)+(20,30)$) circle [radius=0.2];

 \draw[color=red, very thick, densely dashed, -stealth] ($(g)+(-10,-10)$)--($(g)+(5,-10) $) ;
 \draw[color=green,very thick, densely dotted, -stealth] ($(g)+(-10,-15)$)--($(g)+(5,-15) $) ;

 \draw[color=gray, thick, , -stealth] ($(g)+(30,-10)$)--($(g)+(45,-10) $) ;
 \draw[color=zgnily, very thick, -stealth] ($(g)+(30,-15)$)--($(g)+(45,-15) $) ;
 
  \node at ($(g)+(12,-15) $) {{\color{black} \small Thisbe}};
  \node at ($(g)+(13.5,-10) $) {{\color{black} \small Pyramus}};
  
  \node at ($(g)+(55.8,-15) $) {{\color{black} \small Encounter}};
  \node at ($(g)+(54,-10) $) {{\color{black} \small Channel}};
  
    \node at ($(0,28) $) {{\color{black} \small $c_0$}};
     \node at ($ (0,1)$) {{\color{black} \small $c_0'$}};

\end{tikzpicture}

%% file: fig2.tex
\begin{tikzpicture}[scale=0.27]
\usetikzlibrary{arrows.meta}  

\newcommand{\tpunkt}[1]{
\draw[line width=0.5, densely dotted, color=green,-stealth] ($#1$) to [out=270,in=90] (el);
\draw [black, thick, fill=black ] ($#1$) circle [radius=0.2]; }

\newcommand{\ppunkt}[1]{
\draw[densely dashed, color=red,-stealth] ($#1$) to [out=90,in=270] (er); 
\draw [black, thick, fill=black ] ($#1$) circle [radius=0.2]; }

\newcommand{\bpunkt}[1]{
\draw[line width=0.6, densely dotted, color=green,-stealth] ($#1$) to [out=90,in=270] (er); 
\draw[ densely dashed, color=red,-stealth] ($#1$) to [out=270,in=90] (el);
\draw[black, thick, fill=black ] ($#1$) circle [radius=0.2]; }

 \newcommand{\redlajn}[2]{
 \draw[line width=1.2, densely dashed, color=red,-stealth] ($#1$) to ($#2+(0.2,0)$) ; }

 \newcommand{\grinlajn}[2]{
 \draw[line width=1.2, densely dotted, color=green,-stealth] ($#1$) to ($#2+(0.2,0)$) ; }

\coordinate (el) at (7,-12);
\coordinate (er) at (13,12);

\coordinate (b0) at (0,0);
\coordinate (b1) at (8,0);
\coordinate (b2) at (16,0);
\coordinate (b3) at (28,0);
\coordinate (b4) at (31,0);

\draw[-latex, ultra thick, color=zgnily,-stealth] (b4) to (b3); 
\draw[line width=1.2, color=red,densely dashed, -stealth] (b4) to [out=90,in=60] ($(b3)+(0.2,0.2)$);


\coordinate (p1) at (2,2);
\coordinate (p2) at (4,3);
\coordinate (p3) at (6,2);
\coordinate (p4) at (9,3);

\coordinate (p5) at (10,5);
\coordinate (p6) at (11,6);
\coordinate (p7) at (13,6);
\coordinate (p8) at (14,5);
\coordinate (p9) at (15,3);
\coordinate (p10) at (17,2.5);
\coordinate (p11) at (19,4);
\coordinate (p12) at (21,5);
\coordinate (p13) at (23,5);
\coordinate (p14) at (25,4);
\coordinate (p15) at (27,2.5);

\redlajn{(b3)}{(p15)};
\redlajn{(p15)}{(p14)};
\redlajn{(p14)}{(p13)};
\redlajn{(p13)}{(p12)};
\redlajn{(p12)}{(p11)};
\redlajn{(p11)}{(p10)};
\redlajn{(p10)}{(b2)};
\redlajn{(b2)}{(p9)};
\redlajn{(p9)}{(p8)};
\redlajn{(p8)}{(p7)};
\redlajn{(p7)}{(p6)};
\redlajn{(p6)}{(p5)};
\redlajn{(p5)}{(p4)};
\redlajn{(p4)}{(b1)};
\redlajn{(b1)}{(p3)};
\redlajn{(p3)}{(p2)};
\redlajn{(p2)}{(p1)};
\redlajn{(p1)}{(b0)};

\ppunkt{(p1)};
\ppunkt{(p2)};
\ppunkt{(p3)};
\ppunkt{(p4)};
\ppunkt{(p5)};
\ppunkt{(p6)};
\ppunkt{(p7)};
\ppunkt{(p8)};
\ppunkt{(p9)};
\ppunkt{(p10)};
\ppunkt{(p11)};
\ppunkt{(p12)};
\ppunkt{(p13)};
\ppunkt{(p14)};
\ppunkt{(p15)};

\coordinate (t1) at (1,-3);
\coordinate (t2) at (2,-5);
\coordinate (t3) at (3,-6);
\coordinate (t4) at (5,-6);
\coordinate (t5) at (6,-5);
\coordinate (t6) at (7,-3);
\coordinate (t7) at (10,-2);
\coordinate (t8) at (12,-3);
\coordinate (t9) at (14,-2);
\coordinate (t10) at (18,-2);
\coordinate (t11) at (20,-3.5);
\coordinate (t12) at (22,-4);
\coordinate (t13) at (24,-3.5);
\coordinate (t14) at (26,-2);
\coordinate (t15) at (29.5,-2);

\grinlajn{(b4)}{(t15)};
\grinlajn{(t15)}{(b3)};
\grinlajn{(b3)}{(t14)};
\grinlajn{(t14)}{(t13)};
\grinlajn{(t13)}{(t12)};
\grinlajn{(t12)}{(t11)};
\grinlajn{(t11)}{(t10)};
\grinlajn{(t10)}{(b2)};
\grinlajn{(b2)}{(t9)};
\grinlajn{(t9)}{(t8)};
\grinlajn{(t8)}{(t7)};
\grinlajn{(t7)}{(b1)};
\grinlajn{(b1)}{(t5)};
\grinlajn{(t5)}{(t4)};
\grinlajn{(t4)}{(t3)};
\grinlajn{(t3)}{(t2)};
\grinlajn{(t2)}{(t1)};
\grinlajn{(t1)}{(b0)};

\tpunkt{(t1)};
\tpunkt{(t2)};
\tpunkt{(t3)};
\tpunkt{(t4)};
\tpunkt{(t5)};
\tpunkt{(t6)};
\tpunkt{(t7)};
\tpunkt{(t8)};
\tpunkt{(t9)};
\tpunkt{(t10)};
\tpunkt{(t11)};
\tpunkt{(t12)};
\tpunkt{(t13)};
\tpunkt{(t14)};
\tpunkt{(t15)};

\bpunkt{(b0)};
\bpunkt{(b1)};
\bpunkt{(b2)};
\bpunkt{(b3)};
\bpunkt{(b4)};

\draw[black, thick, fill=black ] ($(er)+(0,0.3)$) circle [radius=0.2]; 
\draw[black, thick, fill=black ] ($(el)-(0,0.3)$) circle [radius=0.2]; 

\node at ($(el)-(1,0.3)$) {$e_1$}; 
\node at ($(er)+(1,0.3)$) {$e_2$}; 
\node at ($(b0)+(1.4,0)$) {$b_0$};
\node at ($(b1)+(1.4,0)$) {$b_1$};
\node at ($(b2)+(1.4,0)$) {$b_2$};
\node at ($(b3)-(1,0)$) {$b_3$};
\node at ($(b4)+(1,0)$) {$b_4$};

\end{tikzpicture}